\documentclass[11pt]{article}

\usepackage[utf8]{inputenc}
\usepackage{lipsum}

\usepackage{amsmath, amssymb,amsfonts,amsthm,mathtools}
\usepackage{xspace,graphicx,relsize,bm,bbm,xcolor}
\usepackage{soul} % provides  s p a c i n g o u t, underlining and some derivatives such as overstriking and highlighting: https://ctan.mc1.root.project-creative.net/macros/generic/soul/soul.pdf

\usepackage{parskip}  % no indentation but some space between paragraphs

\usepackage{libertine}
\usepackage{libertinust1math}
\usepackage{dsfont}
\usepackage[T1]{fontenc}
\usepackage{enumitem}
\usepackage{nicefrac}
\usepackage{comment}

\usepackage[
    backend=biber,
    style=numeric-comp,
    sorting=none,
    minalphanames=3,
    maxalphanames=3,
    maxnames=99,
    backref=true
    ]{biblatex}
\DefineBibliographyStrings{english}{%
  backrefpage = {page},% originally "cited on page"
  backrefpages = {pages},% originally "cited on pages"
}

\usepackage{sepfootnotes}
\newendnotes{x}
\renewcommand\xnotesize\normalsize

\usepackage{hyperref}
\usepackage[margin=1.75cm]{geometry}
\definecolor{linkcol}{rgb}{0.0,0.55,0.7}
\definecolor{citecol}{rgb}{0.0, 0.6, 0.45}
\definecolor{urlcol}{rgb}{0.7, 0.0, 0.55}
\hypersetup{
	colorlinks,
	linkcolor={linkcol},
	citecolor={citecol},
	urlcolor={urlcol}
}

\usepackage{url}
\usepackage{subcaption}
\usepackage{mleftright}
\usepackage{hyperref}
\usepackage{multirow}
\usepackage{physics}  % all braket notation at once

\usepackage{algorithm}
\usepackage{algpseudocodex}[indLines = true,italicComments = false]

\usepackage{zref-clever}
\zcsetup{cap = true}
\newcommand{\Cref}{\zcref}

\usepackage{authblk}  % for Authors

\def\01{\{0,1\}}

\DeclareMathOperator*{\argmin}{arg\,min}

\let\Pr\relax
\DeclareMathOperator*{\Pr}{\mathbf{Pr}}

\newtheoremstyle{mydefinitionsty}% 〈name〉
{10pt}% 〈Space above〉
{10pt}% 〈Space below〉
{}% 〈Body font〉
{}% 〈Indent amount〉1
{}% 〈Theorem head font〉
{}% 〈Punctuation after theorem head〉
{.5em}% 〈Space after theorem head〉2
{\textbf{\thmname{#1}~\thmnumber{#2}:  }\thmnote{(#3)}}% 〈Theorem head spec (can be left empty, meaning ‘normal’)〉

\newtheoremstyle{myproblemsty}% 〈name〉
{10pt}% 〈Space above〉
{10pt}% 〈Space below〉
{}% 〈Body font〉
{}% 〈Indent amount〉1
{}% 〈Theorem head font〉
{}% 〈Punctuation after theorem head〉
{.5em}% 〈Space after theorem head〉2
{\textbf{\thmname{#1}~\thmnumber{#2}:  }\thmnote{(#3)}\newline}% 〈Theorem head spec (can be left empty, meaning ‘normal’)〉

\newtheoremstyle{mythmsty}% 〈name〉
{10pt}% 〈Space above〉
{10pt}% 〈Space below〉
{\itshape}% 〈Body font〉
{}% 〈Indent amount〉1
{}% 〈Theorem head font〉
{}% 〈Punctuation after theorem head〉
{.5em}% 〈Space after theorem head〉2
{\textbf{\thmname{#1}~\thmnumber{#2}:  }\thmnote{(#3)}}% 〈Theorem head spec (can be left empty, meaning ‘normal’)〉
\theoremstyle{mythmsty}
\newtheorem{theorem}{Theorem}[section]
\newtheorem{proposition}[theorem]{Proposition}
\newtheorem{lemma}[theorem]{Lemma}
\newtheorem{corollary}[theorem]{Corollary}

\newtheorem{question}[theorem]{Question}
\AddToHook{env/theorem/begin}{%
\zcsetup{countertype={theorem=theorem}}}
\AddToHook{env/proposition/begin}{%
\zcsetup{countertype={theorem=proposition}}}
\AddToHook{env/lemma/begin}{%
\zcsetup{countertype={theorem=lemma}}}
\AddToHook{env/corollary/begin}{%
\zcsetup{countertype={theorem=corollary}}}
\AddToHook{env/claim/begin}{%
\zcsetup{countertype={theorem=claim}}}
\AddToHook{env/conjecture/begin}{%
\zcsetup{countertype={theorem=conjecture}}}
\AddToHook{env/fact/begin}{%
\zcsetup{countertype={theorem=fact}}}
\zcRefTypeSetup{fact}{
Name-sg = Fact ,
name-sg = fact ,
Name-pl = Facts ,
name-pl = facts ,
}
\zcRefTypeSetup{assumption}{
  Name-sg = Assumption,
  name-sg = assumption,
  Name-pl = Assumptions,
  name-pl = assumptions,
}

\AddToHook{env/question/begin}{%
\zcsetup{countertype={theorem=question}}}

\theoremstyle{mydefinitionsty}
\newtheorem{definition}[theorem]{Definition}

\newtheorem{observation}[theorem]{Observation}
\newtheorem{example}[theorem]{Example}
\newtheorem{assumption}[theorem]{Assumption}
\newtheorem{remark}[theorem]{Remark}
\AddToHook{env/definition/begin}{%
\zcsetup{countertype={theorem=definition}}}
\AddToHook{env/notation/begin}{%
\zcsetup{countertype={theorem=notation}}}
\AddToHook{env/example/begin}{%
\zcsetup{countertype={theorem=example}}}
\AddToHook{env/assumption/begin}{%
\zcsetup{countertype={theorem=assumption}}}
\AddToHook{env/remark/begin}{%
\zcsetup{countertype={theorem=remark}}}

\theoremstyle{myproblemsty}

\AddToHook{env/problem/begin}{%
\zcsetup{countertype={theorem=problem}}}

\numberwithin{equation}{section}
\definecolor{alexcolor}{rgb}{0.0, 0.47, 0.75}   % {0.27, 0.51, 0.71}  %

\definecolor{questioncolor}{rgb}{0.36, 0.54, 0.66}

\title{Towards Surrogate Based Dequantization of Quantum Reinforcement Learning}

\author[1,2]{Pablo Rodriguez-Grasa\,\thanks{pablorogra@gmail.com }}
\author[3,4]{Sofiene Jerbi}
\author[1,5,6,7]{Mikel Sanz}
\author[8,9,10]{Ryan Sweke\,\thanks{rsweke@aims.ac.za}}
\affil[1]{Department of Physical Chemistry, University of the Basque Country UPV/EHU, Apartado 644, 48080 Bilbao, Spain}
\affil[2]{TECNALIA, Basque Research and Technology Alliance (BRTA), 48160 Derio, Spain}
\affil[3]{Dahlem Center for Complex Quantum Systems, Freie Universit\"{a}t Berlin, 14195 Berlin, Germany}
\affil[4]{Helmholtz-Zentrum Berlin f{\"u}r Materialien und Energie, 14109 Berlin, Germany}
\affil[5]{EHU Quantum Center, University of the Basque Country UPV/EHU, Apartado 644, 48080 Bilbao, Spain}
\affil[6]{IKERBASQUE, Basque Foundation for Science, Plaza Euskadi 5, 48009, Bilbao, Spain}
\affil[7]{Basque Center for Applied Mathematics (BCAM), Alameda de Mazarredo, 14, 48009 Bilbao, Spain}
\affil[8]{African Institute for Mathematical Sciences (AIMS), South Africa}
\affil[9]{Department of Mathematical Sciences, Stellenbosch University, Stellenbosch 7600, South Africa}
\affil[10]{National Institute for Theoretical and Computational Sciences (NITheCS), South Africa}

\date{\today}

\begin{document}
\maketitle

\begin{abstract}
    In recent years, the utility of parameterized quantum circuits as function approximators has been widely studied. In the context of reinforcement learning, this approach has led to variational quantum algorithms such as quantum Q-learning. While these methods show promising empirical results, and can provide provable advantages for artificial problems, it remains unclear whether they can provide a provable quantum advantage over classical approaches for problems of practical relevance. A natural way to investigate this question is through the lens of \textit{dequantization}: The construction of efficient classical algorithms capable of matching the performance of quantum variational methods. Building on recent kernel-based dequantization results for supervised learning, we take steps towards extending this surrogate-based dequantization program to reinforcement learning.
    Specifically, we study the simplified setting of reinforcement learning with a uniform generative model in which uniformly random state-action samples are available, which models the regime of sampling from a large experience replay buffer after sufficient exploration. Within this setting, we provide finite sample guarantees for classical kernelized Fitted Q-Iteration, with classical kernels designed to match the inductive bias of particular parameterized quantum circuits. Using these results, we then provide a set of sufficient conditions, on the data-encoding strategy of a parameterized quantum circuit, the corresponding classical kernel, and the problem structure, under which kernelized Fitted Q-Iteration provides a meaningful dequantization of quantum Q-learning, in this simplified setting. Apart from providing rigorous dequantization guarantees when these conditions are met, these results also motivate the use of kernelized fitted Q-iteration as a dequantization heuristic when these sufficient conditions cannot be verified.
\end{abstract}

\newpage

\tableofcontents
\newpage

\section{Introduction}
\label{sec:introduction}

Given rapid recent progress in both quantum computing and machine learning, the field of quantum machine learning (QML) has emerged as an extremely active field of research. Among the many proposals for quantum-enhanced machine learning, variational quantum algorithms, which optimize parameterized quantum circuits (PQCs) to solve learning tasks, have attracted particular attention, primarily as a result of their near-term implementability on noisy intermediate-scale quantum (NISQ) devices~\cite{Cerezo_2021, Benedetti_2019, Bharti2022noisy}. However, despite significant research, the question of whether and to which extent such variational quantum algorithms can gain meaningful advantages over classical approaches remains a central open question.

A particularly active line of research aimed at addressing this question for QML more broadly, is that of \textit{dequantization}, i.e. the construction of classical algorithms which can match the performance of their quantum counterparts with only polynomial overhead in computational resources \cite{Tang_2019,chia_2020,Tang_2021,landman2022classicallyapproximatingvariationalquantum,Sweke2025potential,sweke2025kernelbaseddequantizationvariationalqml,sahebi2025dequantizationsupervisedquantummachine,herrerogonzalez2025bornultimatumconditionsclassical,masotllima2025prospectsquantumadvantagemachine,deq_tns,gil-fuster2025on}. Dequantization results are significant for two complementary reasons: on one hand, they identify regimes in which quantum approaches offer no exponential advantage, thereby sharpening our understanding of where quantum speedups are and are not possible; on the other hand, they often yield new and independently interesting classical algorithms, inspired by the structure of the quantum models they are designed to simulate or surpass.

For \textit{variational} QML in particular, there are two broad approaches to dequantization. The first is the \textit{simulation-based} approach, which proceeds by developing classical simulation algorithms for directly simulating different classes of parameterized quantum circuits~\cite{gil-fuster2025on}. Indeed, there has been significant progress on this front via a broad array of techniques including tensor network methods and Pauli Path Propagation~\cite{McCaskey2018TNQVM,Pang2020Efficient2DTN,Guo2023DifferentiableMPS,Xu2024MPSVQE,Lykov2022StepDependentTN,Sander2025LocalTDVPCircuitSimulation,Watanabe2026TNSurrogatesVQC,Fontana2025NoisyVQC,Lerch2024LandscapePatches,Angrisani2025NoiselessObservables,Angrisani2026LocalNoisePauliPropagation,Rudolph2025PauliPropagationFramework}. The second is the \textit{surrogate-based} approach, in which one attempts to identify and optimize over a classical model class with an inductive bias which mirrors that of the PQC models~\cite{landman2022classicallyapproximatingvariationalquantum,Sweke2025potential,sweke2025kernelbaseddequantizationvariationalqml,sahebi2025dequantizationsupervisedquantummachine,herrerogonzalez2025bornultimatumconditionsclassical,masotllima2025prospectsquantumadvantagemachine}.

To date, the surrogate-based approach has been studied only in the context of supervised learning~\cite{landman2022classicallyapproximatingvariationalquantum,Sweke2025potential,sweke2025kernelbaseddequantizationvariationalqml,sahebi2025dequantizationsupervisedquantummachine} and generative modelling~\cite{herrerogonzalez2025bornultimatumconditionsclassical}. A common thread running through the works on supervised learning is the observation that PQC models are linear models with respect to a feature map defined by the data-encoding strategy of the circuit~\cite{Dataencoding}, and therefore lie within a reproducing kernel Hilbert space (RKHS) which can, under certain conditions, be optimized over efficiently using classical kernel methods. This observation has enabled the development of a surrogate-based approach to dequantization, in which one replaces PQC optimization with classical kernel ridge regression in the associated RKHS. This approach is particularly appealing, as one can often provide rigorous guarantees on the quality of the resulting solution relative to the best solution achievable by the quantum model~\cite{landman2022classicallyapproximatingvariationalquantum,Sweke2025potential,sweke2025kernelbaseddequantizationvariationalqml,sahebi2025dequantizationsupervisedquantummachine}.

In parallel, there has been growing interest in the application of quantum algorithms to reinforcement learning (RL)~\cite{Dunjko_2015, Dunjko_2016, Dunjko_2017, qpga_TQC, Hamann_2021, wang2021quantum, saggio_experimental_2021, wan_log_regrets, Hamann_2022, Buchholz2025multiarmedbandits, ambainis2025bitfreedomgoeslong, PQP_RL, Pgradients_VQC, qpg_action_decoding, QNPG, sofiene_PRX, vqc_for_deepRL, RL_with_QVC, qagents_gym, atari_hybrid_QCRL, QRL_hardware_errors, QRL_in_continuous, implementation_Lokes, chen2023quantumdeepqlearning, chen_deep_2024, freinberger2024quantumAtari, meyer2024surveyquantumreinforcementlearning}. Many early proposals assume quantum oracles for the environment and aim for provable speedups in query complexity~\cite{Dunjko_2015, Dunjko_2016, Dunjko_2017, qpga_TQC, Hamann_2021, wang2021quantum, saggio_experimental_2021, wan_log_regrets, Hamann_2022, Buchholz2025multiarmedbandits, ambainis2025bitfreedomgoeslong, doriguello2026improvedquantumalgorithmsreinforcement}. More recent NISQ-oriented works, however, focus on hybrid variational approaches that interact with classical environments~\cite{PQP_RL, Pgradients_VQC, qpg_action_decoding, QNPG, qpga_TQC, sofiene_PRX, vqc_for_deepRL, RL_with_QVC, qagents_gym, atari_hybrid_QCRL, QRL_hardware_errors, QRL_in_continuous, implementation_Lokes, chen2023quantumdeepqlearning, chen_deep_2024, freinberger2024quantumAtari, meyer2024surveyquantumreinforcementlearning}. In particular, quantum Q-learning replaces the classical neural network in deep Q-learning~\cite{human_level} with a PQC as a function approximator for the action-value function (or $Q$-function). While empirical studies have reported promising results in simulated environments, and provable advantages can be obtained in artificial environments~\cite{PQP_RL, qagents_gym}, the theoretical foundations of quantum Q-learning remain largely unexplored. In particular, to the best of our knowledge, no prior work has addressed whether quantum Q-learning can be dequantized or identified the conditions under which classical algorithms can provably match its performance.

We take steps towards addressing this gap. Specifically, in order to provide rigorous insights, we study the simplified setting of reinforcement learning with a \textit{uniform generative model} -- i.e. the setting in which uniformly random state-action samples are available, which models the setting of sampling from a large experience replay buffer after sufficient exploration.  Within this simplified setting, we propose and rigorously analyze a surrogate-based classical algorithm for reinforcement learning, and identify a set of explicit and interpretable sufficient conditions under which this algorithm provides an efficient dequantization of PQC based quantum Q-learning.

Our algorithm is based on Fitted Q-Iteration (FQI)~\cite{tree_based, FQI_bounds, theoretical_deepQ} with Kernel Ridge Regression (KRR), via a kernel chosen in such a way that its RKHS contains all functions expressible by the PQC model. Within the simplified setting of uniformly random state-action samples, we prove that whenever the sufficient conditions are met, this classical algorithm, using only polynomial time and polynomially many queries to the environment, outputs with high probability a policy which is $\epsilon$-close to the optimal policy achievable with quantum Q-learning. As such, in the limited setting of reinforcement learning with a uniform generative model, we prove that no exponential advantage is possible via quantum Q-learning, when the identified sufficient conditions are met.

In addition to technical assumptions on the structure of the problem, the main conditions we identify are (a) a sufficiently rapid decay of the eigenvalues of the chosen kernel (b) the alignment of the kernel spectrum with the optimal value functions encountered during learning and (c) the efficient implementability of the kernel. We note that, despite a significantly different analysis, these conditions are similar in spirit to the sufficient conditions identified in the context of dequantization of supervised learning~\cite{Sweke2025potential,sahebi2025dequantizationsupervisedquantummachine}, and recent work has shown how enforcing a tensor network structure on the kernel can ensure that condition (c) is always met~\cite{sweke2025kernelbaseddequantizationvariationalqml}.

We stress that the result discussed above is obtained as a corollary of rigorous finite sample guarantees for classical kernelized Fitted Q-Iteration with PQC-inspired kernels (see Theorem~\ref{thm:main-classical}), which may be of independent interest. Moreover, while our dequantization results are proven only in the simplified setting in which the learner has access to a generative model providing uniformly random state-action samples, we hope this work lays the theoretical groundwork for extending surrogate-based dequantization to reinforcement learning more broadly, and motivates future exploration of these algorithms as practical heuristics in standard settings.

The remainder of this paper is organized as follows. In Section~\ref{sec:dequantization-game} we introduce and discuss various notions of dequantization of variational QML, and describe in detail the surrogate-based approach adopted in this work. In Section~\ref{sec:rl-preliminaries} we provide the necessary background on reinforcement learning and Q-learning. In Section~\ref{sec:pqc-kernels} we introduce PQC models, describe their associated model classes and kernels, and establish the key structural properties that underpin our dequantization algorithm. In Section~\ref{sec:dequantization-algo} we then introduce the key simplifying assumption used in our work -- that of access to a uniform generative model -- and present our classical dequantization algorithm, namely kernelized Fitted Q-Iteration with PQC inspired kernels, in detail. With this in hand, we then discuss in Section~\ref{s:dequantization-question} how finite sample guarantees on kernelized Fitted Q-Iteration are sufficient to obtain dequantization results, before providing such guarantees in Section~\ref{sec:main-theorem}. Finally, we discuss in Section~\ref{ss:putting-together} the implications of these classical finite sample guarantees for dequantization of quantum Q-learning. 

\section{Dequantization of Variational Quantum Machine Learning}
\label{sec:dequantization-game}

Our goal in this work is to provide insights into the extent to which quantum approaches to reinforcement learning based on variational quantum circuits can be dequantized. However, in order to understand what this even means, we take a step back in this section, and consider the variety of different ways in which one can define ``dequantization'' of a variational quantum machine learning algorithm.  To this end, let's consider the following high-level scenario: Alice has a classical computer, Bob has a quantum computer, and both of them want to solve a machine learning problem $\mathcal{P}$. Importantly, we assume that both Alice and Bob start by agreeing on a precise definition of how to measure the ``quality'' of a hypothesis solution $h$. For example:
\begin{enumerate}
\item \textbf{Supervised learning:} If the problem is a supervised learning problem, then Alice and Bob might agree to use the \textit{empirical risk} of the hypothesis function with respect to a training dataset as a measure of solution quality. Alternatively, they might agree to use the \textit{true risk} as a measure of solution quality (as per the definition of PAC learning).
\item \textbf{Generative modeling:} If the problem is a generative modeling problem, then Alice and Bob might agree to use the Maximum Mean Discrepancy with respect to the training samples as a measure of solution quality. Alternatively, they might agree to use the total variation distance between the hypothesis distribution and the true distribution.
\item \textbf{Reinforcement learning:} Here, Alice and Bob might agree to use the distance between the hypothesis policy and the optimal policy as a measure of solution quality. 
\end{enumerate}
With this in mind, we can formulate a variety of notions of dequantization. To do this, let's use the notation $\mathcal{L}$ (for ``loss function'') to denote the function which measures the quality of a solution, and assume that $\mathcal{L}(h)\geq 0$ for any hypothesis $h$, and that the optimal quality value is zero -- i.e that hypothesis $h_a$ is better than hypothesis $h_b$ whenever $\mathcal{L}(h_a) < \mathcal{L}(h_b)$. Additionally, let's denote by $h_q(t,m)$ and $h_c(t,m)$ the random variables which are the output hypothesis of Bob's (quantum) and Alice's (classical) algorithm respectively, after running for time $t$ and using $m$ samples. The first notion we consider is what we will call a ``relative dequantization'':

\begin{definition}[Relative dequantization] We say that Alice has an $(f,g)$ relative dequantization of Bob's algorithm if she has a guarantee that, with high probability,
\begin{equation}
\mathcal{L}[h_c(f(t,m,\epsilon),g(t,m,\epsilon))] \leq  \mathcal{L}[h_q(t,m)] + \epsilon.
\end{equation}
\end{definition}
Intuitively, if Alice has a relative dequantization, then she knows exactly the overhead with respect to both time and samples that she needs to achieve a hypothesis which is at most $\epsilon$ worse than Bob's hypothesis (measured with respect to the agreed upon notion of solution quality). In other words, relative dequantization results provide an upper bound on the quantum advantage which can be obtained by Bob. One obstacle towards \textit{proving} such relative dequantization results is that in the context of variational QML, Bob's algorithm is almost always a \textit{heuristic} (such as gradient based optimization of a parameterized quantum circuit) and as such getting an explicit expression for the benchmark quality $\mathcal{L}[h_q(t,m)]$ can be extremely difficult.

Given the potential difficulty in quantifying the quality of Bob's solution after using a given number of time and samples, we could also ask what resources Alice needs to output a hypothesis comparable to the \textit{best possible} solution Bob could obtain. In particular, this can sometimes be easier to identify and quantify when Bob optimizes over a well characterized model class, such as all functions realizable by a given parameterized quantum circuit architecture. To formalize this, let's denote Bob's model class with $\mathcal{M}_q$ and define Bob's optimal solution via
\begin{equation}
h^*_q = \argmin_{h\in\mathcal{M}_q} \mathcal{L}[h].
\end{equation}
With this in hand, we formalize the notion of a ``relative-to-best'' dequantization as follows:

\begin{definition}[Relative-to-best dequantization] We say that Alice has an $(f,g)$ relative-to-best dequantization of Bob's algorithm if she has a guarantee that, with high probability
\begin{equation}
\mathcal{L}[h_c(f(d,\epsilon),g(d,\epsilon))]  \leq \mathcal{L}[h_q^*] + \epsilon,
\end{equation}
where $d$ is the size of the problem.
\end{definition}
Intuitively, when Alice has a relative-to-best dequantization, then she knows exactly the resources that she requires in order to output a solution which is no more than $\epsilon$ worse than the \textit{best possible} solution Bob could obtain. Here it's worth making a few observations: 

\begin{enumerate}
\item The notion of a relative-to-best dequantization is relatively strong, since in practice Bob may not come close to achieving the optimal hypothesis. The motivation for considering such a notion is indeed precisely that we do not want to (or cannot) quantify the performance of Bob's optimization algorithm. Instead, we want to exploit only properties of Bob's model class $\mathcal{M}_q$.
\item Relative-to-best dequantization results \textit{do not} rule out quantum advantages, but as before they allow us to place an \textit{upper bound} on the quantum advantages achievable by Bob. For example, if both $f = O(\mathrm{poly}(d,1/\epsilon))$ and $g = O(\mathrm{poly}(d,1/\epsilon))$ then Bob cannot achieve any exponential advantages (in either time or sample complexity) with respect to Alice.
\item We note that a relative-to-best dequantization result can be understood as a type of \textit{agnostic learning} result. Specifically, in the language of Ref~\cite{classicalverificationofquantumlearning}, a relative-to-best dequantization is a result for the agnostic learnability of Alice's model class with respect to Bob's model class as the benchmark class.
\end{enumerate}
As discussed in more detail in Section~\ref{ss:our-contribution}, we focus in this work on relative-to-best dequantization.

\subsection{Surrogation vs Simulation}\label{ss:surrogate-vs-simulate}

As mentioned in the introduction, there are a variety of ways one could go about trying to construct classical algorithms for which one can prove either relative or relative-to-best dequantization results. Here we highlight and contrast two main approaches, namely the surrogate-based approach and the simulation-based approach. See also Ref.~\cite{masotllima2025prospectsquantumadvantagemachine} for more details.

\subsubsection{Simulation-based dequantization}

In simulation based dequantization, one attempts to construct a classical algorithm which directly \textit{simulates} Bob's variational quantum algorithm. This approach is particularly well suited to proving relative dequantization results, as one can plausibly hope to quantify the resource overhead needed for simulating each iteration of Bob's algorithm up to a certain accuracy. The last few years have witnessed impressive progress in simulation based dequantization, using a wide variety of techniques for the simulation of quantum circuits such as tensor networks and Pauli Path Propagation~\cite{McCaskey2018TNQVM,Pang2020Efficient2DTN,Guo2023DifferentiableMPS,Xu2024MPSVQE,Lykov2022StepDependentTN,Sander2025LocalTDVPCircuitSimulation,Watanabe2026TNSurrogatesVQC,Fontana2025NoisyVQC,Lerch2024LandscapePatches,Angrisani2025NoiselessObservables,Angrisani2026LocalNoisePauliPropagation,Rudolph2025PauliPropagationFramework}.

\subsubsection{Surrogate-based dequantization}

In surrogate-based dequantization one \textit{does not} make any attempt to simulate Bob's variational quantum algorithm. Instead the idea of surrogate-based dequantization is to identify a \textit{surrogate model class} which shares the \textit{inductive bias} of Bob's model class, but for which one can optimize over using an efficient classical algorithm which admits provable guarantees. More specifically, one does the following:
\begin{enumerate}
\item Use knowledge of Bob's model class $\mathcal{M}_{q}$ to identify a closely related model class $\tilde{\mathcal{M}}_q$ which one can optimize over \textit{classically}.
\item Try to prove time and query complexity bounds for classical optimization over $\tilde{\mathcal{M}}_q$ and quantify the difference in quality of output solution with respect to the optimal solution in Bob's (closely related) model class $\mathcal{M}_q$.
\end{enumerate}
We note that this strategy is particularly well suited for proving relative-to-best dequantization results, as the dequantization algorithm only exploits knowledge of Bob's model class, and ignores completely the details of Bob's optimization algorithm. As we will describe in detail in Section~\ref{sec:pqc-kernels}, this strategy has been enabled by the precise characterization of the model classes associated with parameterized quantum circuits (PQCs)~\cite{Dataencoding}, together with the observation that all such models lie in a reproducing kernel Hilbert space which can be both succinctly described, and optimized over using kernel based algorithms, which admit rigorous guarantees~\cite{landman2022classicallyapproximatingvariationalquantum}. Indeed, precisely this approach has been used to provide surrogate based dequantization algorithms for PQC based variational regression~\cite{landman2022classicallyapproximatingvariationalquantum, Sweke2025potential, sweke2025kernelbaseddequantizationvariationalqml} as well PQC based variational classification and quantum kernel methods~\cite{sahebi2025dequantizationsupervisedquantummachine}.

\subsection{Our Contribution}\label{ss:our-contribution}
With the above established, we now have the language to state clearly, at a high level, the main contribution of this work from a dequantization perspective. Specifically:
\begin{center}
\textit{We provide a set of sufficient conditions for surrogate-based relative-to-best dequantization of a class of variational quantum algorithms for reinforcement learning, within the simplified setting of reinforcement learning with a uniform generative model.}
\end{center}
As already mentioned in Section~\ref{sec:introduction}, the result above follows from rigorous finite sample guarantees for our proposed classical surrogate-based algorithm, within this simplified setting, which may be of independent interest. In the following sections, we present the preliminaries necessary to make the above statement more precise.

\section{Reinforcement Learning Preliminaries}
\label{sec:rl-preliminaries}

Reinforcement learning (RL) studies sequential decision-making under uncertainty. An agent interacts with an environment over discrete time steps: at time $t$, it observes a state $S_t\in\mathcal{S}$, selects an action $A_t\in\mathcal{A}$, receives an immediate reward $R_t\in\mathbb{R}$, and the environment transitions to a new state $S_{t+1}$. The agent’s objective is to learn a strategy that maximizes long-term cumulative reward.

This interaction is commonly modeled as a discounted Markov Decision Process (MDP), defined by the tuple
\begin{equation}
(\mathcal{S},\mathcal{A},P,R,\gamma).
\end{equation}
Here, $\mathcal{S}\subseteq \mathbb{R}^{d_\mathcal{S}}$ denotes the state space and $\mathcal{A}$ is a finite action set. Without loss of generality, we assume that the state space is normalized such that $\mathcal{S} \subseteq [0,1]^{d_{\mathcal{S}}}$. Given an MDP defining a reinforcement learning problem, we refer to the tuple $(d_\mathcal{S},|\mathcal{A}|)$ as the ``size'' of the problem. In particular, we are interested in the scaling of reinforcement learning algorithms with respect to these parameters. Additionally, in this work we will assume, without loss of generality, that $\mathcal{A}\subseteq \{0,1\}^{d_{\mathcal{A}}}$, where $d_A\leq |\mathcal{A}|$ and depends on the encoding used (for example, a binary encoding has $d_A=\log_2(|\mathcal{A}|)$ and a one-hot encoding has $d_\mathcal{A}=|\mathcal{A}|$).  We denote by $\mathcal{P}(\mathcal{X})$ the set of probability distributions over a measurable space $\mathcal{X}$. The environment is characterized by two stochastic components: the transition kernel
\begin{equation}
P:\mathcal{S}\times\mathcal{A}\to\mathcal{P}(\mathcal{S}),
\end{equation}
which specifies the conditional distribution of the next state given the current state and action, and the reward distribution
\begin{equation}
R:\mathcal{S}\times\mathcal{A}\to\mathcal{P}(\mathbb{R}),
\end{equation}
which models the immediate feedback provided by the environment. We assume rewards are uniformly bounded: for all $(s,a)\in\mathcal{S}\times\mathcal{A}$ and $R\sim R(S,A)$, we have that $R\in [-R_{\max},R_{\max}]$ almost surely. The discount factor $\gamma\in(0,1)$ controls the relative importance of future rewards.

Given an initial state $S_0=s$, the interaction between the agent and the environment induces a trajectory $\{(S_t,A_t,R_t)\}_{t\ge 0}$. The quality of a trajectory at time $t$ is measured by the return 
\begin{equation}
G_{t}=\sum_{k=0}^{\infty}\gamma^k R_{t+k}.
\end{equation}
The fundamental goal of RL is to find a decision-making strategy that maximizes the expected return. To define this we need the notion of a policy, which is a mapping
\begin{equation}
\pi:\mathcal{S}\to\mathcal{P}(\mathcal{A}),
\end{equation}
assigning to each state a probability distribution over actions. Policies may be deterministic or stochastic. Under a policy $\pi$, the dynamics evolve according to
\begin{equation}
A_t\sim\pi(S_t),\qquad R_t\sim R(S_t,A_t),\qquad S_{t+1}\sim P(S_t,A_t).
\end{equation}
Once a policy $\pi$ is fixed, it induces a probability measure over trajectories through the interaction with the environment. For a given initial state $S_0=s$, the performance of the policy is quantified by the expected return $G_0$. This expected return is precisely the state-value function of the policy. Formally, the state-value function $V^\pi:\mathcal{S}\rightarrow \mathbb{R}$ is defined as 
\begin{equation}
V^{\pi}(s)=\mathbb{E}_{\pi}\!\left[G_0\,\middle|\,S_0=s\right]
=
\mathbb{E}_{\pi}\!\left[\sum_{t=0}^{\infty}\gamma^t R_t\,\middle|\,S_0=s\right].
\end{equation}
We also define the action-value function (or $Q$-function) $Q^\pi:\mathcal{S}\times\mathcal{A}\rightarrow \mathbb{R}$ of a policy $\pi$, which is defined via
\begin{equation}
Q^{\pi}(s,a)=\mathbb{E}_{\pi}\!\left[\sum_{t=0}^{\infty}\gamma^t R_t\,|S_0=s,\,A_0=a\right].
\end{equation}
These functions are related by 
\begin{equation}
V^{\pi}(s)=\mathbb{E}_{a\sim\pi(s)}\bigl[Q^{\pi}(s,a)\bigr],
\end{equation}
and satisfy the Bellman equation 
\begin{equation}
Q^{\pi}(s,a)=r(s,a)+\gamma\,\mathbb{E}_{s'\sim P(s,a)}\bigl[V^{\pi}(s')\bigr],
\end{equation}
where $r(s,a)=\mathbb{E}_{r\sim R(s,a)}[r]$. Notice that since rewards satisfy $|R| \leq R_{\max}$  and $\gamma \in (0,1)$, 
the return is uniformly bounded as
\begin{equation}
|G_t| 
\leq \sum_{k=0}^\infty \gamma^k R_{\max} 
= \frac{R_{\max}}{1-\gamma}.
\end{equation}
We therefore define
\begin{equation}
V_{\max} \equiv \frac{R_{\max}}{1-\gamma},
\end{equation}
and note that $\|Q^\pi\|_\infty \leq V_{\max}$ for any policy $\pi$. Define the Bellman operator $T^\pi:\mathbb{R}^{\mathcal{S}\times\mathcal{A}}\rightarrow \mathbb{R}^{\mathcal{S}\times\mathcal{A}}$ associated with $\pi$ by
\begin{equation}
(T^{\pi}Q)(s,a)=r(s,a)+\gamma\,\mathbb{E}_{s'\sim P(s,a),\,a'\sim\pi(s')}[Q(s',a')].
\end{equation}
Then $Q^{\pi}$ is the unique fixed point of $T^{\pi}$. The optimal action-value function is
\begin{equation}
Q^{*}(s,a)=\sup_{\pi}Q^{\pi}(s,a),\qquad \forall(s,a)\in\mathcal{S}\times\mathcal{A}.
\end{equation}
A policy $\pi^{*}$ is optimal if it is greedy with respect to $Q^{*}$, that is, it always selects an action that maximizes the optimal action-value function
\begin{equation}
\pi^{*}(s) \in \arg\max_{a \in \mathcal{A}} Q^{*}(s,a),
\qquad \forall s \in \mathcal{S}.
\end{equation}
The function $Q^{*}$ satisfies the Bellman optimality equation
\begin{equation}\label{eq:bellman-opt-1}
Q^{*}=TQ^{*},
\end{equation}
where 
\begin{equation}\label{eq:Bellman_optimality}
(TQ)(s,a)=r(s,a)+\gamma\,\mathbb{E}_{s'\sim P(s,a)}\Bigl[\max_{a'\in\mathcal{A}}Q(s',a')\Bigr].
\end{equation}
With the above established, we can now define a natural loss function to characterize the quality of a given action-value function $Q$. 
\begin{definition}[Loss of an action-value function]\label{def:loss-function} Given an MDP $(\mathcal{S},\mathcal{A},P,R,\gamma)$ with optimal action-value function $Q^*$, we define the loss $\mathcal{L}(Q)$ of any other action-value function $Q$ via
\begin{equation}
\mathcal{L}(Q) = \mathbb{E}_{(s,a)\sim U(\mathcal{S}\times\mathcal{A})}\left[|Q^*(s,a) - Q(s,a)|\right]
\end{equation}
where $U(\mathcal{S}\times \mathcal{A})$ denotes the uniform distribution over $\mathcal{S}\times\mathcal{A}$.
\end{definition}
We note that the loss function defined above satisfies the requirements discussed in Section~\ref{sec:dequantization-game}. Additionally, we have chosen to use the uniform distribution over states and actions to evaluate the learned action-value function, however in principle one also may be interested in alternative distributions given some prior knowledge of the environment. 

Finally, we note that broadly speaking, RL methods can be divided into policy-based approaches, which directly parameterize and optimize $\pi$, and value-based approaches, which estimate long-term values and derive a policy via greedy selection. Our analysis focuses on the latter, as described in the following section.

\subsection{(Parameterized) Q-learning}\label{ss:parameterized-q-learning}
One of the primary value-based approaches to RL is Q-learning. As we will be concerned with dequantizing modern quantum versions of this approach, we provide in this section an introduction to both classical and quantum (parameterized) Q-learning.

\subsubsection{Q-learning}
Q-learning is a foundational model-free reinforcement learning algorithm that aims to learn the optimal action-value function $Q^*(s,a)$ directly through interaction with the environment, without requiring a model of the transition dynamics or reward function. Introduced by Watkins in 1989 \cite{Watkins1989LearningFD}, Q-learning operates by iteratively updating estimates of the action-value function based on observed transitions and rewards, leveraging the Bellman optimality equation defined in Eqs.~\eqref{eq:bellman-opt-1} and~\eqref{eq:Bellman_optimality}.

In its basic form, Q-learning maintains a tabular representation of the $Q$-function, where each state-action pair $(s,a)$ is associated with a stored value $Q(s,a)$. The algorithm initializes these values arbitrarily and updates them using samples from the environment. Specifically, at each time step, the agent observes the current state $s$, selects an action $a$, receives a reward $r$, and transitions to the next state $s'$. The Q-value is then updated according to the rule
\begin{equation}
    Q(s,a)\xleftarrow{}Q(s,a)+\alpha\left[r+\gamma\max_{a'\in\mathcal{A}}Q(s',a')-Q(s,a)\right]
\end{equation}
where $\alpha\in(0,1]$ is the learning rate. This update reduces the temporal-difference error, which intuitively captures the mismatch between what the agent currently expects from taking action $a$ in state $s$ and what it newly observes: the immediate reward plus its own (discounted) estimate of the best possible future return from the next state. This process effectively performs a stochastic approximation of the Bellman optimality operator $T$ defined in Eq.~\eqref{eq:Bellman_optimality}. Under suitable conditions, such as visiting all state-action pairs infinitely often, Q-learning is guaranteed to converge to $Q^*$ in finite MDPs \cite{Watkins1992}.

\subsubsection{Deep Q-learning}\label{ss:deepqlearning}
The tabular approach becomes impractical for large state and action spaces. To address this, function approximation techniques were introduced to represent the $Q$-function \cite{qlearning_linear}. Deep Q-learning extends this idea by using deep neural networks (DNNs) as universal function approximators for $Q(s,a;\,\boldsymbol{\theta})$, where $\boldsymbol{\theta}$ denotes the network parameters.

A major practical breakthrough was achieved by \cite{human_level} with the Deep Q-Network (DQN) algorithm, which demonstrated that deep neural networks could be successfully trained to approximate the $Q$-function in high-dimensional environments. DQN trains the parameterized $Q$-function using stochastic gradient descent on a loss derived from the temporal-difference error. Its key innovations include experience replay, which stores transitions in a replay buffer and samples them uniformly to break temporal correlations and improve sample efficiency, and the use of a target network, a periodically updated copy of the main network that stabilizes the regression targets.

In deep Q-learning, the update rule is reformulated to minimize the mean squared error
\begin{equation}
    \mathcal{L}(\boldsymbol{\theta})=\mathbb{E}_{(s,a,r,s')\sim \mathcal{D}}\left[\left(r+\gamma \max_{a'}Q(s',a';\,\boldsymbol{\theta^+})-Q(s,a;\,\boldsymbol{\theta})\right)^2\right]
\end{equation}
where $\mathcal{D}$ is the replay buffer and $\boldsymbol{\theta}^+$ are the parameters of the target network. Deep Q-learning has achieved superhuman performance in complex tasks like Atari games and has become a cornerstone of deep reinforcement learning.

\subsubsection{Quantum Q-learning}\label{ss:qdeeplearning}
There is a rich history of work on quantum algorithms for reinforcement learning~\cite{Dunjko_2015, Dunjko_2016, Dunjko_2017, qpga_TQC, Hamann_2021, wang2021quantum, saggio_experimental_2021, wan_log_regrets, Hamann_2022, Buchholz2025multiarmedbandits, ambainis2025bitfreedomgoeslong, PQP_RL, Pgradients_VQC, qpg_action_decoding, QNPG, sofiene_PRX, vqc_for_deepRL, RL_with_QVC, qagents_gym, atari_hybrid_QCRL, QRL_hardware_errors, QRL_in_continuous, implementation_Lokes, chen2023quantumdeepqlearning, chen_deep_2024, freinberger2024quantumAtari, meyer2024surveyquantumreinforcementlearning}; we refer the reader to Ref.~\cite{meyer2024surveyquantumreinforcementlearning} for a comprehensive survey. Many early proposals assume quantum access to the environment and aim for provable speedups in query complexity~\cite{Dunjko_2015, Dunjko_2016, Dunjko_2017, qpga_TQC, Hamann_2021, wang2021quantum, saggio_experimental_2021, wan_log_regrets, Hamann_2022, Buchholz2025multiarmedbandits, ambainis2025bitfreedomgoeslong}. In this work, we are concerned uniquely with proposals for quantum reinforcement learning algorithms which interact with a standard classical environment. Within this setup, approaches can be broadly categorized into policy-based and value-based methods.

Policy-based methods use parameterized quantum circuits to directly approximate the policy~\cite{PQP_RL, Pgradients_VQC, qpg_action_decoding, QNPG, qpga_TQC}. Value-based approaches instead use parameterized quantum circuits to approximate the action-value function. The idea of using a parameterized quantum circuit as a Q-function approximator within a deep Q-learning framework, which we refer to as \textit{quantum Q-learning}, was originally proposed in Refs.~\cite{vqc_for_deepRL, RL_with_QVC}, and subsequently developed in Ref.~\cite{qagents_gym}, and further studied in Refs.~\cite{atari_hybrid_QCRL, QRL_hardware_errors, implementation_Lokes, chen2023quantumdeepqlearning, chen_deep_2024, freinberger2024quantumAtari}. We also note that Ref.~\cite{sofiene_PRX} proposes a distinct approach in which quantum subroutines are used to enhance action selection in value-based reinforcement learning with large action spaces, rather than to parameterize the Q-function via a PQC, and that Ref.~\cite{QRL_in_continuous} extends the variational approach to continuous action spaces via a quantum actor-critic method, employing parameterized quantum circuits for both the Q-function and the policy.

We focus in this work on quantum Q-learning. Quantum Q-learning builds on the parametrized Q-learning paradigm by replacing classical neural networks with parametrized quantum circuits to approximate the action-value function $Q(s,a;\,\boldsymbol{\theta})$. Training proceeds similarly to deep Q-learning, using gradient-based optimization (e.g., via the parameter-shift rule~\cite{Wierichs_2022}) on a temporal-difference loss, often incorporating experience replay and target circuits. While empirical results show promise in simulated environments, quantum Q-learning remains largely heuristic, with theoretical guarantees limited by the challenges of analyzing nonconvex optimization in quantum settings.

\section{PQC Model Classes and Associated Kernels}
\label{sec:pqc-kernels}

As we are concerned in this work with the dequantization of variational quantum algorithms (for reinforcement learning) based on the optimization of parameterized quantum circuits (PQCs) -- as explained at a high-level in Section~\ref{ss:qdeeplearning} -- we provide here a brief overview of the relevant PQC models, and the model classes associated with them. We also provide an overview of the connection between PQC model classes and reproducing kernel Hilbert spaces, which as described in Section~\ref{sec:dequantization-algo} is at the core of our classical dequantization algorithm.

\subsection{PQC Models}\label{ss:pqc-models}

In this work we are primarily concerned with parameterized quantum circuit (PQC) models for functions $f:\mathcal{S}\times \mathcal{A}\rightarrow\mathbb{R}$, where as explained in Section~\ref{sec:rl-preliminaries}, $\mathcal{S}\subset[0,1]^{d_{\mathcal{S}}}$ is the \textit{state space} and $\mathcal{A}\subseteq \{0,1\}^{d_\mathcal{A}}$ is some finite \textit{action space} of an MDP. More specifically, we consider models of the type
\begin{equation}\label{eq:PQC-models}
f_{\theta}(s,a) = \langle 0|U^{\dagger}(s,\theta)O_aU(s,\theta)|0\rangle,
\end{equation}
where $U(s,\theta)$ is some unitary parameterized by both states $s\in\mathcal{S}$ and trainable parameters $\theta\in\Theta$, and additionally $\{O_a\,|\,a\in\mathcal{A}\}$ is a set of observables indexed by actions $a\in\mathcal{A}$. In particular, such PQC models were originally proposed in Refs.~\cite{vqc_for_deepRL, RL_with_QVC}, subsequently developed in Ref.~\cite{qagents_gym}, and further studied in Refs.~\cite{atari_hybrid_QCRL, QRL_hardware_errors, implementation_Lokes, chen2023quantumdeepqlearning, chen_deep_2024, freinberger2024quantumAtari} as models for use in quantum Q-learning (as described in Section~\ref{ss:qdeeplearning}), and are the models whose use we are interested in dequantizing.

In order to exploit existing results on structural properties of PQCs~\cite{Dataencoding}, we will make a variety of assumptions. Firstly, we assume that there exists a parameterized unitary $V(a)$, parameterized by $a\in\{0,1\}^{d_\mathcal{A}}$, such that
\begin{align}
f_{\theta}(s,a) &= \langle 0|U^{\dagger}(s,\theta)O_aU(s,\theta)|0\rangle\label{eq:first-form} \\
&= \langle 0|U^{\dagger}(s,\theta)V^\dagger(a)OV(a)U(s,\theta)|0\rangle,\label{eq:second-form}
\end{align}
where $O$ is some fixed observable. Secondly, we will assume that all data is encoded via a \textit{Hamiltonian encoding} strategy. More specifically, we assume that:
\begin{enumerate}
\item Every gate in $U$ which depends on $s$ is of the form $U_{(j,k)}(s_j) = e^{-iH^{(j)}_ks_j}$. In other words, that every data-dependent gate in $U$ depends only on a single component $s_j$ for some $j\in [d_\mathcal{S}]$, and implements the time evolution of some Hamiltonian $H^{(j)}_k$ for time $s_j$.
\item Every gate in $V$ which depends on $a$ is of the form $V_{(j,k)}(a_j) = e^{-i\tilde{H}^{(j)}_k a_j}$ for some $j\in [d_\mathcal{A}]$.
\end{enumerate}
With this assumption, we denote by $\mathcal{D}^{(j)} = \{H^{(j)}_k\,|\,k\in [N_j]\}$ the set of all $N_j$ Hamiltonians used to encode the component $s_j$ at some point in the circuit, and we denote by $\tilde{\mathcal{D}}^{(j)} = \{\tilde{H}^{(j)}_k\,|\,k\in [\tilde{N}_j]\}$ the set of all $\tilde{N}_j$ Hamiltonians used to encode the component $a_j$ at some point in the circuit. Additionally, we define the \textit{data-encoding strategy} as the tuple
\begin{align}\label{eq:data-encoding-strategy}
\mathcal{D}&\coloneqq\left(\mathcal{D}^{(1)},\ldots, \mathcal{D}^{(d_\mathcal{S})},\tilde{\mathcal{D}}^{(1)},\ldots, \tilde{\mathcal{D}}^{(d_\mathcal{A})} \right) \\
&\coloneqq \left(\mathcal{D}^{(1)},\ldots, \mathcal{D}^{(d)}\right),
\end{align}
where we have defined
\begin{equation}
d = d_\mathcal{S} + d_\mathcal{A}.
\end{equation}
With this established, we then define the quantities
\begin{align}\label{eq:data-encoding-properties}
N_{\mathrm{max}} &\coloneqq \max_{j\in [d_{\mathcal{S}}+d_{\mathcal{A}}]} \left[N_j\right],\\
\lambda_\mathrm{max} &\coloneqq \max_{j\in [d_{\mathcal{S}}+d_{\mathcal{A}}]}\left[\max_{k\in [N_j]}\left[\|H_k^{(j)}\|_{\rm op}\right]\right].
\end{align}
Specifically, $N_\mathrm{max}$ is the maximum number of Hamiltonians used to encode any input component, and $\lambda_\mathrm{max}$ is the operator norm (equivalently, the largest absolute eigenvalue for the Hermitian encoding Hamiltonians considered here) of any encoding Hamiltonian.

We note that while these data-encoding assumptions do exclude some encoding strategies, such as those via time evolution of parameterized linear combinations of Hamiltonians, they cover a wide variety of common encoding strategies used in practice, and are not restrictive assumptions. Additionally, we show in Appendix~\ref{app:action-encoding} how some previously proposed models for Quantum Q-learning in the form of Eq.~\eqref{eq:first-form} can be brought into the form of Eq~\eqref{eq:second-form} via parameterized unitaries $V(a)$ with a suitable Hamiltonian data-encoding.

Under these assumptions, using the results of Ref.~\cite{Dataencoding} one can show that, for all $(s,a)\in\mathcal{S}\times\mathcal{A}$, one has
\begin{align}\label{eq:fourier_rep}
f_\theta(s,a) &= \sum_{\omega\in\tilde{\Omega}_{\mathcal{D}}} c_{\omega}(\theta) e^{i\langle\omega, z\rangle}, \\
&\equiv f_{\theta}(z),
\end{align}
where $z\equiv(s,a)\in \mathcal{S}\times\mathcal{A} \subset \mathbb{R}^{d}$. In particular, one has that: 
\begin{enumerate}
\item The set of frequency vectors $\tilde{\Omega}_{\mathcal{D}}\subseteq\mathbb{R}^{d}$ is completely determined by the data encoding strategy $\mathcal{D}$~\cite{Dataencoding,gil2020input,Caro_2021}. We describe exactly how $\tilde{\Omega}_{\mathcal{D}}$ is constructed from $\mathcal{D}$ in Appendix~\ref{app:data-encoding}.
\item The frequency coefficients $c_{\omega}(\theta)$ depend on all the gates in the circuit, but not on $z$, and are constrained in some way which usually does not admit a concise expression~\cite{mhiri2024constrainedvanishingexpressivityquantum}. 
\end{enumerate}
Additionally, we know that $\omega_0 \equiv (0,\ldots,0)\in\tilde{\Omega}_{\mathcal{D}}$, and that the non-zero frequencies in $\tilde{\Omega}_{\mathcal{D}}$ come in mirror pairs -- i.e. $\omega\in\tilde{\Omega}_{\mathcal{D}}$ implies $-\omega\in\tilde{\Omega}_{\mathcal{D}}$. Moreover, one can show that $c_\omega(\theta) = c_{-\omega}^*(\theta)$, which ensures that $f_\theta(x)$ is real. Finally, we note that given the structure of $\tilde{\Omega}_{\mathcal{D}}$, we can perform an arbitrary splitting of pairs to  redefine $\tilde{\Omega}_{\mathcal{D}} \coloneqq \Omega_{\mathcal{D}} \cup \left(-\Omega_{\mathcal{D}}\right)$, where $\Omega_{\mathcal{D}} \cap \left(-\Omega_{\mathcal{D}}\right) = \{\omega_0\}$. We define $M = | \Omega_{\mathcal{D}}\setminus \{\omega_0\} |$ which allows us to write $\Omega_{\mathcal{D}} = \{\omega_0,\omega_1,\ldots,\omega_{M}\}$ and $\tilde{\Omega}_{\mathcal{D}} = \{-\omega_M,\ldots,-\omega_1,\omega_0,\omega_1,\ldots,\omega_M\}$. 

With this established, given a PQC model in the form of Eq.~\eqref{eq:second-form}, specified by parameterization set $\Theta$, observable $O$ and parameterized unitaries $U,V$ with data-encoding strategy $\mathcal{D}$, we define the associated PQC model class -- i.e. the set of all functions that can be expressed by such a PQC model -- as
\begin{align}
\mathcal{F}_{\Theta, \mathcal{D},\mathcal{O}} &= \left\{f_\theta(s,a) = \langle 0|U^{\dagger}(s,\theta)V^\dagger(a)OV(a)U(s,\theta)|0\rangle \,|\,\theta\in \Theta\right\} \\
&= \left\{f_{\theta}(s,a) = \sum_{\omega\in\tilde{\Omega}_{\mathcal{D}}} c_{\omega}(\theta) e^{i\langle\omega, z\rangle}\,|\, z= (s,a)\,,\, \theta\in \Theta\right\}.
\end{align}
Additionally, we note that for all $f\in \mathcal{F}_{\Theta, \mathcal{D},\mathcal{O}}$ we have $\|f\|_\infty\leq \|O\|_\infty$.

\subsection{Linearity of PQC models}\label{ss:PQC_linear}

It is known that all such PQC models are \textit{linear models} with respect to a feature map defined by $\mathcal{D}$. To see this, we first define
\begin{align}
a_\omega(\theta) &= c_\omega(\theta) + c_{-\omega}(\theta),\\
b_\omega(\theta) &= i(c_\omega(\theta) - c_{-\omega}(\theta)).
\end{align}
Using this, together with the structure of $\tilde{\Omega}_{\mathcal{D}}$ discussed above, we can rewrite Eq.~\eqref{eq:fourier_rep} as
\begin{align}
f_\theta(z) &= c_{\omega_0}(\theta) + \sum_{i=1}^{M}\left(a_{\omega_i}(\theta)\cos(\langle\omega_i,z\rangle) + b_{\omega_i}(\theta)\sin(\langle\omega_i,z\rangle)\right)\\
&=\langle c(\theta),\phi_{\mathcal{D}}(z)\rangle,\label{eq:linear_map}
\end{align}
where we have defined 
\begin{align}
c(\theta) &= \left(c_{\omega_0}(\theta),a_{\omega_1}(\theta),b_{\omega_1}(\theta),\ldots,a_{\omega_M}(\theta),b_{\omega_M}(\theta) \right),\\
\phi_\mathcal{D}(z) &= \big(1,\cos(\langle\omega_1,z\rangle),\sin(\langle\omega_1,z\rangle),\ldots,\cos(\langle\omega_M,z\rangle),\sin(\langle\omega_M,z\rangle) \big).
\end{align}
As such, we see from Eq.~\eqref{eq:linear_map} that $f_\theta$ is indeed a linear function with respect to the feature map $\phi_\mathcal{D}:\mathbb{R}^{d}\rightarrow\mathbb{R}^{(2M + 1)}$. However, it is important to note that due to the constraints imposed by the circuit architecture, the PQC model may not contain \textit{all} linear functions defined by the feature map $\phi_\mathcal{D}$. To make this clear, let's define the set of \textit{all} linear functions with respect to $\phi_\mathcal{D}$ via 
\begin{equation}
\mathcal{F}_{\mathcal{D}} = \{f_v(\cdot) = \langle v,\phi_\mathcal{D}(\cdot)\rangle\,|\, v\in\mathbb{R}^{|\tilde{\Omega}_{\mathcal{D}}|}\}.
\end{equation}
We then immediately have 
\begin{equation}
\mathcal{F}_{\Theta,\mathcal{D},O} \subset \mathcal{F}_\mathcal{D}.
\end{equation}
The above inclusion is strict, given the fact that all functions in $\mathcal{F}_{\Theta,\mathcal{D},O}$ have bounded infinity norm, whereas functions in $\mathcal{F}_\mathcal{D}$ can be unbounded. However, if we define
\begin{equation}
\mathcal{F}_{\mathcal{D},O} = \{f\in \mathcal{F}_\mathcal{D} \,|\, \|f\|_\infty \leq \|O\|_\infty\}
\end{equation}
then indeed there could exist architectures for which $\mathcal{F}_{\Theta,\mathcal{D},O} = \mathcal{F}_{\mathcal{D},O}$. To summarize, we have
\begin{equation}
\mathcal{F}_{\Theta,\mathcal{D},O} \subseteq \mathcal{F}_{\mathcal{D},O} \subset \mathcal{F}_{\mathcal{D}}.
\end{equation}

\subsection{PQC-inspired kernels and the PQC-RKHS}\label{ss:pqc_inspired}

From any feature map $\phi:\mathbb{R}^{d_1} \rightarrow \mathbb{R}^{d_2}$ and positive semidefinite matrix $A\in\mathbb{R}^{d_2\times d_2}$, one can define a  kernel ${K_{(\phi,A)}:\mathbb{R}^{d_1}\times\mathbb{R}^{d_1}\rightarrow \mathbb{R}}$ via 
\begin{equation}
K_{(\phi,A)}(x,x') = \langle \phi(x),A\phi(x')\rangle.
\end{equation}
This follows from the fact that every positive semidefinite $A$ defines a valid inner product on the feature space $\mathbb{R}^{d_2}$. Given some parameterized quantum circuit with data encoding strategy $\mathcal{D}$, we call the kernel $K_{(\phi_{\mathcal{D}},A)}$ defined via the feature map $\phi_\mathcal{D}$ a \textit{PQC-inspired} kernel, for any positive semidefinite $A$. For simplicity, we denote this kernel with $K_{(\mathcal{D},A)}$. Additionally, to every kernel we can associate a reproducing kernel Hilbert space (RKHS) which one can roughly think of as the model class optimized over by any kernel based algorithm using the kernel. Importantly, every function in this class is linear with respect to the feature map defining this kernel. As this space will play a key role in our analysis, we denote the RKHS associated with the PQC-kernel $K_{(\mathcal{D},A)}$ as $\mathcal{H}_{{(\mathcal{D},A)}}$, and refer to this RKHS as the \textit{PQC-RKHS}.

With this in hand, we now define a special case, which will be of particular interest to us. To do this, let $w = (w_0,\ldots,w_M)\in\mathbb{R}^{(|\Omega_{\mathcal{D}}|)}$ be any vector with all components non-zero, and define $A_w\in \mathbb{R}^{2M+1,2M+1}$ via
\begin{equation}
A_w = \mathrm{diag}\left(\frac{1}{\|w\|_2}(w_0,w_1,w_1,\ldots,w_M,w_M)\right), 
\end{equation}
and $\tilde{A}_w = A_w^2$, which is positive definite by construction. We now note that
\begin{align}
K_{(\mathcal{D},\tilde{A}_w)}(z,z') &= \langle \phi_{\mathcal{D}}(z),\tilde{A}_w\phi_{\mathcal{D}}(z')\rangle\\
&=\frac{1}{\|w\|_2^2}\sum_{i=0}^M w_i^2 \left[\cos(\langle \omega_i,z-z'\rangle)\right],
\end{align}
and therefore $K_{(\mathcal{D},\tilde{A}_w)}$ is both \textit{normalized} (i.e. $K_{(\mathcal{D},\tilde{A}_w)}(z,z) = 1$) and \textit{shift-invariant} (i.e. depends only on $z-z'$). These are convenient properties, which motivate this special case. We refer to $w$ as a \textit{reweighting} of the feature map $\phi_\mathcal{D}$, and to the kernel $K_{(\mathcal{D},\tilde{A}_w)}$ as the PQC-inspired kernel reweighted by $w$. For simplicity we will denote $K_{(\mathcal{D},\tilde{A}_w)}$ with $K_{(\mathcal{D},w)}$, and the associated PQC-RKHS with $\mathcal{H}_{(\mathcal{D},w)}$. 

Additionally, note that if we define $\phi_{(\mathcal{D},w)}:\mathbb{R}^{d}\rightarrow \mathbb{R}^{(2M+1)}$ via:
\begin{align}
\phi_{(\mathcal{D},w)}(z) &\coloneqq \frac{1}{\|w\|_2}\big( w_0,w_1\cos(\langle \omega_1, z\rangle), w_1\sin(\langle \omega_1, z\rangle),\ldots,w_{M}\cos(\langle \omega_{M}, z\rangle),w_{M}\sin(\langle \omega_{M}, z\rangle)\big)
\end{align}
Then one has
\begin{align}
K_{(\mathcal{D},w)}(z,z') &= \langle \phi_{\mathcal{D}}(z),\tilde{A}_w\phi_{\mathcal{D}}(z')\rangle \\
&= \langle \phi_{(\mathcal{D},w)}(z),\phi_{(\mathcal{D},w)}(z')\rangle, \label{eq:feature-map-form}
\end{align}
which will be useful to us in later sections.

We stress at this stage however that it is a-priori not clear how $K_{(\mathcal{D},w)}$ can be implemented efficiently -- i.e. in \textit{time} polynomial with respect to the size of the problem $d_\mathcal{S}$ and $|\mathcal{A}|$. Indeed,  $\phi_{(\mathcal{D},w)}:\mathbb{R}^{d_\mathcal{S} + d_\mathcal{A}}\rightarrow \mathbb{R}^{2M+1}$, with $M=|\Omega_{\mathcal{D}}\setminus \{\omega_0\}|$, and as discussed in Appendix~\ref{app:data-encoding}, $\Omega_{\mathcal{D}}$ has a cartesian product structure which leads to $M$ scaling \textit{exponentially} with respect to $d$. As a result the naive approach of evaluating $K_{(\mathcal{D},w)}$ by first evaluating $\phi_{(\mathcal{D},w)}(z')$ and $\phi_{(\mathcal{D},w)}(z)$ then taking the inner product, will unfortunately \textit{not} be efficient.

To circumvent this obstacle, two approaches have been developed:
\begin{enumerate}
\item \textbf{Efficient approximation by Random Fourier Features (RFF):} In this approach, developed in the dequantization context in Refs.~\cite{landman2022classicallyapproximatingvariationalquantum,Sweke2025potential,sahebi2025dequantizationsupervisedquantummachine}, one \textit{approximates} the kernel $K_{(\mathcal{D},w)}$ by using a polynomial dimension feature map with features sampled randomly from $\Omega_{\mathcal{D}}$, via a distribution defined by a suitable normalization of $w$. In particular, this approach is efficient whenever the distribution defined by $w$ can be sampled from efficiently. However, as this approach \textit{approximates} the kernel, it introduces an additional approximation error which has to be suitably bounded. Indeed, the main contributions of Refs~\cite{Sweke2025potential,sahebi2025dequantizationsupervisedquantummachine} are precisely to provide such bounds, in the context of supervised learning problems.
\item \textbf{Efficient exact implementation via tensor network contractions:} In Ref~\cite{sweke2025kernelbaseddequantizationvariationalqml} it was shown that the cartesian product structure which causes the exponential scaling of $M$ can in fact sometimes be \textit{exploited} to facilitate \textit{exact} and efficient implementation of $K_{(\mathcal{D},w)}$ via tensor network contractions. More specifically, it was shown in Ref.~\cite{sweke2025kernelbaseddequantizationvariationalqml} that the feature map $\phi_{(\mathcal{D},w)}$ is in fact a bond-dimension 1 matrix product state (MPS). Therefore whenever $w$ is \textit{induced from a symmetric MPS} (see Ref.~\cite{sweke2025kernelbaseddequantizationvariationalqml} for a precise definition) with polynomial bond-dimension, the kernel $K_{(\mathcal{D},w)}$ can be evaluated exactly and efficiently via tensor network contractions. In other words, for such weight vectors, the kernel $K_{(\mathcal{D},w)}$ admits a \textit{kernel trick}. In this work, we will focus on this approach to avoid any approximation resulting from randomized approximations.
\end{enumerate}

Finally, as discussed in Observation 1 in Ref.~\cite{Sweke2025potential}, we note that by definition of the RKHS $\mathcal{H}_{(\mathcal{D},w)}$, for any $w$ with all entries non-zero, we have
\begin{equation}
\mathcal{F}_{\Theta,\mathcal{D},O} \subseteq \mathcal{F}_{\mathcal{D},O} \subset \mathcal{F}_{\mathcal{D}} \subseteq \mathcal{H}_{(\mathcal{D},w)}.
\end{equation}
Specifically, the RKHS $\mathcal{H}_{(\mathcal{D},w)}$ of the PQC inspired kernel $K_{(\mathcal{D},w)}$ contains \textit{all} functions expressible by any PQC model with data-encoding strategy $\mathcal{D}$. This observation is what motivates the use of $\mathcal{H}_{(\mathcal{D},w)}$ as a \textit{surrogate} model class for PQC optimization over $\mathcal{F}_{\Theta,\mathcal{D},O}$.

\subsection{Generalized PQC Models}
For transparency, we stress that the PQC models we have discussed in the previous sections are \textit{not} the most general one could consider. Importantly, one does not have to make some of the assumptions we have made in Section~\ref{ss:pqc-models}, and without these assumptions our analysis may not be immediately applicable. More specifically, in contrast to the PQC models we have discussed above, one could also consider:
\begin{enumerate}
\item \textbf{Trainable output scaling:} In Ref.~\cite{qagents_gym}, the Q-values are computed as $Q(s,a) = \langle 0|U_\theta^\dagger(s) O_a U_\theta(s)|0\rangle \cdot w_{o_a}$, where $w_{o_a}\in\mathbb{R}$ is a \emph{trainable} classical weight that we have excluded from consideration. We note however that our results will still hold for such models, by considering instead the PQC model $Q(s,a) = \langle 0|U_\theta^\dagger(s) \tilde{O}_a U_\theta(s)|0\rangle$ with $\tilde{O}_a = \max(w_{O_a})O_a$.
\item \textbf{Incompatible action observables:} Given a PQC model $Q(s,a) = \langle 0|U_\theta^\dagger(s) O_a U_\theta(s)|0\rangle$, with an observable set $\{O_a\,|\,a\in\mathcal{A}\}$, we have assumed that there exists some parameterized unitary $V(a)$ such that $O_a = V^\dagger(a)OV(a)$ for some fixed observable $a$. As discussed in Appendix~\ref{app:action-encoding} this \textit{is} the case in previous works on quantum Q-learning~\cite{qagents_gym}, but will not be satisfied whenever there exist two observables $O_{a}$ and $O_b$ with different spectra. When this condition is not satisfied, our analysis will not apply to the PQC model. 
\item \textbf{Trainable data-encoding}: In the above, we have only considered data-encoding gates which themselves \textit{do not} depend on trainable parameters. However, as adopted for quantum Q-learning in~\cite{qagents_gym} and further studied in Ref.~\cite{BenJaderberg}, one could also consider data-encoding gates which themselves depend on variational parameters -- e.g., by parameterizing the data-encoding Hamiltonians. As our entire analysis relies on the existence of a \textit{fixed} effective feature map defined via the data-encoding strategy, our analysis will not hold for models such as these in which the data-encoding strategy is not fixed.
\end{enumerate}
With the above caveats in mind, we note that the assumptions we have made on the PQC models do cover the vast majority of models proposed and used in practice, and we leave it as an interesting direction to expand the tools and results of this work to the more general models mentioned above.

\section{A Dequantization Algorithm for Quantum Q-learning}
\label{sec:dequantization-algo}

Given the preliminaries above, in this section we are finally ready to develop and present our proposed classical algorithm for the dequantization of quantum Q-learning. 

\subsection{The Idea}

Let's assume that Bob runs Quantum Q-learning, as described in Section~\ref{ss:qdeeplearning}, using a parameterized quantum circuit model of the form
\begin{align}
Q_{\theta}(s,a) &= \langle 0|U^{\dagger}(s,\theta)O_aU(s,\theta)|0\rangle,\\
&= \langle 0|U^{\dagger}(s,\theta)V^\dagger(a)O V(a) U(s,\theta)|0\rangle,
\end{align}
where both $U$ and $V$ use a \textit{Hamiltonian data-encoding strategy} $\mathcal{D}$, with associated frequency set $\Omega_{\mathcal{D}}$, as introduced in Section~\ref{sec:pqc-kernels}. Additionally, we assume that 
\begin{equation}
\|O\|_\infty = V_\mathrm{max} \equiv \frac{R_\mathrm{max}}{1-\gamma}
\end{equation}
to ensure that $Q_\theta$ has the correct range. Succinctly, in the notation of the previous section, we assume that Bob is optimizing over the model class $\mathcal{F}_{\Theta,\mathcal{D},O}$.

Using the observation from the previous section that, for any (strictly non-zero) $w\in\mathbb{R}^{|\Omega_{\mathcal{D}}|}$, the PQC-RKHS $\mathcal{H}_{(\mathcal{D},w)}$ contains all functions expressible via the PQC model $Q_{\theta}(s,a)$, the high level idea for Alice is to do the following:
\begin{center}
\textit{Implement classical Q-learning using the PQC-RKHS $\mathcal{H}_{(\mathcal{D},w)}$ of the associated PQC-kernel $K_{(\mathcal{D},w)}$ as a model class.}
\end{center}
More specifically, instead of updating the PQC model via iterations of stochastic gradient descent based updates, in each iteration Alice can use Kernel Ridge Regression with kernel $K_{(\mathcal{D},w)}$ to identify the optimal $Q$-function in the associated RKHS $\mathcal{H}_{(\mathcal{D},w)}$. This approach is motivated by the fact that both the PQC optimization and the kernel method are optimizing over linear functions with respect to the same feature map!

We stress that the idea above is simply the natural instantiation of \textit{kernel-based dequantization} for quantum Q-learning, and as such an instance of \textit{surrogate-based} dequantization. This approach has already been explored and analyzed for the dequantization of PQC-based  regression~\cite{landman2022classicallyapproximatingvariationalquantum,sweke2025kernelbaseddequantizationvariationalqml, Sweke2025potential} as well as both PQC-based classification and quantum kernel methods~\cite{sahebi2025dequantizationsupervisedquantummachine}.

\subsection{Some Problems}

In order to \textit{prove} that the above idea works, there are two obstacles one has to overcome:

\begin{enumerate}
\item \textbf{Efficiency:} As discussed in the previous section, given that $K_{(\mathcal{D},w)}$ is explicitly specified by a feature map into a potentially exponentially large feature space, it is not immediately obvious how to efficiently evaluate the kernel $K_{(\mathcal{D},w)}$, for an arbitrary choice of weighting $w$.
\item \textbf{Relative performance guarantees:} One has to provide a rigorous analysis of Alice's proposed algorithm. In particular, as discussed in Section~\ref{sec:dequantization-game}, one requires sample and time complexity bounds for learning $\epsilon$-optimal policies (at least, relative to Bob's optimal solution).
\end{enumerate}

As discussed, for the first obstacle above, there are luckily by now a variety of strategies one can use. In particular, the first approach, in the setting where the distribution defined by $w$ can be efficiently sampled, is to \textit{approximate} the kernel using Random Fourier Features~\cite{landman2022classicallyapproximatingvariationalquantum,Sweke2025potential,sahebi2025dequantizationsupervisedquantummachine}. The second approach is to use tensor network techniques, given the results of Ref~\cite{sweke2025kernelbaseddequantizationvariationalqml} which show that $K_{(\mathcal{D},w)}$ can be efficiently and exactly evaluated whenever $w$ is induced from a symmetric matrix product state of polynomial bond dimension. In this work, we will always assume that the weight vector has such a structure, so that evaluation of $K_{(\mathcal{D},w)}$ can be done exactly efficiently. 

Unfortunately, the second obstacle above requires more work to overcome. As Alice's proposed algorithm is essentially parameterized Q-learning with an RKHS model class, a natural place to look for inspiration for how to analyze Alice's algorithm would be existing analysis of parameterized Q-learning algorithms. Unfortunately however, parameterized Q-learning (with either DNNs or PQCs as a model class) is used exclusively as a \textit{heuristic}, and there are a wide variety of obstacles which make its rigorous analysis extremely challenging. For example: samples collected along trajectories generated by environment interactions are temporally correlated, violating independence assumptions needed for statistical analysis. Additionally, learning proceeds via stochastic optimization in a nonconvex landscape, complicating convergence guarantees. Finally, the update targets depend on the current $Q$-estimate, leading to bootstrapping instability. As mentioned in Section~\ref{ss:deepqlearning}, heuristics like experience replay and target networks mitigate these issues \textit{empirically}, but obscure the underlying statistical structure, making theoretical bounds difficult to derive. While not all of these obstacles would be present in the analysis of Alice's proposed algorithm -- for example the SGD based updates would be replaced with rigorous empirical risk minimization via Kernel Ridge Regression -- significantly new technical tools and techniques would be necessary to provide a rigorous analysis of Alice's high-level algorithm as proposed above.

\subsection{A Simplifying Assumption: Uniform Environment Samples}\label{ss:simplifying-assumption}

In order to overcome the analytical challenges mentioned above, we make the simplifying assumption that both Alice and Bob are able to obtain samples $(s,a)$ from the \textit{uniform distribution} over $\mathcal{S}\times\mathcal{A}$. Specifically, we assume the following:

\begin{assumption}[Sampling model]\label{ass:sampling-model}
We assume the ability to draw state-action tuples $(s,a)$ from the uniform distribution $U(\mathcal{S}\times\mathcal{A})$ over $\mathcal{S}\times\mathcal{A}$.
\end{assumption}

We stress, that this assumption is \textit{not} immediately satisfied in the standard reinforcement learning setting, where state-action samples are collected through \textit{exploration} within the environment. However, in many practical reinforcement learning algorithms (including deep Q-learning), the samples collected via exploration are often collected into a ``replay buffer'', which is then sampled from uniformly to provide a dataset for model updates (so called ``experience replay''). As such, in the regime of a large replay buffer, this training data can sometimes be reasonably assumed to arise from a \textit{fixed} distribution. Moreover, in the regime of sufficient exploration -- in which every region of the state-action space is visited with non-negligible frequency -- this sampling distribution converges toward the uniform distribution as the buffer size grows. As such, the assumption of access to uniform samples from $\mathcal{S}\times\mathcal{A}$ provides an idealization of sampling from a sufficiently large replay buffer, after sufficient exploration.

We note also that this simplifying assumption, referred to in the reinforcement learning literature as the assumption of \textit{access to a generative model}~\cite{azar2012samplecomplexityreinforcementlearning}, is made in many works aimed at providing rigorous analysis of RL algorithms, and in particular in prior work on theoretical analysis of deep Q-learning~\cite{theoretical_deepQ}. Indeed,  this assumption isolates the statistical components of the analysis and allows us to treat the problem as a sequence of standard supervised regression problems over samples drawn independently from $\mathrm{Unif}(\mathcal{S}\times\mathcal{A})$. Without it, the temporal correlations between samples collected along environment trajectories complicate the statistical analysis significantly. Of course, while any dequantization results obtained under the sampling model of Assumption~\ref{ass:sampling-model} do not immediately yield dequantization results in the standard setting, they can be used to provide preliminary evidence in this regard. Furthermore, the classical algorithms for which provable guarantees can be obtained under this assumption are then natural \textit{heuristic} candidates for dequantization in the standard setting. Indeed, this is the case with deep Q-learning, which can be rigorously analyzed under Assumption~\ref{ass:sampling-model}~\cite{theoretical_deepQ}, and is then used heuristically in the standard setting.

\subsection{Fitted Q-Iteration}
Under assumption~\ref{ass:sampling-model}, Fitted Q-Iteration (FQI) is a parameterized Q-learning algorithm which admits rigorous theoretical guarantees~\cite{tree_based, FQI_bounds, theoretical_deepQ}. Indeed, FQI was introduced precisely to provide rigorous theoretical evidence in favour of heuristic parameterized Q-learning algorithms such as deep Q-learning (via the simplifying assumption of state-action samples from a fixed distribution).

In FQI,  Q-learning is achieved by repeatedly solving supervised regression problems to approximate Bellman updates. More specifically, given a function class $\mathcal{F}$ and an estimate $\hat{Q}_k$ of the action-value function at iteration $k$, we sample i.i.d.\ state–action pairs $\{(s_i,a_i)\}_{i=1}^m$ from a distribution $\sigma$ (which here is assumed to be the uniform distribution). For each $i\in[m]$, let $r_i$ and $s_i'$ denote the immediate reward and the next state obtained after executing action $a_i$ in state $s_i$. We then compute the regression targets $y_i=r_i+\gamma\,\max_{a\in\mathcal{A}}\hat{Q}_k(s_i',a)$ and update 
\begin{equation}
    \hat{Q}_{k+1}=\arg\min_{f\in\mathcal{F}}\frac{1}{m}\sum_{i=1}^m\bigl(y_i-f(s_i,a_i)\bigr)^2.
\end{equation}
This cleanly separates sampling, approximation, and optimization, enabling provable guarantees, under the assumption of transition samples from a fixed distribution. In light of this, FQI provides a natural candidate dequantization algorithm for Alice, under Assumption~\ref{ass:sampling-model}.

\subsection{The Dequantization Algorithm}\label{ss:dequant-algo}
With the above established, we can finally present our proposed dequantization algorithm in detail.

Given Assumption~\ref{ass:sampling-model}, assume that Bob runs quantum Q-learning using a PQC model
\begin{equation}
Q_{\theta}(s,a) 
= \langle 0|U^{\dagger}(s,\theta)V^\dagger(a)O V(a) U(s,\theta)|0\rangle,
\end{equation}
with data-encoding strategy $\mathcal{D}$. Alice's high-level classical dequantization algorithm is then the following:
\begin{enumerate}
\item Choose some reweighting vector $w$ (as discussed above, here we assume that this is induced by a symmetric matrix product state, to ensure efficient evaluation of $K_{(\mathcal{D},w)}$~\cite{sweke2025kernelbaseddequantizationvariationalqml}).
\item Run Fitted Q-iteration (FQI) using the RKHS $\mathcal{H}_{(\mathcal{D},w)}$ as a model class. Specifically, use Kernel Ridge Regression via $K_{(\mathcal{D},w)}$ to solve the regression problem in each iteration.
\end{enumerate}
To be more precise, let $\mathcal{Z}=\mathcal{S}\times\mathcal{A}$ denote the joint state-action space. As per Assumption~\ref{ass:sampling-model} we assume access to i.i.d. samples from $U(\mathcal{S}\times\mathcal{A}$). At each iteration $k$, collect a batch of $m$ i.i.d.\ pairs $\{(s_i, a_i)\}_{i=1}^m \sim U(\mathcal{S}\times\mathcal{A})$, and for each $i$ observe the corresponding reward $r_i$ and next state $s_i'$. Using the action-value estimate $\hat{Q}_k$ output from the previous iteration, we then construct the Bellman regression targets
\begin{equation}
  y_i=r_i+\gamma\, \max_{a\in\mathcal{A}} \hat{Q}_k(s_i',a),\quad i=1,...,m,
\end{equation}
from which we can extract the regression dataset $\{(z_i,y_i)\}_{i=1}^m$, with $z_i=(s_i,a_i)$.

The next RKHS predictor is then defined as the empirical risk minimizer with $\ell_2$-regularization in $\mathcal{H}_{(\mathcal{D},w)}$:
\begin{equation}\label{eq:krr_step}
F_{k+1} 
= \arg\min_{f \in \mathcal{H}_{(\mathcal{D},w)}} 
\left\{ \sum_{i=1}^m \bigl( y_i - f(z_i) \bigr)^2 + \lambda^2 \|f\|_{\mathcal{H}_{(\mathcal{D},w)}}^2 \right\},
\end{equation}
where $\lambda > 0$ is the regularization parameter. We note that the loss $\mathcal{L}$ used to evaluate the resulting policy is defined via the $L^1$-norm (Definition~\ref{def:loss-function}), whereas Eq.~\eqref{eq:krr_step} solves an $L^2$-regularized regression subproblem at each iteration; the connection between per-iteration $L^2$ errors and the global $L^1$ policy loss is made precise by the error propagation result presented in Section~\ref{sec:main-theorem}. We can then solve this problem using Kernel Ridge Regression with PQC-inspired kernel $K_{(\mathcal{D},w)}$.  More specifically, by the representer theorem, the solution admits the explicit finite-dimensional form
\begin{equation}\label{eq:explicit_predictor}
F_{k+1}(z) 
= \mathbf{K}_m(z)^\top (K_m + \lambda^2 \mathds{1}_m)^{-1} \mathbf{y}_m,
\end{equation}
where  $\mathbf{K}_m(z) = \bigl(K_{(\mathcal{D},w)}(z,z_1),\dots,K_{(\mathcal{D},w)}(z,z_m)\bigr)^\top \in \mathbb{R}^m$ is the vector of kernel evaluations, $K_m \in \mathbb{R}^{m\times m}$ is the Gram matrix with $(K_m)_{ij} = K_{(\mathcal{D},w)}(z_i,z_j)$ and $\mathbf{y}_m = (y_1,\dots,y_m)^\top$ is the vector of Bellman targets.

Finally, in order to ensure that the output action-value function $\hat{Q}_{k+1}$ respects the maximum possible rewards, i.e. that $\|\hat{Q}_{k+1}\|_\infty\leq V_\mathrm{max}$, we define the pointwise clipping operator $\Pi_{[-V_{\max},V_{\max}]}$ via
\begin{equation}
\Pi_{[-V_{\max},V_{\max}]}(f)(z) 
\equiv \max\bigl(-V_{\max},\,\min\bigl(V_{\max},\,f(z)\bigr)\bigr),
\end{equation}
and define the output action-value function of the $k$'th iteration via $\hat{Q}_{k+1}=\Pi_{[-V_{\max},V_{\max}]}F_{k+1}$.

The above process is repeated $K$ times, and after $K$ iterations, we output the greedy policy
\begin{equation}
  \pi_K(s) \in \arg\max_{a \in \mathcal{A}} \hat{Q}_K(s,a).
\end{equation}
This is summarized in Algorithm~\ref{alg:dequantization}. The algorithm mirrors quantum Q-learning's structure but replaces optimization over the PQC model class with classical KRR in $\mathcal{H}_{(\mathcal{D},w)}$. Assuming  classical evaluation of  $K_{(\mathcal{D},w)}(\cdot,\cdot)$ it runs with purely classical resources. In the next sections, we proceed to rigorously analyze this algorithm, with the goal of providing sufficient conditions for this to provide an efficient relative-to-best dequantization of quantum Q-learning.

\begin{algorithm}[H]\label{alg:FQI-KRR-PQC}
\caption{Fitted Q-Iteration with PQC inspired kernel $K_{(\mathcal{D},w)}$}
\label{alg:dequantization}
\begin{algorithmic}[1]
    \State \textbf{Input:} Data encoding strategy $\mathcal{D}$, discount factor $\gamma$, iterations $K$, samples per iteration $m$, regularization $\lambda>0$
    \State \textbf{Choose} Reweighting vector $w$ defining the kernel $K_{(\mathcal{D},w)}$
  \State Initialize $F_0 \equiv 0$ (or any function in $\mathcal{H}_{(\mathcal{D},w)}$)
  \State $\hat{Q}_0 \gets \Pi_{[-V_{\max},\,V_{\max}]}F_0$
    \For{$k =0,\dots,K-1$}
        \State Sample $m$ i.i.d.\ state--action pairs $\{(s_i,a_i)\}_{i=1}^m \sim U(\mathcal{S}\times\mathcal{A})$
        \For{$i=1,\dots,m$}
            \State Observe reward $r_i$ and next state $s_i'$
      \State $y_i \gets r_i + \gamma \max_{a\in\mathcal{A}} \hat{Q}_k(s_i',a)$
        \EndFor
        \State Form dataset $\{(z_i,y_i)\}_{i=1}^m$ with $z_i=(s_i,a_i)$ 
        \State Compute Gram matrix $(K_m)_{ij} = K_{(\mathcal{D},w)}(z_i,z_j)$
    \State $F_{k+1} \gets \mathbf{K}_m(\cdot)^\top (K_m + \lambda^2 \mathds{1}_m)^{-1} \mathbf{y}_m$ \Comment{Kernel Ridge Regression} 
    \State $\hat{Q}_{k+1}\gets\Pi_{[-V_{\max},\,V_{\max}]}F_{k+1}$ \Comment{Action-value clipping}
    \EndFor
  \State \textbf{Output:} Greedy policy $\pi_K(s) \in \arg\max_{a\in\mathcal{A}} \hat{Q}_K(s,a)$ 
\end{algorithmic}
\end{algorithm}

\section{The Dequantization Question}\label{s:dequantization-question}
Given Algorithm~\ref{alg:dequantization}, we can finally phrase precisely the question that we would like to answer. However, to do this, it will be helpful to first introduce some notation. In particular, given an MDP $(\mathcal{S},\mathcal{A},P,R,\gamma)$ with optimal action value function $Q^*$ (as defined in Section~\ref{sec:rl-preliminaries}), as well as a PQC model
\begin{equation}
Q_{\theta}(s,a) 
= \langle 0|U^{\dagger}(s,\theta)V^\dagger(a)O V(a) U(s,\theta)|0\rangle,
\end{equation}
with $\theta\in\Theta$, data-encoding strategy $\mathcal{D}$ and observable $O$ satisfying $\|O\|_\infty =V_\mathrm{max}$, we define the \textit{optimal PQC model} as the function $Q_{\theta^*}$ with
\begin{align}
\theta^* = \argmin_{\theta\in\Theta}\left[\mathcal{L}(Q_\theta)\right] 
\end{align}
where $\mathcal{L}$ is as per Definition~\ref{def:loss-function}. Intuitively, $Q_{\theta^*}$ is the optimal action-value function that can be obtained by optimization over the PQC model class $\mathcal{F}_{\Theta,\mathcal{D},O}$. With this established, we would like to answer the following question:

\begin{question}[Sufficient conditions for relative-to-best dequantization of Quantum Q-learning via Algorithm~\ref{alg:dequantization}]\label{q:sufficient-conditions}
Given a PQC model
\begin{equation}
Q_{\theta}(s,a) 
= \langle 0|U^{\dagger}(s,\theta)V^\dagger(a)O V(a) U(s,\theta)|0\rangle,
\end{equation}
with $\theta\in\Theta$, data-encoding strategy $\mathcal{D}$ and observable $O$ satisfying $\|O\|_\infty =V_\mathrm{max}$, what are a set of sufficient conditions for Algorithm~\ref{alg:dequantization} to provide a time and sample efficient relative-to-best dequantization of Quantum Q-learning? More specifically, under what conditions can one guarantee that running Algorithm~\ref{alg:dequantization}, with a specifically chosen $w$, and with
\begin{align}
m = \mathcal{O}\left(\mathrm{poly}\left(d_{\mathcal{S}},|\mathcal{A}|, 1/\epsilon,1/\delta\right)\right),\\
K = \mathcal{O}\left(\mathrm{poly}\left(d_{\mathcal{S}}, |\mathcal{A}|, 1/\epsilon,1/\delta\right)\right),
\end{align}
runs time-efficiently and yields, with probability at least $1-\delta$, a policy $\pi_K$ satisfying
\begin{equation}
\mathcal{L}(Q^{\pi_K}) \leq \mathcal{L}(Q_{\theta^*}) + \epsilon.
\end{equation}
\end{question}
In particular, we recall from the discussion in Section~\ref{sec:rl-preliminaries} that the \textit{problem size} is given by the tuple $(d_\mathcal{S},|\mathcal{A}|)$, where $d_\mathcal{S}$ is the dimension of the state space and $|\mathcal{A}|$ is the number of actions. 

Towards answering Question~\ref{q:sufficient-conditions}, let's define the optimal action-value function achievable by optimizing over the PQC-RKHS $\mathcal{H}_{(\mathcal{D},w)}$ via 
\begin{equation}
Q^*_{(\mathcal{D},w)} = \argmin_{Q\in \mathcal{H}_{(\mathcal{D},w)}}\left[\mathcal{L}(Q)\right].
\end{equation}
With this defined, note that
\begin{equation}
\mathcal{F}_{\Theta,\mathcal{D},O} \subset \mathcal{H}_{(\mathcal{D},w)}
\end{equation}
immediately implies
\begin{equation}
\mathcal{L}(Q^*_{(\mathcal{D},w)}) \leq \mathcal{L}(Q_{\theta^*}).
\end{equation}
As a result
\begin{equation}
\mathcal{L}(Q^{\pi_K}) \leq \mathcal{L}(Q^*_{(\mathcal{D},w)}) + \epsilon \implies \mathcal{L}(Q^{\pi_K}) \leq \mathcal{L}(Q_{\theta^*}) + \epsilon,
\end{equation}
and therefore any sufficient conditions which ensure $\mathcal{L}(Q^{\pi_K}) \leq \mathcal{L}(Q^*_{(\mathcal{D},w)}) + \epsilon$ are immediately also sufficient to ensure $\mathcal{L}(Q^{\pi_K}) \leq \mathcal{L}(Q_{\theta^*}) + \epsilon$.

With this in mind, the question we answer in Section~\ref{sec:main-theorem} is the following:
\begin{question}[Sufficient conditions for Algorithm~\ref{alg:dequantization} to be an efficient PAC learner]\label{q:sufficient-conditions-2}
Given a  data-encoding strategy~$\mathcal{D}$, what are a set of sufficient conditions to guarantee that running Algorithm~\ref{alg:dequantization}, with a specifically chosen $w$, and with
\begin{align}
m = \mathcal{O}\left(\mathrm{poly}\left(d_{\mathcal{S}},|\mathcal{A}|, 1/\epsilon,1/\delta\right)\right),\\
K = \mathcal{O}\left(\mathrm{poly}\left(d_{\mathcal{S}},|\mathcal{A}|, 1/\epsilon,1/\delta\right)\right),
\end{align}
runs time-efficiently and yields, with probability at least $1-\delta$, a policy $\pi_K$ satisfying
\begin{equation}\label{eq:goal}
\mathcal{L}(Q^{\pi_K}) \leq \mathcal{L}(Q^*_{(\mathcal{D},w)}) + \epsilon.
\end{equation}
\end{question}
We stress that Question~\ref{q:sufficient-conditions-2} is a \textit{purely classical} question about the properties of FQI with kernel ridge regression via PQC-inspired kernels. However, as discussed, any sufficient conditions identified in answering Question~\ref{q:sufficient-conditions-2} immediately provide sufficient conditions for Algorithm~\ref{alg:dequantization} to provide an efficient relative-to-best dequantization of quantum Q-learning.

\section{Complexity Bounds for FQI with Kernel Ridge Regression and PQC Kernels}
\label{sec:main-theorem}
In this section we proceed to answer Question~\ref{q:sufficient-conditions-2}, by providing a rigorous analysis of Fitted Q-Iteration via Kernel Ridge Regression with PQC inspired kernels (Algorithm~\ref{alg:dequantization}). The section is structured as follows: 

\begin{enumerate}
\item Section~\ref{sub:assumptions}: Here we describe and motivate the assumptions that we will use in answering Question~\ref{q:sufficient-conditions-2}.
\item Section~\ref{ss:putting-together-classical}: We state the main classical result of this work -- Theorem~\ref{thm:main-classical} -- which provides an answer to Question~\ref{q:sufficient-conditions-2}. We sketch the proof, and defer the full proof to Appendix~\ref{app:proof-main-theorem}.
\end{enumerate}

\subsection{Assumptions}\label{sub:assumptions}

In addition to the foundational Assumption~\ref{ass:sampling-model} of uniform state-action samples, as discussed in Section~\ref{ss:simplifying-assumption}, we will require three other assumptions, all of which will then later appear as sufficient conditions for dequantization. Here we present and discuss these assumptions. Additionally, for each assumption, we provide (at least) two accompanying paragraphs. The first, labeled \textit{Justification}, argues why the assumption is plausible or natural in the context of quantum Q-learning. The second, labeled \textit{Role}, explains precisely why the assumption is needed for our analysis, and what would go wrong in its absence. This structure is intended to make the role of each assumption transparent, and to help identify which of them may be relaxed in future work.

\begin{assumption}[Bellman target realizability for clipped RKHS predictors]\label{as:func_classes}
For every RKHS predictor $F$ encountered by Algorithm~\ref{alg:dequantization}, define $Q=\Pi_{[-V_{\max},V_{\max}]}F$. Then the Bellman target
\begin{equation}
g_F \equiv TQ = T\!\left(\Pi_{[-V_{\max},V_{\max}]}F\right)
\end{equation}
belongs to $\mathcal{H}_{(\mathcal{D},w)}$.
\end{assumption}

\textbf{Justification:} This assumption requires that the Bellman targets encountered during the execution of Algorithm~\ref{alg:dequantization} remain within the hypothesis class $\mathcal{H}_{(\mathcal{D},w)}$. It is automatically satisfied whenever $\mathcal{H}_{(\mathcal{D},w)}$ is a \emph{universal} RKHS, since in that case all bounded continuous functions, and in particular all Bellman targets of bounded value functions, belong to the RKHS. For the finite-dimensional PQC-RKHS considered here, universality does not hold in general, and so the assumption is non-trivial: it requires that the specific frequency structure of the data-encoding strategy $\mathcal{D}$ is rich enough to represent the Bellman targets that arise during training. We acknowledge that this is a somewhat restrictive structural condition, and we sketch in Appendix~\ref{ap:relaxing_closure} how to relax it at the cost of an additional approximation error term in the final bound that we obtain with this assumption.

\textbf{Role:} This assumption allows us to exploit the reproducing property of the RKHS when analyzing the regression error at each iteration of Algorithm~\ref{alg:dequantization}. Specifically, by guaranteeing that the ideal regression target $g_k \equiv T\hat{Q}_{k-1}$ belongs to $\mathcal{H}_{(\mathcal{D},w)}$, it enables us to express the bias error of the kernel ridge regression predictor as an inner product in the RKHS, which in turn yields the pointwise error bounds that underpin Lemma~\ref{lem:sufficient_m_K} and, ultimately, Theorem~\ref{thm:main-classical}. Without this assumption, the regression targets would not be representable in the hypothesis class, introducing an irreducible approximation error that must be separately controlled (as discussed in Appendix~\ref{ap:relaxing_closure}).

\begin{remark} As discussed above, Assumption~\ref{as:func_classes} involves the predictors output at each iteration of Algorithm~\ref{alg:dequantization}, and as a result is not easy to verify a-priori. A stronger assumption, which implies the above, and is perhaps easier to verify (but may still be hard!), is the assumption that the \textit{clipped} RKHS is closed under the Bellman operator, i.e. that
\begin{equation}\label{eq:stronger}
T\!\left(\Pi_{[-V_{\max},V_{\max}]}(\mathcal{H}_{(\mathcal{D},w)})\right)\subseteq \mathcal{H}_{(\mathcal{D},w)}.
\end{equation}
In particular, this is the assumption that the image of the \textit{clipped} RKHS under the Bellman operator is contained in the RKHS, whereas technically for our analysis we only require this for the Bellman targets of those predictors encountered during optimization. However, we state this here to make clear that Equation~\ref{eq:stronger} is sufficient but not necessary for Assumption~\ref{as:func_classes}. As discussed in Section~\ref{ss:putting-together} we view it as an interesting and important problem to understand the extent to which these assumptions may be verified in practice.
\end{remark}

Before stating our next assumption, we note the following corollary of the stronger closure condition~\eqref{eq:stronger}:

\begin{corollary}[Bellman closure~\eqref{eq:stronger} implies realizability]\label{cor:realizability}
Under the closure condition~\eqref{eq:stronger}, the optimal action-value function satisfies $Q^*\in\mathcal{H}_{(\mathcal{D},w)}$.
\end{corollary}
\begin{proof}
Define the set $\mathcal{B}\equiv\{f\in\mathcal{H}_{(\mathcal{D},w)} : \|f\|_\infty\leq V_{\max}\}$. We show that $T$ maps $\mathcal{B}$ into itself. Take any $f\in\mathcal{B}$. Since $\|f\|_\infty\leq V_{\max}$, clipping acts as the identity: $\Pi_{[-V_{\max},V_{\max}]}f = f$, so $f\in\Pi_{[-V_{\max},V_{\max}]}(\mathcal{H}_{(\mathcal{D},w)})$. The closure condition~\eqref{eq:stronger} then gives $Tf\in\mathcal{H}_{(\mathcal{D},w)}$. Moreover, $\|Tf\|_\infty\leq R_{\max}+\gamma\|f\|_\infty\leq R_{\max}+\gamma V_{\max}= V_{\max}$, so $Tf\in\mathcal{B}$.

Now, $\mathcal{B}$ is a closed, bounded, convex subset of the finite-dimensional space $\mathcal{H}_{(\mathcal{D},w)}$, hence complete under the sup norm. Since the Bellman optimality operator $T$ is a $\gamma$-contraction on bounded functions under $\|\cdot\|_\infty$ and maps $\mathcal{B}$ into itself, the Banach fixed point theorem guarantees a unique fixed point of $T$ in $\mathcal{B}$. Since $Q^*$ is the unique fixed point of $T$ on the space of all bounded functions, we conclude $Q^*\in\mathcal{B}\subseteq\mathcal{H}_{(\mathcal{D},w)}$.
\end{proof}

In particular, an immediate consequence of Corollary~\ref{cor:realizability} is that $\mathcal{L}(Q^*_{(\mathcal{D},w)})=0$, since $Q^*$ itself belongs to $\mathcal{H}_{(\mathcal{D},w)}$ and achieves loss zero. Given this, we note that the condition
\begin{equation}
\mathcal{L}(Q^{\pi_K}) \leq \mathcal{L}(Q^*_{(\mathcal{D},w)}) + \epsilon.
\end{equation}
at the heart of Question~\ref{q:sufficient-conditions-2}, becomes 
\begin{equation}
\mathcal{L}(Q^{\pi_K}) \leq \epsilon.
\end{equation}
With this established, we now move on to our second main assumption

\begin{assumption}[Bounded transition density]
\label{ass:bounded-density}
Recalling that $P:\mathcal{S}\times\mathcal{A}\rightarrow \mathcal{P}(S)$ denotes the transition kernel of the MDP, we assume that there exists some finite $p_\mathrm{max}(d_\mathcal{S},|\mathcal{A}|) >0$ such that for every $(s,a)\in\mathcal{S}\times\mathcal{A}$ and every measurable set $B\subseteq\mathcal{S}$,
\begin{equation}
P(B\mid s,a)\le p_{\max}\,U(\mathcal{S})(B),
\end{equation}
where $U(\mathcal{S})$ denotes the uniform distribution over $\mathcal{S}$.
\end{assumption}

\textbf{Justification:}
The assumption requires each transition kernel to be uniformly dominated by the uniform measure on the state space. Equivalently, no measurable subset $B\subseteq\mathcal S$ can receive more than $p_{\max}$ times as much transition probability as its mass under the uniform measure. For a finite state space, this condition is automatically satisfied. Indeed, if $|\mathcal S|$ is finite, then $U(\mathcal S)(s')=\frac{1}{|\mathcal S|}$ for every $s'\in\mathcal S$. Since $P(s'\mid s,a)\le 1$, we have $P(s'\mid s,a)\le|\mathcal S|\,U(\mathcal S)(s')$, and, by summing over the elements of any set $B\subseteq\mathcal S$, we get $P(B\mid s,a)\le|\mathcal S|\,U(\mathcal S)(B)$. Thus one may always take $p_{\max}=|\mathcal S|$. In particular, the condition also allows deterministic transitions, which means that if $P(\cdot\mid s,a)$ is concentrated on a single state, the same choice $p_{\max}=|\mathcal S|$ suffices. For continuous state spaces, the condition is instead a regularity assumption, requiring the transition distribution to be uniformly controlled relative to the uniform measure on $\mathcal S$. This is standard in the analysis of fitted value iteration with off-policy sampling; see, e.g., Assumption~A1 of~\cite{FQI_bounds}.

\textbf{Role:} This assumption controls the mismatch between the distribution from which training data is sampled and the distribution under which the learned policy is evaluated. More technically, it enables us to bound the concentrability coefficient $\phi_{U,U}$ appearing in the FQI error propagation (Theorem~\ref{th:fqi_error_propagation}) underlying Theorem~\ref{thm:main-classical}, which governs how per-iteration regression errors amplify into suboptimality of the final policy. Under this assumption, $\phi_{U,U}$ grows at most as $\sqrt{|\mathcal{A}|\,p_{\max}}$, ensuring that the overall sample complexity in Theorem~\ref{thm:main-classical} remains polynomial in the problem size whenever $p_{\max}$ is itself polynomial.

With this established, we now state our final assumption, which concerns the choice of the reweighting vector $w$ defining the PQC-inspired kernel $K_{(\mathcal{D},w)}$. While the above two assumptions involve the structure of the \textit{environment}, something over which Alice does not have control, this assumption is on the weight vector $w$ chosen by Alice when running Algorithm~\ref{alg:dequantization}, and therefore can always be chosen to be satisfied. In particular, we require that the ordered weights decay at most inverse polynomially. More specifically:

\begin{assumption}[Inverse polynomial decay of kernel weights]
\label{ass:poly-eigendecay}
Assume the weight vector $w = (w_0,\dots,w_M)$ chosen in line 2 of Algorithm~\ref{alg:dequantization} has been permuted to be non-increasing, i.e., $w_0 \ge w_1 \ge \dots \ge w_M > 0$, and that there exist constants $C>0$ and $p>1$ such that
\begin{equation}
\frac{w^2_i}{\|w\|^2_2} \le C \, (i+1)^{-p} \qquad\text{for all } i=0,1,\dots,M.
\end{equation}
\end{assumption}

\begin{remark}\label{rem:alice-chooses} As already mentioned, we note that the weighting vector $w$ is \textit{chosen by Alice} when she chooses to run Algorithm~\ref{alg:dequantization}. As such, Alice can always ensure that Assumption~\ref{ass:poly-eigendecay} is satisfied by design.
\end{remark}

\textbf{Justification:} Unlike Assumptions~\ref{as:func_classes} and~\ref{ass:bounded-density}, which impose structural conditions on the environment, this assumption concerns the weight vector $w$ chosen by Alice. As noted in Remark~\ref{rem:alice-chooses}, Alice can always ensure that this condition is satisfied by design. For this reason, no separate justification of its plausibility as a property of the environment is needed.

\textbf{Role:} One can straightforwardly show that the elements of the weight vector $w$ are directly related to the eigenvalues of the kernel integral operator associated with the reweighted PQC kernel $K_{(\mathcal{D},w)}$~\cite{Sweke2025potential}. Additionally, it is known that these eigenvalues control the \textit{inductive} bias of the kernel. More specifically, the eigenvalue associated with an eigenfunction determines how rapidly the coefficient of that eigenfunction is resolved during kernel ridge regression. As such, this assumption ensures a decay of the kernel integral operator eigenvalues which is sufficient to ensure a polynomial query complexity of Algorithm~\ref{alg:dequantization}.

\begin{remark}[Assumptions as sufficient conditions]
Assumption~\ref{as:func_classes} and all other assumptions in this section are \emph{sufficient conditions for the theoretical analysis}, not necessary conditions for practical success. Both classical and quantum Q-learning routinely succeed as heuristics even when these assumptions are violated (e.g., when the hypothesis class is not closed under the Bellman operator). The assumptions serve to make the bounds rigorous and should not be interpreted as restrictions for algorithm design.
\end{remark}

\subsection{Complexity Bounds}\label{ss:putting-together-classical}

We are now in a position to state the main result of this work: Namely, a set of sufficient conditions under which Algorithm~\ref{alg:dequantization} provides a time and sample efficient PAC learner. This provides a concrete answer to Question~\ref{q:sufficient-conditions-2}, and therefore, as discussed in Section~\ref{s:dequantization-question}, has clear implications for the dequantization of quantum Q-learning, which we make explicit in Section~\ref{ss:putting-together}.

Before stating the result, we introduce a key quantity that appears in the theorem statement. Given the PQC-RKHS $\mathcal{H}_{(\mathcal{D},w)}$ associated with the kernel $K_{(\mathcal{D},w)}$, any function $f\in\mathcal{H}_{(\mathcal{D},w)}$ has an associated \textit{RKHS norm} $\|f\|_{\mathcal{H}_{(\mathcal{D},w)}}$.\label{ss:RKHS-norm} This norm can be understood as a measure of the \textit{misalignment} between the frequency spectrum of $f$ and the reweighting vector $w$ defining the kernel (or equivalently, the eigenvalues of the kernel integral operator). More specifically, when the significant frequency components of $f$ receive high weight under $w$, the RKHS norm remains small (often constant or polynomial in $d$), whereas when the function's energy is concentrated on low-weight frequencies, the norm can grow exponentially. As such, the quantity
\begin{equation}
G_K \;\equiv\; \max_{1 \le k \le K} \bigl\| T\hat{Q}_{k-1} \bigr\|_{\mathcal{H}_{(\mathcal{D},w)}}
\end{equation}
captures the worst-case misalignment of the Bellman targets encountered during training with the kernel reweighting. We refer to Ref~\cite{Sweke2025potential} for a more detailed discussion of this quantity, which plays a central role in the sample complexity of Algorithm~\ref{alg:dequantization}.

With this in mind, we have the following result:

\begin{theorem}[Sufficient conditions for Algorithm~\ref{alg:dequantization} to be an efficient PAC learner]\label{thm:main-classical}
Assume the following:
\begin{enumerate}
\item Assumption~\ref{ass:sampling-model} (state-action samples from the uniform distribution).
\item The closure condition~\eqref{eq:stronger} (image of clipped RKHS under Bellman operator contained in the RKHS) and Assumption~\ref{ass:bounded-density}, with 
\begin{equation}
p_{\max} = \mathcal{O}(\mathrm{poly}(d_\mathcal{S},|\mathcal{A}|)).
\end{equation}
\item Assumption~\ref{ass:poly-eigendecay} (polynomial decay of ordered weight vector $w$).
\item The data-encoding strategy $\mathcal{D}$ satisfies:
\begin{align}
N_\mathrm{max} &= 2^{O(\mathrm{poly}(d_\mathcal{S},|\mathcal{A}|))}\\
\lambda_\mathrm{max} &= 2^{O(\mathrm{poly}(d_\mathcal{S},|\mathcal{A}|))}
\end{align}
\item For all $k\in [K]$, the Bellman targets $T\hat{Q}_{k-1}$ are sufficiently well-aligned with $w$ such that
\begin{equation}\label{eq:condition1-thm}
G_K \;\equiv\; \max_{1 \le k \le K} \bigl\| T\hat{Q}_{k-1} \bigr\|_{\mathcal{H}_{(\mathcal{D},w)}} = \mathcal{O}(\mathrm{poly}(d_\mathcal{S},|\mathcal{A}|)).
\end{equation}
\end{enumerate}
Then
\begin{align}
m &= \tilde{\mathcal{O}}\!\left(\mathrm{poly}\!\left(d_\mathcal{S},|\mathcal{A}|,1/\epsilon,1/\delta\right)\right) \\
K &= \mathcal{O}(\log(1/\epsilon))
\end{align}
is sufficient for Algorithm~\ref{alg:dequantization} to yield, with probability at least $1-\delta$, a policy $\pi_K$ satisfying
\begin{equation}
\mathcal{L}(Q^{\pi_K}) \leq \epsilon = \mathcal{L}(Q^*_{(\mathcal{D},w)}) + \epsilon.
\end{equation}
Additionally, provided the weight vector $w$ is induced from a symmetric matrix product state with bond-dimension polynomial in $(d_\mathcal{S}, |\mathcal{A}|)$, the kernel $K_{(\mathcal{D},w)}$ can be evaluated efficiently, so that Algorithm~\ref{alg:dequantization} is both time and sample efficient.
\end{theorem}

\begin{remark}[Explicit complexity bounds]
The asymptotic bounds stated above are given for simplicity. The full non-asymptotic expressions for the sufficient sample size $m$ are significantly more involved, and are worked out in detail in Appendix~\ref{app:proof-main-theorem}.
\end{remark}

\begin{remark}[Efficient evaluation of $K_{(\mathcal{D},w)}$] As already discussed, we note that the time efficiency of evaluating $K_{(\mathcal{D},w)}$ when $w$ is induced by a symmetric matrix product state of polynomial bond dimension is the central result of Ref.~\cite{sweke2025kernelbaseddequantizationvariationalqml}.
\end{remark}

\textbf{Proof sketch.} The proof of Theorem~\ref{thm:main-classical} proceeds through the following steps, with full details provided in Appendix~\ref{app:proof-main-theorem}:

\textit{Step 1: Intermediate complexity bound (Appendix~\ref{ss:main-technical-lemma}).} We first establish Lemma~\ref{lem:sufficient_m_K}, which provides sufficient conditions on $m$ and $K$ for Algorithm~\ref{alg:dequantization} to output an $\epsilon$-optimal policy, expressed in terms of three abstract quantities: the maximum information gain $\Gamma(m)$ of the kernel, the maximum frequency norm $L_{\mathcal{D}} = \max_{\omega\in\Omega_{\mathcal{D}}}\|\omega\|_2$, and the RKHS norm $G_K$ of the Bellman targets. This lemma is proven by combining an existing FQI error propagation result (Theorem~\ref{th:fqi_error_propagation} adapted from~\cite{theoretical_deepQ}) with pointwise kernel ridge regression error bounds, a discretization argument, and a Markov inequality applied to the integrated predictive uncertainty. The proof then proceeds by analyzing and upper bounding these three abstract quantities.

\textit{Step 2: Bounding the maximum information gain (Appendix~\ref{ss:max-info-gain}).} We show that under Assumption~\ref{ass:poly-eigendecay} (polynomial weight decay), the maximum information gain satisfies $\Gamma(m) = \mathcal{O}((m\log^{p-1}m)^{1/p})$. This is achieved by decomposing the kernel into a finite-rank truncation and a tail, bounding each contribution separately (Lemma~\ref{lem:gamma_intermediate}), and then optimally choosing the truncation rank to yield the sublinear bound (Lemma~\ref{lem:gamma_poly}). Substituting this into the conditions of Lemma~\ref{lem:sufficient_m_K} yields Corollary~\ref{cor:poly_sample}, which gives an explicit (though involved) sufficient sample size.

\textit{Step 3: Bounding the frequency norm (Appendix~\ref{ss:Lk-term}).} We show that $L_{\mathcal{D}} \leq 2\sqrt{d}\,N_\mathrm{max}\lambda_\mathrm{max}$ (Lemma~\ref{lem:L_K-upperbound}), where $N_\mathrm{max}$ is the maximum number of encoding Hamiltonians per data component and $\lambda_\mathrm{max}$ is their largest eigenvalue. Under Condition 4 of the theorem, this ensures that the logarithmic dependence on $L_{\mathcal{D}}$ in the sample complexity remains polynomial.

\textit{Step 4: Time complexity (Appendix~\ref{ss:time-complexity}).} Following Ref.~\cite{sweke2025kernelbaseddequantizationvariationalqml}, we note that when $w$ is induced from a symmetric matrix product state with polynomial bond-dimension, the kernel $K_{(\mathcal{D},w)}$ can be evaluated exactly and efficiently via tensor network contractions, ensuring that each iteration of Algorithm~\ref{alg:dequantization} runs in polynomial time.

Combining these four steps with the alignment condition on $G_K$ (Condition 5) yields the stated polynomial sample and time complexity. \qed

\section{Sufficient Conditions for Dequantization of Quantum Q-learning (with Uniform Sampling)}\label{ss:putting-together}

We now return to the dequantization question introduced in Section~\ref{s:dequantization-question}. Recall from the discussion there that any sufficient conditions for Algorithm~\ref{alg:dequantization} to answer Question~\ref{q:sufficient-conditions-2} are immediately also sufficient to answer Question~\ref{q:sufficient-conditions}. Specifically, we showed that
\begin{equation}
\mathcal{F}_{\Theta,\mathcal{D},O} \subset \mathcal{H}_{(\mathcal{D},w)}
\end{equation}
implies
\begin{equation}
\mathcal{L}(Q^*_{(\mathcal{D},w)}) \leq \mathcal{L}(Q_{\theta^*}),
\end{equation}
and therefore
\begin{equation}
\mathcal{L}(Q^{\pi_K}) \leq \mathcal{L}(Q^*_{(\mathcal{D},w)}) + \epsilon \implies \mathcal{L}(Q^{\pi_K}) \leq \mathcal{L}(Q_{\theta^*}) + \epsilon.
\end{equation}
As a result, Theorem~\ref{thm:main-classical} immediately yields a set of sufficient conditions for Algorithm~\ref{alg:dequantization} to provide a time and sample efficient relative-to-best dequantization of Quantum Q-learning whenever uniform state-action samples are available, as per Question~\ref{q:sufficient-conditions}. Whenever these conditions hold, our result rules out any exponential quantum advantage via Quantum Q-learning, in the simplified setting in which uniform state-action samples are available.

We now discuss the significance and interpretation of these conditions. As Assumptions~\ref{ass:sampling-model} and~\ref{as:func_classes}  have been discussed in detail, and Condition 3 of Theorem~\ref{thm:main-classical} is a mild condition satisfied by typical data-encoding strategies, we focus on the role played by the reweighting vector $w$.

\begin{figure}[t]
\begin{center}
\includegraphics[width=\textwidth]{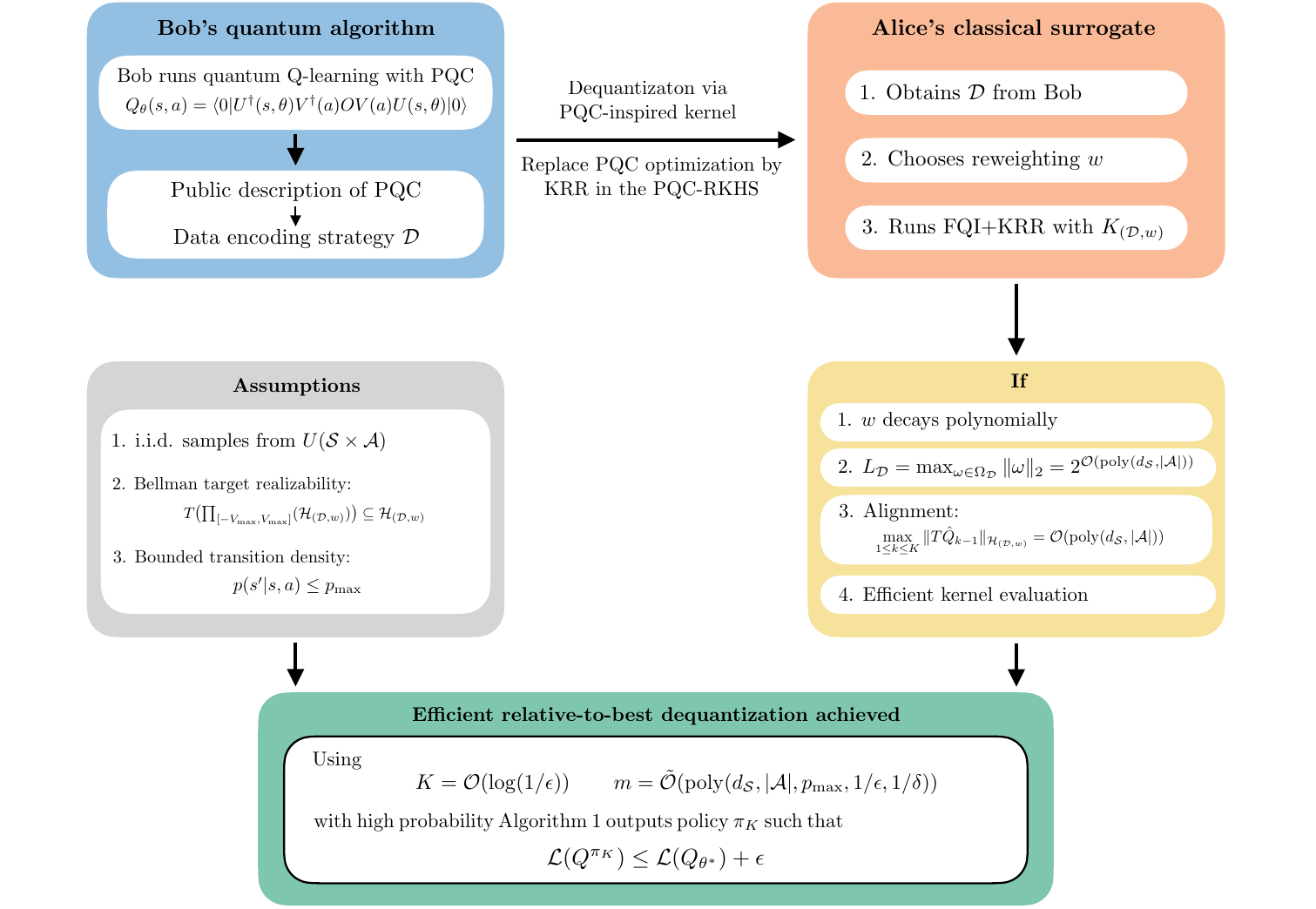}
\caption{A summary of the dequantization algorithm proposed and analyzed in this work. Specifically, we prove that if the assumptions shown above are met, and the reweighting vector $w$ and data encoding strategy $\mathcal{D}$ satisfies the identified sufficient conditions, then Algorithm~\ref{alg:dequantization} provides an efficient relative-to-best dequantization of quantum Q-learning, in the setting in which one assumes access to a uniform generative model.}
\label{fig:together}
\end{center}
\end{figure} 

More specifically, as noted in Observation~\ref{rem:alice-chooses}, recall that Alice is given the data encoding strategy $\mathcal{D}$, and must choose the reweighting vector $w \in \mathbb{R}^{|\Omega_{\mathcal{D}}|}$ before running Algorithm~\ref{alg:dequantization}.
Theorem~\ref{thm:main-classical} makes clear that this choice cannot be made arbitrarily: in order for the algorithm to be provably efficient, it must simultaneously satisfy three conditions, which we now discuss in turn.

\noindent \textbf{Polynomial eigenvalue decay (Assumption~\ref{ass:poly-eigendecay}).} The first requirement is that the components of $w$ decay polynomially when sorted in non-increasing order. As discussed in Section~\ref{ss:max-info-gain}, this controls the intrinsic complexity of the kernel $K_{(\mathcal{D},w)}$ via its Mercer eigenvalue decay, and is precisely what guarantees that the maximum information gain $\Gamma(m)$ grows sublinearly in $m$, thereby ensuring that the sample complexity does not scale exponentially with the problem size. 

\noindent\textbf{Efficient kernel evaluation via a tensor network structure.} The second requirement is that $w$ is induced from a symmetric matrix product state with bond-dimension polynomial in $(d_\mathcal{S}, |\mathcal{A}|)$. As established in Ref.~\cite{sweke2025kernelbaseddequantizationvariationalqml}, this algebraic structure on $w$ is precisely what allows $K_{(\mathcal{D},w)}$ to be evaluated exactly and efficiently via tensor network contractions, circumventing the otherwise exponential cost of naively computing the inner product in the full feature space. Without this, or an analogous condition enabling efficient Random Fourier Feature approximation~\cite{landman2022classicallyapproximatingvariationalquantum,Sweke2025potential,sahebi2025dequantizationsupervisedquantummachine}, each iteration of Algorithm~\ref{alg:dequantization} would be computationally intractable regardless of how small the sample complexity is. We stress that while one could in principle also use Random Fourier Feature techniques to efficiently \textit{approximate} the kernel, this would require handling an additional approximation error term.

\noindent\textbf{Alignment of $w$ with all regression functions (Condition~\eqref{eq:condition1-thm}} The third and most problem-dependent condition is that $w$ must be well-aligned with the frequency decomposition of \textit{all} the regression functions encountered during learning, in the sense discussed in Section~\ref{ss:RKHS-norm}. At a high level, this is the condition that the inductive bias of the kernel must be well suited to the problem being solved. We stress that the optimal regression functions are not known a priori -- indeed knowing them would amount to having solved the problem -- and so in practice it is not clear how Alice can ensure that her choice of $w$ will satisfy this condition. As such, while this is a succinct and rigorous sufficient condition, which can hopefully guide both the design of QML architectures and the identification of problems that may admit quantum advantage, in practice Alice would need to use heuristics to choose $w$. We view the design of such heuristics, informed by this condition, as an interesting open direction for further research. We stress however that there is a fundamental tension between the requirement of polynomial eigenvalue decay (Assumption~\ref{ass:poly-eigendecay}) and the alignment condition discussed here. In particular, the alignment is of course defined \textit{with respect to} the reweighting $w$, and therefore in order for a weighting vector to both decay polynomially \textit{and} be aligned with the optimal regression functions, it should be the case that the spectrum of these regression functions decay polynomially. We view the development of algorithms for \textit{testing} whether this is the case as an important open question which may be of independent interest.

\noindent \textbf{Summary.} Taken together, the three conditions on $w$ described above characterize the regime in which Alice can \textit{provably} efficiently dequantize Bob's Quantum Q-learning algorithm, in the simplified setting in which uniform state-action samples are available. The first and second conditions are structural requirements that Alice can satisfy by design, through an appropriate choice of $w$ with the correct decay profile and tensor network structure. The third condition, by contrast, is problem-dependent and cannot be evaluated efficiently before solving the problem, but makes rigorous the intuition that dequantization is possible when the problem structure is well matched to the inductive bias of the classical surrogate. When all three conditions are met, Theorem~\ref{thm:main-classical} guarantees that Algorithm~\ref{alg:dequantization} solves the reinforcement learning problem to $\epsilon$-accuracy with polynomial time and sample complexity, thereby ruling out any exponential quantum advantage via Quantum Q-learning in this simplified setting. See Figure~\ref{fig:together} for an overview of these conditions and their implications.

\begin{remark}[Algorithm~\ref{alg:dequantization} as a \textit{heuristic}] As per the discussion above, Theorem~\ref{thm:main-classical} makes clear that under certain conditions, Algorithm~\ref{alg:dequantization} provides a \textit{provably} efficient dequantization of quantum Q-learning. While we do not provide any guarantees when these conditions are not met, we believe that our results provide motivation for exploring the performance of Algorithm~\ref{alg:dequantization} as a \textit{heuristic}, even when these conditions are not met, or cannot be verified.
\end{remark}

\begin{remark}[On the relation to Refs~\cite{Sweke2025potential,sweke2025kernelbaseddequantizationvariationalqml,sahebi2025dequantizationsupervisedquantummachine}]
We note that the conditions on $w$ discussed above are very similar to the sufficient conditions identified for kernel-based dequantization of supervised learning in Refs.~\cite{Sweke2025potential,sweke2025kernelbaseddequantizationvariationalqml,sahebi2025dequantizationsupervisedquantummachine}, despite the significantly different analysis required in the reinforcement learning setting. Fundamentally, this similarity arises because of the key role that the spectrum of a kernel (its Mercer eigenvalues) plays in the analysis of kernel algorithms, and the tight relation between the reweighting vector $w$ and the eigenvalues of the PQC kernel $K_{(\mathcal{D},w)}$. We stress however, that unlike the case of supervised regression, where the alignment condition applies to a \textit{single} fixed target function, in reinforcement learning the regression target at step~$k$ is $g_k = T\hat{Q}_{k-1}$, which depends on the previous iterate and therefore changes as the algorithm progresses. The condition $G_K = \max_k \|g_k\|_{\mathcal{H}_{(\mathcal{D},w)}} = O(\mathrm{poly})$ must therefore hold \emph{uniformly} over all $K$ iterations, constraining not just the optimal $Q$-function but the entire trajectory of Bellman iterates. This makes the requirement strictly stronger than the analogous condition in supervised dequantization~\cite{Sweke2025potential,sahebi2025dequantizationsupervisedquantummachine}.
\end{remark}

\section*{Acknowledgements}

RS thanks the Alexander von Humboldt foundation for their support, under the German Research Chair program at the African Institutes for Mathematical Sciences. MS and PRG thank support from the Basque Government BasQ initiative under the Q-STREAM project. They also acknowledge support from OpenSuperQ+100 (Grant No. 101113946) of the EU Flagship on Quantum Technologies, from Project Grant No. PID2024-156808NB-I00 and Spanish Ram\'on y Cajal Grant No. RYC-2020-030503-I funded by MICIU/AEI/10.13039/501100011033 and by “ERDF A way of making Europe” and “ERDF Invest in your Future”, by the EU through the Recovery, Transformation and Resilience Plan–Next Generation EU within the framework of the Digital Spain 2026 Agenda, and by the Basque Government through Grant No. IT1887-26.

\section*{AI Disclosure} The first version of this manuscript, containing all the essential ideas, constructions, proofs and method of presentation, was obtained without the use of generative AI. Once we had a first draft, Claude Opus 4.8 (Max) was then used to generate a detailed and critical referee report, with a focus on verifying mathematical correctness. This referee report flagged a variety of mathematical issues, which were then fixed via a process of interaction with the model. This process materially effected details in all proofs, however was predominantly focused on the proof of Lemma~\ref{lem:sufficient_m_K}.

\newpage
\appendix 

\section{On encoding of actions into PQC models}\label{app:action-encoding}

As discussed in Section~\ref{sec:pqc-kernels}, previous works on quantum Q-learning, such as Ref.~\cite{qagents_gym}, have used PQC models of the form
\begin{equation}\label{eq:first-form-app}
Q_\theta(s,a)=\langle 0| U^\dagger(\theta,s) O_a U(\theta,s) |0\rangle,
\end{equation}
for some set of observables $\{O_a\}_{a\in\mathcal{A}}$. In this work however, we have assumed that the PQC model is of the form
\begin{equation}\label{eq:second-form-app}
Q_\theta(s,a)=\langle 0| U^\dagger(\theta,s)V^\dagger(a) OV(a) U(\theta,s) |0\rangle,
\end{equation}
where $V(a)$ uses only Hamiltonian data-encodings. In this appendix we show that this assumption is not unreasonable, in the sense that previously used models in the form of Eq.~\eqref{eq:first-form-app} can indeed be written in the form of Eq.~\eqref{eq:second-form-app}, with unitaries $V(a)$ using a Hamiltonian data-encoding strategy. We illustrate this with two concrete examples from Ref.~\cite{qagents_gym}.

\begin{example}[Example 1: Frozen Lake Model (Ref.~\cite{qagents_gym})]In the \emph{Frozen Lake} environment there are four actions $(1,2,3,4)$, and Ref.~\cite{qagents_gym} uses a four-qubit PQC model in the form of Eq.~\eqref{eq:first-form-app} with action-dependent observables $(O_1,O_2,O_3,O_4) = \{Z_1, Z_2, Z_3, Z_4\}$. Here we show two different ways to write this model in the form of Eq.~\eqref{eq:second-form-app}, by constructing a suitable action-parameterized unitary $V(a)$.

To do this, first define
\begin{equation}
H_{i,j} = \frac{1}{4}\left(\mathds{1} - X_iX_j - Y_iY_j - Z_iZ_j\right)
\end{equation}
so that
\begin{equation}
\mathrm{SWAP}_{i,j} = e^{-i\pi H_{i,j}}.
\end{equation}
With this, we then have the following:

\textbf{Method 1 (one-hot action encoding):} Let's use a one-hot encoding to denote the actions -- i.e. 
\begin{equation}
\mathcal{A} = (1,2,3,4) = \big((1,0,0,0),(0,1,0,0),(0,0,1,0),(0,0,0,1)\big)\subset{\{0,1\}^4}.
\end{equation}
Additionally, let's define
\begin{align}
H^{(2)} &= H_{1,2},\\
H^{(3)} &= H_{1,3},\\
H^{(4)} &= H_{1,4},
\end{align}
and
\begin{equation}
V(a) = V(a_1,a_2,a_3,a_4) = e^{-i\pi a_2 H^{(2)}}e^{-i\pi a_3H^{(3)}}e^{-i\pi a_4H^{(4)}}.
\end{equation}
One then has
\begin{align}
V(1) &= \mathds{1} \\
V(2) &= \mathrm{SWAP}_{1,2} \\
V(3) &= \mathrm{SWAP}_{1,3} \\
V(4) &= \mathrm{SWAP}_{1,4} 
\end{align}
so that
\begin{equation}
O_a = V^\dagger(a)Z_1V(a)
\end{equation}
for all $a\in\mathcal{A}$ as desired. 

\textbf{Method 2 (Binary action encoding):} Ideally one wants to to use an encoding of the actions into $\{0,1\}^{d_\mathcal{A}}$ with $d_\mathcal{A}$ as small as possible. To this end, we show that the method above can be improved by using a binary encoding of the actions. More specifically, if we define
\begin{equation}
\mathcal{A} = (1,2,3,4) = \big((0,0),(0,1),(1,0),(1,1)\big)\subset{\{0,1\}^4}
\end{equation}
as well as
\begin{align}
H^{(1)} &= H_{1,3} + H_{2,4},\\
H^{(2)} &= H_{1,2},
\end{align}
\begin{equation}
V(a) = V(a_1,a_2) =  e^{-i\pi a_1 H^{(1)}}e^{-i\pi a_2H^{(2)}},
\end{equation}
then
\begin{align}
V(1) &= \mathds{1} \\
V(2) &= \mathrm{SWAP}_{1,2} \\
V(3) &= \mathrm{SWAP}_{1,3}\mathrm{SWAP}_{2,4} \\
V(4) &= \mathrm{SWAP}_{1,3}\mathrm{SWAP}_{2,4}\mathrm{SWAP}_{1,2} 
\end{align}
As a result, we then again have
\begin{equation}
O_a = V^\dagger(a)Z_1V(a)
\end{equation}
for all $a\in\mathcal{A}$ as desired.
\end{example}

\begin{example}[Example 2: Cart Pole Model (Ref.~\cite{qagents_gym})] In the \textit{Cart Pole} environment there are two actions, which we denote as $(0,1)$ and Ref.~\cite{qagents_gym} uses a four-qubit PQC model in the form of Eq.~\ref{eq:first-form-app} with $(O_0,O_1) = (Z_1Z_2,Z_3Z_4)$. Here the actions are already in a binary encoding, so by defining
\begin{equation}
H^{(1)} = H_{1,3} + H_{2,4}
\end{equation}
and 
\begin{equation}
V(a) = e^{-i\pi a H^{(1)}}
\end{equation}
we have
\begin{align}
V(0) &= \mathds{1}\\
V(1) &= \mathrm{SWAP}_{1,3}\mathrm{SWAP}_{2,4}.
\end{align}
As a result, in this case we have
\begin{equation}
O_a = V^\dagger(a)Z_1Z_2V(a)
\end{equation}
for all $a\in\mathcal{A}$ as desired.
\end{example}

\begin{remark}
Of course the examples above do not prove that in general one can \textit{always} rewrite a model in the form of Eq.~\eqref{eq:first-form-app} in the form of Eq.~\eqref{eq:second-form-app}, with an appropriately parameterized $V(a)$ which uses a Hamiltonian data-encoding. Indeed, one \textit{necessary} (but not sufficient) condition for this to be possible is that the spectrum of all observables $\{O_a\}$ should be the same, otherwise they cannot all be obtained by unitary conjugation of a fixed observable. With this in mind, the examples above are meant only to show that, at least for previously considered models this is possible, and to give some general techniques via which this model conversion can sometimes be achieved.
\end{remark}

\section{From Data-Encoding Strategy to Frequency Set}\label{app:data-encoding}
\begin{figure}[h!]
\begin{center}
\includegraphics[scale=.6]{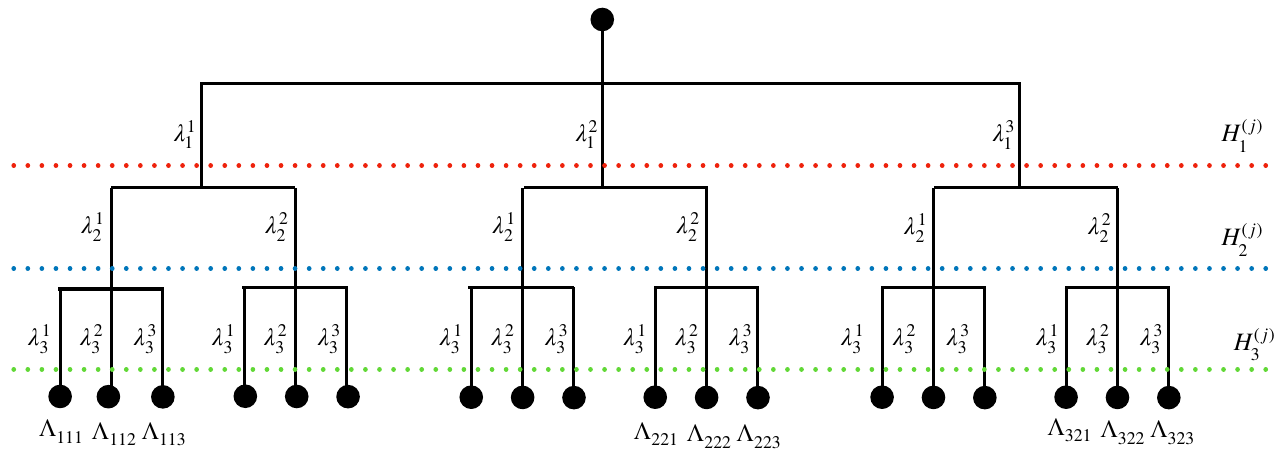}
\caption{Construction of the frequency set $\tilde{\Omega}_{\mathcal{D}}^{(j)}$ from the data-encoding strategy $\mathcal{D}^{(j)}$.}
\label{fig:tree_construction}
\end{center}
\end{figure}  
We describe here the way in which the frequency set $\tilde{\Omega}_{\mathcal{D}}$ of the PQC models considered in Section~\ref{ss:pqc-models} is constructed from the data-encoding strategy $\mathcal{D}$. We follow closely the presentation in Appendix A of Ref.~\cite{Sweke2025potential}. To simplify things, let's relabel the data encoding strategy $\mathcal{D}$ considered in Eq~\eqref{eq:data-encoding-strategy} via
\begin{align}
\mathcal{D} &=\left(\mathcal{D}^{(1)},\ldots, \mathcal{D}^{(d)},\tilde{\mathcal{D}}^{(1)},\ldots, \tilde{\mathcal{D}}^{(d')} \right) \\
&\equiv\left(\mathcal{D}^{(1)},\ldots, \mathcal{D}^{(d+d')}\right).
\end{align}
With this, we then have
\begin{equation}\label{eq:cartesian_freqs}
\tilde{\Omega}_{\mathcal{D}} = \tilde{\Omega}_{\mathcal{D}}^{(1)}\times \ldots \times \tilde{\Omega}_{\mathcal{D}}^{(d+d')}
\end{equation}
where $\tilde{\Omega}^{(j)}_\mathcal{D}\subseteq\mathbb{R}$ depends only on $\mathcal{D}^{(j)}$. We can therefore focus on the construction of $\tilde{\Omega}_{\mathcal{D}}^{(j)}$ for a single co-ordinate. In light of this, let us drop some coordinate-indicating superscripts for ease of presentation. In particular, let us write $\mathcal{D}^{(j)} = \{H_k\,|\, k \in [N_j] \}$, where we have dropped the coordinate-indicating superscripts from the Hamiltonians. We then use $\lambda^{i}_{k}$ to denote the $i$'th eigenvalue of $H_k$, and $N_{k}$ to denote the number of eigenvalues of $H_k$. We also introduce the multi-index $\vec{i} = (i_1,\ldots,i_{L_j})$, with $i_k\in [N_k]$, which allows us to define the sum of the eigenvalues indexed by $\vec{i}$, one from each Hamiltonian, as
\begin{equation}
\Lambda_{\vec{i}} = \lambda^{i_1}_1 +\ldots + \lambda^{i_{L_j}}_{L_j}.
\end{equation}
With this setup, we then have that the frequency set $\tilde{\Omega}_{\mathcal{D}}^{(j)}$ is given by the set of all differences of all possible sums of eigenvalues, i.e., 
\begin{equation}
\tilde{\Omega}_{\mathcal{D}}^{(j)} = \left\{\Lambda_{\vec{i}} - \Lambda_{\vec{j}}\,|\, \vec{i},\vec{j}\right\},
\end{equation}
and as mentioned before, the total frequency set is given by Eq.~\eqref{eq:cartesian_freqs}. There is a convenient graphical way to understand this construction, which is illustrated in Figure.~\ref{fig:tree_construction}. Essentially, one notes that, in order to construct $\tilde{\Omega}^{(j)}_\mathcal{D}$ one can consider a tree, with depth equal to the number of data-encoding gates, whose leaves contain the eigenvalue sums $\Lambda_{\vec{i}}$. The frequency set is then given by all possible pairwise differences between leaves.

\section{Proof of Theorem~\ref{thm:main-classical}}\label{app:proof-main-theorem}
In this appendix we provide the full proof of Theorem~\ref{thm:main-classical}. The proof proceeds through a series of intermediate results that progressively translate the abstract sufficient conditions into concrete polynomial bounds in terms of the problem size $(d_\mathcal{S},|\mathcal{A}|)$. We now give an overview of the proof:

\begin{itemize}
\item \textit{A starting observation (Section~\ref{ss:st_obs}).} We first verify that the regression noise arising at each iteration of FQI--KRR is conditionally sub-Gaussian with parameter $\tau = 2R_{\max}/(1-\gamma)$. This is used in the next step for applying kernel ridge regression error bounds.

\item \textit{Intermediate complexity bound (Section~\ref{ss:main-technical-lemma}).} We establish Lemma~\ref{lem:sufficient_m_K}, which provides sufficient conditions on the sample size $m$ and the number of iterations $K$ for Algorithm~\ref{alg:dequantization} to output an $\epsilon$-optimal policy. These conditions are expressed in terms of three abstract quantities: the maximum information gain $\Gamma(m)$, the maximum 2-norm of all frequency vectors in $\Omega_{\mathcal{D}}$ denoted by $L_{\mathcal{D}}$, and the RKHS norm $G_K$ of the Bellman targets. The proof combines an FQI error propagation result imported from \cite{theoretical_deepQ}(Theorem~\ref{th:fqi_error_propagation}) with pointwise kernel ridge regression error bounds, a discretization argument over the state-action space to deal with the continuous state space, and a Markov-inequality argument that converts the in-expectation bound on the integrated predictive uncertainty into a high-probability statement over the draw of the training set.

\item \textit{Bounding the maximum information gain (Section~\ref{ss:max-info-gain}).} Under Assumption~\ref{ass:poly-eigendecay} (inverse polynomial decay of kernel weight), we show that $\Gamma(m) = \mathcal{O}((m\log^{p-1}m)^{1/p})$ where $p$ controls the decay rate. We do this by decomposing the kernel into a finite-rank truncation and a tail, bounding each contribution separately (Lemma~\ref{lem:gamma_intermediate}), and optimally choosing the truncation rank (Lemma~\ref{lem:gamma_poly}). Substituting this into the conditions of Lemma~\ref{lem:sufficient_m_K} yields the explicit sample complexity shown in Corollary~\ref{cor:poly_sample}.

\item \textit{Bounding the frequency norm (Section~\ref{ss:Lk-term}).} We show that $L_{\mathcal{D}} \leq 2\sqrt{d}\,N_\mathrm{max}\lambda_\mathrm{max}$ (Lemma~\ref{lem:L_K-upperbound}). In particular, for typical Pauli-based data-encoding strategies one has $N_\mathrm{max} = \mathcal{O}(1)$ and $\lambda_\mathrm{max} = 1$, yielding $L_{\mathcal{D}} = \mathcal{O}(\sqrt{d})$, which ensures the logarithmic dependence on $L_{\mathcal{D}}$ in the sample complexity remains polynomial in the problem parameters.

\item \textit{Time complexity (Section~\ref{ss:time-complexity}).} We recall that when the reweighting vector $w$ is induced by a symmetric matrix product state with polynomial bond dimension, the kernel $K_{(\mathcal{D},w)}$ can be evaluated exactly and efficiently through tensor network contractions, yielding polynomial per-iteration complexity.

\item \textit{Combining the bounds (Section~\ref{ss:comb_bounds}).} We now combine the preceding results with the alignment condition on $G_K$ to obtain the stated polynomial sample and time complexity.
\end{itemize}

\subsection{A starting observation}\label{ss:st_obs}
Before proving the main technical lemma, we first establish a basic property of the regression noise that appears at each iteration of Fitted Q-Iteration with kernel ridge regression. Let $\epsilon_i = y_i - g_k^*(z_i)$ denote the regression error. The next lemma shows that these errors are conditionally sub-Gaussian. This fact is needed later when we invoke standard kernel ridge regression error bounds in Section~\ref{ss:main-technical-lemma}.

\begin{observation}[Sub-Gaussian regression noise in FQI--KRR]
\label{obs:subgaussian-noise}
Fix any iteration $k\in\{1,\dots,K\}$, and define
\begin{align}
y_i \;&\equiv\; r_i + \gamma \max_{a\in\mathcal{A}}\hat{Q}_{k-1}(s_i',a),\\
g_k^*(z_i) \;&\equiv\; (T\hat{Q}_{k-1})(z_i),\\
\epsilon_i \;&\equiv\; y_i-g_k^*(z_i),
\end{align}
where $z_i=(s_i,a_i)$. Then $\{\epsilon_i\}_{i=1}^m$ are conditionally independent, conditionally centered, and conditionally sub-Gaussian. In particular, for all $t\in\mathbb{R}$,
\begin{equation}
\mathbb{E}\!\left[e^{t\epsilon_i}\mid z_i\right]
\;\le\;
\exp\!\left(\frac{\tau^2 t^2}{2}\right),
\end{equation}
where
\begin{equation}
\tau \;=\; 2V_{\max} \;=\; \frac{2R_{\max}}{1-\gamma}.
\end{equation}
\end{observation}

\begin{proof} Fix an iteration \(k\), and define
\begin{equation}
y_i = r_i + \gamma\max_{a\in\mathcal A}\hat Q_{k-1}(s_i',a),\qquad
g_k^*(z_i)=(T\hat Q_{k-1})(z_i),\qquad
\epsilon_i=y_i-g_k^*(z_i),
\end{equation}
with \(z_i=(s_i,a_i)\).

We first show boundedness. By clipping in Algorithm~\ref{alg:dequantization},
\begin{equation}
\|\hat Q_{k-1}\|_{\infty}\le V_{\max},\qquad V_{\max}=\frac{R_{\max}}{1-\gamma}.
\end{equation}
Since \(|r_i|\le R_{\max}\) almost surely,
\begin{equation}
|y_i|
\le
R_{\max}+\gamma V_{\max}
=
\frac{R_{\max}}{1-\gamma}
=
V_{\max}.
\end{equation}
Also,
\begin{equation}
|g_k^*(z)|
=
|(T\hat Q_{k-1})(z)|
\le
R_{\max}+\gamma\|\hat Q_{k-1}\|_\infty
\le
V_{\max}.
\end{equation}
Therefore
\begin{equation}
|\epsilon_i|
\le
|y_i|+|g_k^*(z_i)|
\le
2V_{\max}.
\end{equation}

Next, \(\epsilon_i\) is conditionally centered:
\begin{equation}
\mathbb E[\epsilon_i\mid z_i]
=
\mathbb E[y_i\mid z_i]-g_k^*(z_i)
=
(T\hat Q_{k-1})(z_i)-g_k^*(z_i)
=
0.
\end{equation}

By Hoeffding's lemma, any centered random variable bounded in
\([-2V_{\max},2V_{\max}]\) is sub-Gaussian with parameter
\begin{equation}
\tau^2=\frac{(2V_{\max}-(-2V_{\max}))^2}{4}=4V_{\max}^2.
\end{equation}
Hence, for all \(t\in\mathbb R\),
\begin{equation}
\mathbb E\!\left[e^{t\epsilon_i}\mid z_i\right]
\le
\exp\!\left(\frac{\tau^2 t^2}{2}\right),
\qquad
\tau=2V_{\max}=\frac{2R_{\max}}{1-\gamma}.
\end{equation}

Finally, conditional independence of \(\{\epsilon_i\}_{i=1}^m\) follows from the sampling model (independent transition/reward draws across samples given the sampled state-action pairs). This proves Observation~\ref{obs:subgaussian-noise}. 
\end{proof}

\subsection{Intermediate Complexity Bound via RKHS Norm, $L_{\mathcal{D}}$, and Maximum Information Gain}\label{ss:main-technical-lemma}
We begin by establishing an intermediate technical lemma that expresses sufficient conditions on $m$ and $K$ in terms of the following three abstract quantities: 

\textbf{(a) Maximum information gain:}
The first quantity is the maximum information gain of the kernel $K_{(\mathcal{D},w)}$ on $m$ samples, which we denote with $\Gamma(m)$. This quantity captures how quickly the posterior uncertainty of the kernel ridge regressor shrinks with more samples and is the only term that depends on the intrinsic complexity of the PQC kernel. To be more specific, this quantity is defined via
\begin{equation}
\Gamma(m) = \max_{\{Z_i\}_{i=1}^m \subset \mathcal{Z}} \frac{1}{2} \log \det \left( \mathds{1}_m + \frac{1}{\lambda^2} K_{\{Z_i\}} \right),
\end{equation}
where $K_{\{Z_i\}}$ is the Gram matrix of the PQC kernel $K_{(\mathcal{D},w)}$. After stating our complexity bound in terms of $\Gamma(m)$, we will provide a detailed analysis of this quantity in Section~\ref{ss:max-info-gain}.

\textbf{(b) Maximum 2-norm of all frequency vectors in $\Omega_{\mathcal{D}}$:} Specifically, the quantity
\begin{equation}
L_{\mathcal{D}} = \max_{\omega\in\Omega_{\mathcal{D}}}\|\omega\|_2.
\end{equation}

\textbf{(c) Maximum RKHS norm of $g_k\equiv T\hat{Q}_{k-1}$ over $k=1,...,K$.} In the clipped two-sequence formulation of Algorithm~\ref{alg:dequantization}, Bellman targets are built from $\hat{Q}_{k-1}=\Pi_{[-V_{\max},V_{\max}]}F_{k-1}$ by application of the Bellman operator $T$. Given this, we define $g_k\equiv T\hat{Q}_{k-1}$ and note that the regression function in stage $k$ is precisely $g_k$. Moreover, the corresponding RKHS-controlled target-complexity term is $G_K \equiv \max_{1\leq k\leq K}\|g_k\|_{\mathcal{H}_{(\mathcal{D},w)}}=\max_{1\leq k\leq K}\|T\hat{Q}_{k-1}\|_{\mathcal{H}_{(\mathcal{D},w)}}$ -- the maximum RKHS norm attained by any ideal target function across the $K$ regression stages of Algorithm~\ref{alg:dequantization}. This was defined, interpreted and analyzed rigorously in Section~\ref{ss:RKHS-norm}, however here we mention that at a high-level $\|T\hat{Q}_{k-1}\|_{\mathcal{H}_{(\mathcal{D},w)}}$ can be interpreted as the \textit{``misalignment''} of the Bellman target with the weighting vector $w$ defining the PQC kernel $K_{(\mathcal{D},w)}$.

With this established, the key starting point for the proof of Theorem~\ref{thm:main-classical} is the following:

\begin{lemma}[Sufficient conditions for $\epsilon$-accuracy of FQI--KRR]
\label{lem:sufficient_m_K}
Consider Fitted Q-Iteration (FQI) with kernel ridge regression (KRR) in the fixed reproducing kernel Hilbert space (RKHS) $\mathcal{H}_{(\mathcal{D},w)}$ associated with the PQC kernel $K_{(\mathcal{D},w)}$, as per Algorithm~\ref{alg:dequantization}. Assume that the regularization parameter $\lambda > 0$ is a fixed positive constant, independent of the sample size $m$. Let $\epsilon > 0$ be a target accuracy and $\delta \in (0,1)$ a confidence parameter. Under the assumptions of Section~\ref{sub:assumptions}, define
\begin{align}
\varepsilon_{\mathrm{stat}} &\equiv \frac{\epsilon(1-\gamma)^2}{4\gamma\sqrt{|\mathcal{A}|\,p_{\max}(d_\mathcal{S},|\mathcal{A}|)}},\\
G_K &\equiv \max_{1\le k\le K}\|T\hat{Q}_{k-1}\|_{\mathcal{H}_{(\mathcal{D},w)}},\\
B_{\mathrm{net}}(m) &\equiv \max\!\left\{G_K,\; \frac{\sqrt{m}\,V_{\max}}{\lambda}\right\},\\
\mathcal{F}_m &\equiv G_K + \frac{\tau}{\lambda}\sqrt{2\log\frac{4K|\mathcal{Z}_m|}{\delta}},
\end{align}
where $|\mathcal{Z}_m| \le |\mathcal{A}|\bigl(\sqrt{d_{\mathcal{S}}}\,L_{\mathcal{D}}\,B_{\mathrm{net}}(m)\,m + 1\bigr)^{d_{\mathcal{S}}}$, $L_{\mathcal{D}}=\max_{\omega\in\Omega_{\mathcal{D}}}\|\omega\|_2$, $\tau$ is the sub-Gaussian constant from Observation~\ref{obs:subgaussian-noise}, $p_\mathrm{max}$ is the function from Assumption~\ref{ass:bounded-density}, and $\gamma$ is the discount factor.

Suppose the sample size $m$ and the number of iterations $K$ satisfy
\begin{align}
m&\ge \frac{36K\,\Gamma(m+1)\,\mathcal{F}_m^2}{\delta\,\varepsilon_{\mathrm{stat}}^2\,\log(1+1/\lambda^2)}-1, \label{eq:m_condition1}\\[4pt]
m &\ge \frac{18\,\mathcal{F}_m^2}{\varepsilon_{\mathrm{stat}}^2\,B_{\mathrm{net}}(m)}, \label{eq:m_condition2}\\[4pt]
m &\ge \frac{6}{\varepsilon_{\mathrm{stat}}}, \label{eq:m_condition3}\\[4pt]
K &\ge \frac{\log\!\left(\frac{8R_{\max}}{(1-\gamma)^2\epsilon}\right)}{\log(1/\gamma)}-1,
\label{eq:K_condition}
\end{align}
where $R_{\max}$ is an upper bound on the maximum absolute reward. Then, with probability at least $1-\delta$, the greedy policy $\pi_K$ output by Algorithm~\ref{alg:dequantization} satisfies
\begin{equation}\label{eq:acc-guarantee}
\mathcal{L}(Q^{\pi_K}) \;\le\; \mathcal{L}(Q^*_{(\mathcal{D},w)}) + \epsilon = \epsilon.
\end{equation}
\end{lemma}
The complete proof of Lemma~\ref{lem:sufficient_m_K} is given in Appendix~\ref{a:proof_lemm_mk}.

With this established, the remainder of the proof proceeds by analyzing each of the abstract quantities involved in the intermediate lemma above.

\subsection{Bounding the Maximum Information Gain}\label{ss:max-info-gain}

We now show that under Assumption~\ref{ass:poly-eigendecay}, the maximum information gain $\Gamma(m)$ grows sublinearly in $m$, which is a sufficient condition for polynomial sample complexity. Note from Eq.~\eqref{eq:feature-map-form} that
\begin{align}
K_{(\mathcal{D},w)}(z,z') &= \frac{w^2_0}{\|w\|^2_2} + \sum_{i = 1}^M\frac{w^2_i}{\|w\|^2_2} \left[\cos(\langle \omega_i,z\rangle)\cos(\langle \omega_i,z'\rangle) + \sin(\langle \omega_i,z\rangle)\sin(\langle \omega_i,z'\rangle) \right] \\
&= \frac{w^2_0}{\|w\|^2_2} + \sum_{i = 1}^M\frac{w^2_i}{2\|w\|^2_2} \left[\sqrt{2}\cos(\langle \omega_i,z\rangle)\sqrt{2}\cos(\langle \omega_i,z'\rangle) + \sqrt{2}\sin(\langle \omega_i,z\rangle)\sqrt{2}\sin(\langle \omega_i,z'\rangle) \right].
\end{align}
Therefore if we define 
\begin{equation}\label{eq:PQC-mercer_1}
(\nu_0,\nu_1,\ldots,\nu_{2M}) = \left(\frac{w^2_0}{\|w\|_2^2}, \frac{w^2_1}{2\|w\|_2^2}, \frac{w^2_1}{2\|w\|_2^2},\ldots, \frac{w^2_M}{2\|w\|_2^2},\frac{w^2_M}{2\|w\|_2^2} \right)
\end{equation}
and
\begin{equation}\label{eq:PQC-mercer_2}
\big(\phi^{(0)}(x),\phi^{(1)}(x),\ldots, \phi^{(2M)}(x)\big) = \big(1,\sqrt{2}\cos(\langle \omega_1,x\rangle),\sqrt{2}\sin(\langle \omega_1,x\rangle),\ldots, \sqrt{2}\cos(\langle \omega_M,x\rangle),\sqrt{2}\sin(\langle \omega_M,x\rangle)\big)
\end{equation}
then we have the decomposition
\begin{equation}\label{eq:PQC-mercer_3}
K_{(\mathcal{D},w)}(z,z') = \sum_{i = 0}^{2M}\nu_i\phi^{(i)}(z)\phi^{(i)}(z'),
\end{equation}
where $\sum_{i = 0}^{2M}\nu_i = 1$.

Without loss of generality we assume that the weight vector $w = (w_0,\dots,w_M)$ has been permuted to be non-increasing, i.e., $w_0 \ge w_1 \ge \dots \ge w_M > 0$, from which it follows that $\nu_0 \ge \nu_1 = \nu_2 \ge \nu_3 = \nu_4 \ge \dots \ge \nu_{2M-1} = \nu_{2M}$.

\begin{lemma}[Intermediate upper bound on $\Gamma(m)$ via Kernel decomposition -- adapted from Ref~\cite{vakili2024kernelizedreinforcementlearningorder}]
\label{lem:gamma_intermediate} 
Given the decomposition of $K_{(\mathcal{D},w)}$ defined via Eqs.~\eqref{eq:PQC-mercer_1}, \eqref{eq:PQC-mercer_2}, and \eqref{eq:PQC-mercer_3}, let
\begin{equation}
K^T(z,z') \equiv \sum_{i=0}^T \nu_i \, \phi^{(i)}(z) \, \phi^{(i)}(z')
\end{equation}
be the \emph{rank-$T$ truncation} of the decomposition. Define the tail kernel $K^0 \equiv K - K^T$ and the associated tail error
\begin{equation}
\epsilon_T \equiv \sup_{z\in\mathcal{Z}} K^0(z,z).
\end{equation}
Then
\begin{equation}
\Gamma(m) \le \frac{T}{2} \log\!\left(1 + \frac{m}{\lambda^2 T}\right) + \frac{m \epsilon_T}{2\lambda^2}.
\end{equation}
\end{lemma}
The proof of Lemma~\ref{lem:gamma_intermediate} is given in Appendix~\ref{a:gamma_upper}.

Under Assumption~\ref{ass:poly-eigendecay}, one can optimally balance the truncation rank $T$ in Lemma~\ref{lem:gamma_intermediate} to obtain a tight asymptotic characterization of the information gain:

\begin{lemma}[Maximum information gain under polynomial weight decay]
\label{lem:gamma_poly}
Under Assumption~\ref{ass:poly-eigendecay}, and for any fixed constant regularization parameter $\lambda$, there exists a constant $C_4>0$ (depending only on $C$, $p$, and $\lambda$) such that for all sufficiently large $m$
\begin{equation}
\Gamma(m+1) \le C_4 \bigl( (m+1) \, (\log(m+1))^{p-1} \bigr)^{1/p}.
\end{equation}
In particular, $\Gamma(m+1) = o(m)$ for any $p>1$.
\end{lemma}

The proof of Lemma~\ref{lem:gamma_poly} appears in Appendix~\ref{a:gamma_poly}.

Lemma~\ref{lem:gamma_poly} with Lemma~\ref{lem:sufficient_m_K} yields an explicit sufficient sample size:

\begin{corollary}[Sufficient sample complexity for Algorithm~\ref{alg:dequantization}]
\label{cor:poly_sample}
Assume that the reweighting vector $w$ defining the PQC kernel $K_{(\mathcal{D},w)}$ satisfies Assumption~\ref{ass:poly-eigendecay}. Then, under the assumptions of Lemma~\ref{lem:sufficient_m_K}, let $\alpha \equiv (p-1)/p\in(0,1)$ and define
\begin{align}
\mathcal{P} &\equiv 2G_K^2 + \frac{4\tau^2}{\lambda^2}\log\!\frac{4K|\mathcal{A}|\,B_*^{d_{\mathcal{S}}}}{\delta},
\label{eq:calP-def}\\
\mathcal{Q} &\equiv \frac{6\tau^2 d_{\mathcal{S}}}{\lambda^2},
\label{eq:calQ-def}
\end{align}
where $G_K \equiv \max_{1\le k\le K}\|T\hat{Q}_{k-1}\|_{\mathcal{H}_{(\mathcal{D},w)}}$, $B_* \equiv \sqrt{d_{\mathcal{S}}}\,L_{\mathcal{D}}(G_K + V_{\max}/\lambda) + 1$, $L_{\mathcal{D}}\equiv\max_{\omega\in\Omega_{\mathcal{D}}}\|\omega\|_2$, and $V_{\max}\equiv R_{\max}/(1-\gamma)$. Then, provided that
\begin{equation}
K \ge \frac{\log\!\left(\frac{8R_{\max}}{(1-\gamma)^2\epsilon}\right)}{\log(1/\gamma)}-1,
\end{equation}
one has that 
\begin{equation}
m \;=\; \widetilde{\mathcal{O}}\left(\left( \frac{|\mathcal A|\,p_{\max}} {\delta\,\epsilon^2} \right)^{p/(p-1)} \max\!\left\{G_K^2,\; L,\; d_{\mathcal{S}}\right\}^{p/(p-1)}\right),
\end{equation}
is sufficient for 
\begin{equation}
\mathcal{L}(Q^{\pi_K}) \;\le\; \epsilon,
\end{equation}
where $L = \log\bigl(4K|\mathcal{A}|\,B_*^{d_{\mathcal{S}}}/\delta\bigr)$, and $\widetilde{O}(\cdot)$ hides polylogarithmic factors in $1/\epsilon$, $1/\delta$, and constants depending only on $\gamma$, $\lambda$, $\tau$, $C$, and $p$.
\end{corollary}
The proof of Corollary~\ref{cor:poly_sample} can be found in Appendix~\ref{a:coro_sample}.

\begin{observation}\label{obs:poly-sample-complexity}
Corollary~\ref{cor:poly_sample} shows that the sample complexity $m$ is polynomial in $d_{\mathcal{S}},|\mathcal{A}|, p_{\max}, 1/\epsilon,1/\delta$ provided that $G_K = \mathcal{O}(\mathrm{poly}(d_{\mathcal{S}},|\mathcal{A}|))$, $L_{\mathcal{D}}= 2^{\mathcal{O}(\mathrm{poly}(d_{\mathcal{S}},|\mathcal{A}|))}$, and $p_{\max} = \mathcal{O}(\mathrm{poly}(d_{\mathcal{S}},|\mathcal{A}|))$.
\end{observation}

\subsection{Bounding the Frequency Norm $L_{\mathcal{D}}$}\label{ss:Lk-term}

We now establish that the frequency norm $L_{\mathcal{D}}$ is controlled by the data-encoding strategy. Recall that 
\begin{align}
N_{\mathrm{max}} &\coloneqq \max_{j\in [d_{\mathcal{S}}+d_{\mathcal{A}}]} \left[N_j\right],\\
\lambda_\mathrm{max} &\coloneqq \max_{j\in [d_{\mathcal{S}}+d_{\mathcal{A}}]}\left[\max_{k\in [N_j]}\left[\|H_k^{(j)}\|_{\rm op}\right]\right]
\end{align}
denote the maximum number of encoding Hamiltonians per data component and the largest eigenvalue, respectively. We then have

\begin{lemma}[Upper bound on $L_{\mathcal{D}}$]\label{lem:L_K-upperbound} We have 
\begin{equation}
L_{\mathcal{D}}\leq 2\sqrt{d}N_\mathrm{max}\lambda_\mathrm{max},
\end{equation}
where $d = d_\mathcal{S} + d_\mathcal{A}\leq d_\mathcal{S} + |\mathcal{A}|$.
\end{lemma}

\begin{proof} It follows from the way in which $\Omega_\mathcal{D}$ is constructed from $\tilde{\Omega}_\mathcal{D}$ that $\max_{\omega\in\Omega_{\mathcal{D}}}\|\omega\|_2 = \max_{\omega\in\tilde{\Omega}_{\mathcal{D}}}\|\omega\|_2$, and we therefore focus on upper bounding the latter quantity. To this end, recall that the frequency set $\tilde{\Omega}_\mathcal{D}$ has a Cartesian product structure
\begin{equation}
\tilde{\Omega}_\mathcal{D} = \tilde{\Omega}^{(1)}_\mathcal{D} \times \cdots \times \tilde{\Omega}^{(d)}_\mathcal{D},
\end{equation}
and as a result for any $\omega\in\tilde{\Omega}_{\mathcal{D}}$, we have
\begin{equation}
\|\omega\|_2 = \sqrt{\sum_{j=1}^d (\omega^{(j)})^2} \leq \sqrt{\sum_{j=1}^d \left(\omega^{(j)}_{\max}\right)^2},
\end{equation}
where $\omega^{(j)}_{\max} = \max\{\omega^{(j)} \in \tilde{\Omega}^{(j)}_\mathcal{D}\}$. Therefore
\begin{equation}
L_{\mathcal{D}} = \max_{\omega\in\tilde{\Omega}_{\mathcal{D}}}\|\omega\|_2 \leq \sqrt{\sum_{j=1}^d \left(\omega^{(j)}_{\max}\right)^2}.
\end{equation}
Given this, it remains to upper bound $\omega^{(j)}_{\max}$ for each $j\in[d]$. Recall from Appendix~\ref{app:data-encoding} that the frequencies in $\tilde{\Omega}^{(j)}_\mathcal{D}$ are constructed from differences of sums of eigenvalues of the encoding Hamiltonians $\{H^{(j)}_1,\ldots,H^{(j)}_{N_j}\}$. Specifically,
\begin{equation}
\tilde{\Omega}^{(j)}_\mathcal{D} = \left\{\Lambda^{(j)}_{\vec{i}} - \Lambda^{(j)}_{\vec{j}} \,\Big|\, \vec{i}, \vec{j}\right\},
\end{equation}
where 
\begin{equation}
\Lambda^{(j)}_{\vec{i} = i_1,\ldots,i_{N_j}} = \sum_{k=1}^{N_j} \lambda^{i_k}_k
\end{equation}
is the sum of the $i_k$'th eigenvalues of the $k$ Hamiltonians in $\mathcal{D}^{(j)}$. As such, if we define
\begin{equation}
\lambda^{(j)}_{\max} = \max_{k\in N_j}[\|H^{(j)}_k\|_{\rm op}]
\end{equation}
to be the largest singular value of all Hamiltonians encoding component $j$ then we immediately have
\begin{equation}
\Lambda^{(j)}_{\vec{i}} \leq N_j\lambda^{(j)}_{\max}
\end{equation}
for all $\vec{i}$, from which it follows that
\begin{equation}
\max_{\vec{i}}\left(\Lambda^{(j)}_{\vec{i}}\right) \leq N_j \lambda^{(j)}_{\max}.
\end{equation}
Therefore
\begin{equation}
\omega^{(j)}_{\max} \leq 2N_j\lambda^{(j)}_{\max} \leq 2N_{\max}\lambda_{\max},
\end{equation}
where in the last step we used the definitions $N_{\max} = \max_{j\in[d]}[N_j]$ and $\lambda_{\max} = \max_{j\in[d]}\left[\lambda^{(j)}_\mathrm{max}\right]$. Substituting back gives
\begin{equation}
L_{\mathcal{D}} = \sqrt{\sum_{j=1}^d \left(\omega^{(j)}_{\max}\right)^2} \leq \sqrt{\sum_{j=1}^d \left(2N_{\max}\lambda_{\max}\right)^2}= 2\sqrt{d}\, N_{\max}\lambda_{\max},
\end{equation}
which completes the proof. 
\end{proof}
For typical Pauli encoding strategies, $N_\mathrm{max}$ and $\lambda_\mathrm{max}$ are constants, ensuring $L_{\mathcal{D}} = \mathcal{O}(\sqrt{d})$.

\subsection{Time Complexity: Efficient Kernel Evaluation}\label{ss:time-complexity}

Each iteration of Algorithm~\ref{alg:dequantization} requires multiple evaluations of the kernel $K_{(\mathcal{D},w)}$. Since the feature space dimension $M$ scales exponentially with $d$, naive evaluation is intractable. Two approaches circumvent this:
\begin{enumerate}
\item \textbf{Random Fourier Features (RFF)~\cite{landman2022classicallyapproximatingvariationalquantum,Sweke2025potential,sahebi2025dequantizationsupervisedquantummachine}:} Approximates $K_{(\mathcal{D},w)}$ using a polynomial-dimension random feature map, efficient whenever $w$ defines an efficiently samplable distribution over $\Omega_{\mathcal{D}}$.
\item \textbf{Tensor network contractions~\cite{sweke2025kernelbaseddequantizationvariationalqml}:} Evaluates $K_{(\mathcal{D},w)}$ exactly and efficiently when $w$ is induced from a symmetric MPS with polynomial bond-dimension, exploiting the Cartesian product structure of $\Omega_{\mathcal{D}}$.
\end{enumerate}
Since Alice chooses $w$, she can always select a weight vector admitting efficient evaluation. For the analysis of the preceding sections to hold without modification (i.e., without additional approximation error from RFF), it suffices that $w$ is induced from a symmetric MPS as in Ref.~\cite{sweke2025kernelbaseddequantizationvariationalqml}.

\subsection{Combining the Bounds}\label{ss:comb_bounds}
Collecting the results above: under Conditions 1--5 of Theorem~\ref{thm:main-classical}, Corollary~\ref{cor:poly_sample} together with Lemma~\ref{lem:L_K-upperbound} and the alignment condition $G_K = \mathcal{O}(\mathrm{poly}(d_\mathcal{S},|\mathcal{A}|))$ ensure that the sample complexity $m$ is polynomial. The tensor network structure on $w$ ensures time efficiency. The number of iterations $K = \mathcal{O}(\log(1/\epsilon))$ is always polynomial. This completes the proof of Theorem~\ref{thm:main-classical}. \qed

\section{Proof of \Cref{lem:sufficient_m_K}}\label{a:proof_lemm_mk}

Although \Cref{lem:sufficient_m_K} is stated for the PQC kernel $K_{(\mathcal{D},w)}$, we prove a more general version that holds for any kernel $k$ whose associated RKHS $\mathcal{H}_K$ satisfies the assumptions in \Cref{sub:assumptions}. The key structural assumption used below is realizability of the clipped Bellman targets: if $\hat{Q}_{k-1}=\Pi_{[-V_{\max},V_{\max}]}F_{k-1}$, then $g_{k-1}\equiv T\hat{Q}_{k-1}\in\mathcal{H}_K$. We also use that the noise terms are independent $\tau^2$-sub-Gaussian random variables (Observation~\ref{obs:subgaussian-noise}). In the following, we denote by $\mathcal{H}_K$ the RKHS associated with the kernel $k$, and all norms and inner products without subscripts refer to this space. We note that the error in the FQI error-propagation theorem below is measured with respect to the $L^2$ norm weighted by the sampling distribution $\sigma$; this weighted $L^2$ norm is the natural error metric throughout the proof.

We begin by recalling the following error-propagation result for Fitted Q-Iteration, restated from Ref.~\cite{theoretical_deepQ}. In the statement below, $\sigma$ denotes the state-action distribution from which the training data is sampled at each iteration of Fitted Q-Iteration, while $\mu$ denotes the (possibly different) state-action distribution with respect to which the resulting policy is evaluated.

\begin{theorem}[FQI error propagation (Theorem 6.1 in \cite{theoretical_deepQ})]
\label{th:fqi_error_propagation}
Let $\{\hat{Q}_k\}_{k=0}^K$ be the sequence of action-value estimators produced by Fitted Q-Iteration, and let
\begin{equation}
\pi_K(s) = \arg\max_a \hat{Q}_K(s,a)
\end{equation}
be the greedy policy induced by the final estimate $\hat{Q}_K$. Let $\sigma$ be the sampling distribution over $\mathcal Z=\mathcal S\times\mathcal A$ used by Fitted Q-Iteration, let $\mu$ be a fixed distribution over $\mathcal Z$, and let $T$ denote the optimal Bellman operator. Assume there exists a constant $\phi_{\mu,\sigma}<\infty$ such that
\begin{equation}
(1-\gamma)^2
\sum_{m\ge 1}
\gamma^{m-1}m\,
\sup_{\pi_1,\ldots,\pi_m}
\left[
\mathbb E_{\sigma}
\left|
\frac{
d(P^{\pi_m}P^{\pi_{m-1}}\cdots P^{\pi_1}\mu)
}{
d\sigma
}
\right|^2
\right]^{1/2}
\le
\phi_{\mu,\sigma},
\end{equation}
where $P^{\pi_j}$ denotes the transition operator induced by policy $\pi_j$. Then
\begin{equation}
\mathbb{E}_{(s,a)\sim\mu}\!\left[|Q^*(s,a)-Q^{\pi_K}(s,a)|\right]
\le
\frac{2\phi_{\mu,\sigma}\gamma}{(1-\gamma)^2}
\max_{1\le k\le K}
\left(
\mathbb{E}_{(s,a)\sim\sigma}
\!\left[
|\hat{Q}_k(s,a)-T\hat{Q}_{k-1}(s,a)|^2
\right]
\right)^{1/2}
+
\frac{4\gamma^{K+1}}{(1-\gamma)^2}R_{\max}.
\label{eq:error_decomposition_main}
\end{equation}
\end{theorem}

In our setting, Algorithm~\ref{alg:dequantization} samples training data from $\sigma = U(\mathcal{Z})$ (Assumption~\ref{ass:sampling-model}), and the loss $\mathcal{L}$ of Definition~\ref{def:loss-function} evaluates the resulting policy with respect to this same uniform distribution. We will therefore use Theorem~\ref{th:fqi_error_propagation} with $\mu=\sigma=U(\mathcal{Z})$. For this, we need to bound the resulting concentrability coefficient $\phi_{U,U}$. We do so using Assumption~\ref{ass:bounded-density}, which requires that the transition density $p(s'\mid s,a)$ with respect to $U(\mathcal{S})$ is uniformly bounded by $p_{\max}<\infty$.

\begin{lemma}[Concentrability coefficient for uniform sampling and evaluation]
\label{lem:concentrability-UU}
Under Assumption~\ref{ass:bounded-density}, if we choose $\mu=\sigma=U(\mathcal{Z})$ in Theorem~\ref{th:fqi_error_propagation}, then
\begin{equation}
\phi_{U,U} \;\le\; \sqrt{|\mathcal{A}|\,p_{\max}}.
\end{equation}
\end{lemma}

\begin{proof}
For any probability distribution $\eta$ on $\mathcal{Z}$, write $f_\eta \equiv d\eta/dU(\mathcal{Z})$ for its density with respect to $U(\mathcal{Z})=U(\mathcal{S})\times U(\mathcal{A})$, so that $\mathbb{E}_{U(\mathcal{Z})}[f_\eta]=1$. Fix any policy $\pi$ and let $\eta' = P^{\pi}\eta$ be the state-action distribution obtained by sampling $(s,a)\sim\eta$, then $s'\sim P(\cdot\mid s,a)$ and $a'\sim\pi(\cdot\mid s')$. Unwinding the definition of $\eta'$ in terms of $f_\eta$ and the transition density $p(\cdot\mid\cdot,\cdot)$ gives, for every $(s',a')\in\mathcal{Z}$,
\begin{equation}
f_{\eta'}(s',a')
= \pi(a'\mid s')\sum_{a\in\mathcal{A}}\mathbb{E}_{s\sim U(\mathcal{S})}\bigl[f_\eta(s,a)\,p(s'\mid s,a)\bigr].
\end{equation}
By Assumption~\ref{ass:bounded-density}, for every $(s,a)$ and every measurable $B\subseteq\mathcal{S}$,
\[
P(B\mid s,a)\le p_{\max}U(\mathcal{S})(B),
\]
which implies that $P(\cdot\mid s,a)$ has a density $p(s'\mid s,a)$ with respect to $U(\mathcal{S})$ satisfying $p(s'\mid s,a)\le p_{\max}$. Using this, joint with the facts that $\pi(a'\mid s')\le 1$, and $\mathbb{E}_{U(\mathcal{Z})}[f_\eta]=1$, we obtain
\begin{equation}
f_{\eta'}(s',a')
\le p_{\max}\sum_{a\in\mathcal{A}}\mathbb{E}_{s\sim U(\mathcal{S})}\bigl[f_\eta(s,a)\bigr]
= |\mathcal{A}|\,p_{\max}\,\mathbb{E}_{U(\mathcal{Z})}[f_\eta]
= |\mathcal{A}|\,p_{\max}.
\end{equation}
Crucially, this bound holds for \emph{any} distribution $\eta$ and any policy $\pi$, so it applies in particular at every stage of an arbitrary composition of transitions. Specifically, let $\eta_m \equiv P^{\pi_m}\cdots P^{\pi_1}U(\mathcal{Z})$ for $m\ge 1$ and any policy sequence $\pi_1,\ldots,\pi_m$. Since $\eta_m$ is obtained from $\eta_{m-1}$ (with $\eta_0=U(\mathcal{Z})$) by a single application of $P^{\pi_m}$, the bound above gives, for every $m\ge 1$,
\begin{equation}
f_{\eta_m}(s,a) \le |\mathcal{A}|\,p_{\max} \qquad \text{for all } (s,a)\in\mathcal{Z},
\end{equation}
uniformly over $\pi_1,\ldots,\pi_m$. Since $f_{\eta_m}\ge 0$ and $\mathbb{E}_{U(\mathcal{Z})}[f_{\eta_m}]=1$, this pointwise bound also implies
\begin{equation}
\mathbb{E}_{U(\mathcal{Z})}\bigl[f_{\eta_m}^2\bigr] \le |\mathcal{A}|\,p_{\max}\cdot\mathbb{E}_{U(\mathcal{Z})}[f_{\eta_m}] = |\mathcal{A}|\,p_{\max},
\end{equation}
so that, for every $m\ge 1$,
\begin{equation}
\sup_{\pi_1,\ldots,\pi_m}\left(\mathbb{E}_{U(\mathcal{Z})}\bigl[f_{\eta_m}^2\bigr]\right)^{1/2} \le \sqrt{|\mathcal{A}|\,p_{\max}}.
\end{equation}
Substituting this into the series defining $\phi_{\mu,\sigma}$ in Theorem~\ref{th:fqi_error_propagation}, and using $\sum_{m\ge 1} m\gamma^{m-1} = (1-\gamma)^{-2}$,
\begin{equation}
(1-\gamma)^2\sum_{m\ge 1}\gamma^{m-1}m\sup_{\pi_1,\ldots,\pi_m}\left(\mathbb{E}_{U(\mathcal{Z})}\bigl[f_{\eta_m}^2\bigr]\right)^{1/2}
\le (1-\gamma)^2\sqrt{|\mathcal{A}|\,p_{\max}}\sum_{m\ge 1}m\gamma^{m-1}
= \sqrt{|\mathcal{A}|\,p_{\max}}.
\end{equation}
Hence $\phi_{U,U}=\sqrt{|\mathcal{A}|\,p_{\max}}$ satisfies the defining inequality of $\phi_{\mu,\sigma}$, which proves the lemma.
\end{proof}

The decomposition of the error given in Theorem~\ref{th:fqi_error_propagation} separates the total error into a \emph{statistical} component (the first term: the maximum approximation error at each iteration) and an \emph{algorithmic} component (the second term: which vanishes exponentially with the number of iterations). Our strategy to ensure the total error is at most $\epsilon$ is to bound each component by at most $\epsilon/2$ with high probability over both the random training data and the noise. Controlling the algorithmic term by requiring the second term to be at most $\epsilon/2$ yields condition~\eqref{eq:K_condition}. The prefactors in the statistical term depend only on $\gamma$, so it remains to bound the quantity $\max_{1\le k\le K} \left(\mathbb{E}_{(s,a)\sim\sigma}\!\left[|\hat{Q}_k(s,a) - T\hat{Q}_{k-1}(s,a)|^2\right]\right)^{1/2}$.

For full rigor, we analyze each iteration conditionally on the past. For $k\in\{1,\dots,K\}$, let $\mathcal{F}_{k-1}$ be the $\sigma$-field generated by all batches and noises up to round $k-1$. Then $\hat{Q}_{k-1}$ and $g_{k}=T\hat{Q}_{k-1}$ are $\mathcal{F}_{k-1}$-measurable, while the current batch and current noise are fresh draws independent of $\mathcal{F}_{k-1}$. Hence every concentration/discretization statement at round $k$ is first proved under $\Pr(\cdot\mid \mathcal{F}_{k-1})$. For readability, we suppress this conditioning in the notation below, but it is always understood.

In order to upper bound this quantity, we first derive a high-probability bound on the RKHS predictor error $|g_{k-1}(z)-F_k(z)|$, where $g_{k-1}\equiv T\hat{Q}_{k-1}$ denotes the Bellman target and $F_k$ denotes the (unclipped) kernel ridge regression output at iteration $k$, before applying the projection $\Pi_{[-V_{\max},V_{\max}]}$. To do this, let $\mathbf{K}_m(z) = (K(z,z_1),\dots,K(z,z_m))^\top$ be the kernel vector, $K_m$ the Gram matrix, and define the vectors of target values and noise as
\begin{equation}
\mathbf{y}_{k-1} = \bigl(g_{k-1}(z_1),\dots,g_{k-1}(z_m)\bigr)^\top,
\qquad
\boldsymbol{\epsilon}_m = (\epsilon_1,\dots,\epsilon_m)^\top,
\end{equation}
where, under Observation~\ref{obs:subgaussian-noise}, the regression labels are modeled as $y_i=g_{k-1}(z_i)+\epsilon_i$, with the noise terms $\epsilon_i$ independent $\tau^2$-sub-Gaussian random variables. The kernel ridge regression predictor at iteration $k$ is the function
\begin{equation}\label{eq:Fk_expression}
F_k(z) = \mathbf{K}_m(z)^\top (K_m + \lambda^2 \mathds{1}_m)^{-1} (\mathbf{y}_{k-1} + \boldsymbol{\epsilon}_m).
\end{equation}
We stress that $F_k(z)$ is a random variable depending on both the randomly drawn training set $\{z_i\}_{i=1}^m$ and the noise realization $\boldsymbol{\epsilon}_m$. Consequently, the prediction error $g_{k-1}(z) - F_k(z)$ is also a random variable depending on these same sources of randomness. A direct computation yields the exact decomposition
\begin{equation}\label{eq:det-noise}
g_{k-1}(z) - F_k(z)
= \underbrace{g_{k-1}(z) - \mathbf{K}_m(z)^\top (K_m + \lambda^2 \mathds{1}_m)^{-1} \mathbf{y}_{k-1}}_{\epsilon_{\mathrm{det}}(z;\{z_i\}):\ \text{bias error}}
- \underbrace{\mathbf{K}_m(z)^\top (K_m + \lambda^2 \mathds{1}_m)^{-1} \boldsymbol{\epsilon}_m}_{\epsilon_{\mathrm{noise}}(z;\{z_i\}):\ \text{noise-induced error}}.
\end{equation}
We emphasize that both $\epsilon_{\mathrm{det}}(z;\{z_i\})$ and the predictive uncertainty $\sigma_m(z;\{z_i\})$ introduced below depend on the training set $\{z_i\}_{i=1}^m$ in addition to $z$. We suppress this dependence and write simply $\epsilon_{\mathrm{det}}(z)$ and $\sigma_m(z)$ when the training set is understood from context.

To bound the above terms it will be helpful to introduce $\sigma_m(z)$, the predictive uncertainty of the kernel ridge predictor using $m$ data samples, defined via
\begin{equation}
\sigma_m(z)=\sqrt{K(z,z)-\mathbf{K}_m(z)^\top (K_m + \lambda^2 \mathds{1}_m)^{-1}\mathbf{K}_m(z)}.
\end{equation}
With this in hand, the bias error $\epsilon_{\mathrm{det}}(z)$ is controlled via the following lemma, which holds for any fixed realization of the training set $\{z_i\}_{i=1}^m$ and any fixed $z\in\mathcal{Z}$ (no randomness is invoked in its proof):

\begin{lemma}[Bias error bound (fixed dataset)]
\label{lem:deterministic}
For any $z\in\mathcal{Z}$ and any fixed training set $\{z_i\}_{i=1}^m$,
\begin{equation}
\bigl|\epsilon_{\mathrm{det}}(z)\bigr|
\le \|g_{k-1}\|_{\mathcal{H}_K}\,\sigma_m(z).
\end{equation}
\end{lemma}
\begin{proof}
Define the coefficient vector
\begin{equation}
\vec{\xi}_m = (K_m + \lambda^2 \mathds{1}_m)^{-1} \mathbf{K}_m(z) \in \mathbb{R}^m,
\end{equation}
whose components are $\xi_i = [\vec{\xi}_m]_i$ for $i=1,\dots,m$. The noise-free kernel ridge predictor evaluated at $z$ can then be written in terms of $\vec{\xi}_m$ via
\begin{equation}
\mathbf{K}_m(z)^\top (K_m + \lambda^2 \mathds{1}_m)^{-1} \mathbf{y}_{k-1} = \sum_{i=1}^m \xi_i \, g_{k-1}(z_i).
\end{equation}
Using Assumption~\ref{as:func_classes}, we have $g_{k-1}=T\hat{Q}_{k-1}\in \mathcal{H}_K$. Up to this point the analysis did not rely on this assumption; however, it is needed here to ensure that the Bellman target belongs to the RKHS, which allows us to express it through the reproducing property. In particular, for every $i$,
\begin{equation}
\sum_{i=1}^m \xi_i \, g_{k-1}(z_i) = \sum_{i=1}^m \xi_i \langle g_{k-1}, k(\cdot,z_i) \rangle_{\mathcal{H}_K}
= \Big\langle g_{k-1}, \sum_{i=1}^m \xi_i k(\cdot,z_i) \Big\rangle_{\mathcal{H}_K}.
\end{equation}
The analysis can proceed without this assumption, but doing so requires handling the case in which the Bellman target does not lie in the hypothesis class, which introduces an additional approximation error term and significantly complicates the argument. We discuss this more general setting in \Cref{ap:relaxing_closure}.

Therefore, using the reproducing property again, the bias error admits the inner-product representation
\begin{align}
\epsilon_{\mathrm{det}}(z)
&= g_{k-1}(z) - \mathbf{K}_m(z)^\top (K_m + \lambda^2 \mathds{1}_m)^{-1} \mathbf{y}_{k-1} \nonumber\\
&= \langle g_{k-1}, k(\cdot,z) \rangle_{\mathcal{H}_K} - \Big\langle g_{k-1}, \sum_{i=1}^m \xi_i k(\cdot,z_i) \Big\rangle_{\mathcal{H}_K} \nonumber\\
&= \Big\langle g_{k-1}, k(\cdot,z) - \sum_{i=1}^m \xi_i k(\cdot,z_i) \Big\rangle_{\mathcal{H}_K}.\nonumber
\end{align}
Applying the Cauchy--Schwarz inequality in the RKHS immediately yields
\begin{equation}\label{eq:det_CS}
\bigl|\epsilon_{\mathrm{det}}(z)\bigr|
\le \|g_{k-1}\|_{\mathcal{H}_K} \Big\| k(\cdot,z) - \sum_{i=1}^m \xi_i k(\cdot,z_i) \Big\|_{\mathcal{H}_K}.
\end{equation}
To obtain the lemma statement, we now need to upper bound $\| k(\cdot,z) - \sum_{i=1}^m \xi_i k(\cdot,z_i) \|_{\mathcal{H}_K}$ by $\sigma_m(z)$. To this end, we use the following technical lemma:
\begin{lemma}\label{lem:norm_as_sup}
Let $\mathbf{g}_m = (g(z_1),\dots,g(z_m))^\top$ for any $g\in\mathcal{H}_K$. Then
\begin{equation}
\Big\| k(\cdot,z) - \sum_{i=1}^m \xi_i k(\cdot, z_i) \Big\|_{\mathcal{H}_K}^2 = \sup_{\|g\|_{\mathcal{H}_K} \leq 1} \Big( g(z) - \vec{\xi}_m^\top \mathbf{g}_m \Big)^2.
\end{equation}
\end{lemma}
\begin{proof}
Define
\begin{equation}
v \equiv k(\cdot,z) - \sum_{i=1}^m \xi_i k(\cdot, z_i).
\end{equation}
Importantly, note that $v\in \mathcal{H}_K.$ As such, by using the reproducing property of $\mathcal{H}_K$ one has
\begin{align}
\langle g, v\rangle_{\mathcal{H}_K} &= \langle g, k(\cdot, z)\rangle_{\mathcal{H}_K} -\sum_{i=1}^m\xi_i\langle g, k(\cdot, z_i)\rangle_{\mathcal{H}_K} \nonumber\\
&= g(z) - \sum_{i=1}^m\xi_i g(z_i)\nonumber\\
&= g(z) -  \vec{\xi}_m^\top \mathbf{g}_m, \nonumber
\end{align}
which implies
\begin{equation*}
\sup_{\|g\|_{\mathcal{H}_K} \leq 1} \Big( g(z) - \vec{\xi}_m^\top \mathbf{g}_m \Big)^2 = \sup_{\|g\|_{\mathcal{H}_K} \leq 1} (\langle g, v\rangle_{\mathcal{H}_K})^2.
\end{equation*}
Using the Cauchy--Schwarz inequality we have
\begin{equation}
\bigl| \langle g, v \rangle_{\mathcal{H}_K} \bigr| \leq \|g\|_{\mathcal{H}_K} \cdot \|v\|_{\mathcal{H}_K} ,
\label{eq:CS1}
\end{equation}
which implies
\begin{align}
\sup_{\|g\|_{\mathcal{H}_K} \leq 1} (\langle g, v\rangle_{\mathcal{H}_K})^2 &\leq \sup_{\|g\|_{\mathcal{H}_K} \leq 1} \|g\|^2_{\mathcal{H}_K} \cdot \|v\|^2_{\mathcal{H}_K} \label{eq:CS2}\\
&\leq \|v\|^2_{\mathcal{H}_K}. \nonumber
\end{align}
To finish the proof, we show that this upper bound is \textit{tight}. The inequality in Eq.~\eqref{eq:CS1} becomes equality when $g$ is parallel to $v$, and the supremum in Eq.~\eqref{eq:CS2} is attained when $\|g\|_{\mathcal{H}_K} = 1$. Both conditions are simultaneously satisfied by choosing
\begin{equation}
g^* = \frac{v}{\|v\|_{\mathcal{H}_K}} \quad \text{(assuming $v \neq 0$; the case $v=0$ is trivial)},
\end{equation}
which has $\|g^*\|_{\mathcal{H}_K} = 1$ and is exactly in the direction of $v$. Substituting gives
\begin{equation}
\bigl( \langle g^*, v \rangle_{\mathcal{H}_K} \bigr)^2 = \|v\|_{\mathcal{H}_K}^2,
\end{equation}
so the supremum equals $\|v\|_{\mathcal{H}_K}^2 = \| k(\cdot,z) - \sum_{i=1}^m \xi_i k(\cdot, z_i) \|_{\mathcal{H}_K}^2$.
\end{proof}
Now, with this in hand, we use the following characterization of $\sigma_m(z)$ from Ref.~\cite{Vakili2021a}:

\begin{lemma}[Uncertainty of the kernel ridge predictor (Proposition 1 in \cite{Vakili2021a})]
\label{lemma:vakili}
Let $\vec{\xi}_m = (K_m + \lambda^2 \mathds{1}_m)^{-1} \mathbf{K}_m(z) \in \mathbb{R}^m$ and $\mathbf{g}_m = (g(z_1),\dots,g(z_m))^\top$ for any $g\in\mathcal{H}_K$ as defined before. Then
\begin{equation}
\sigma_m^2(z) = \sup_{\|g\|_{\mathcal{H}_K}\le 1} \Bigl(g(z) - \vec{\xi}_m^\top \mathbf{g}_m\Bigr)^2 + \lambda^2 \|\vec{\xi}_m\|_2^2.
\end{equation}
\end{lemma}
In particular, using Lemma~\ref{lem:norm_as_sup} and noting that $\lambda^2 \|\vec{\xi}_m\|_2^2 \geq 0$, it follows that
\begin{equation}
\Big\| k(\cdot,z) - \sum_{i=1}^m \xi_i k(\cdot, z_i) \Big\|_{\mathcal{H}_K}^2 \leq \sigma_m^2(z).
\end{equation}
Substituting this inequality back into \eqref{eq:det_CS} gives the desired bound
\begin{equation}
\bigl|\epsilon_{\mathrm{det}}(z)\bigr|
\le \|g_{k-1}\|_{\mathcal{H}_K} \, \sigma_m(z).
\end{equation}
This completes the proof.
\end{proof}

The noise-induced error $\epsilon_\mathrm{noise}$ is then handled by a standard concentration result for sub-Gaussian variables. As with Lemma~\ref{lem:deterministic}, this statement is conditional on a fixed (but arbitrary) training set $\{z_i\}_{i=1}^m$: the probability below is taken only over the noise $\boldsymbol{\epsilon}_m$.

\begin{lemma}[Noise-induced error (fixed dataset)]
\label{lem:noise}
Under Observation~\ref{obs:subgaussian-noise}, namely that $\epsilon_1,\dots,\epsilon_m$ are independent $\tau^2$-sub-Gaussian random variables, for any fixed $z\in\mathcal{Z}$ and any fixed training set $\{z_i\}_{i=1}^m$, with probability at least $1-\delta_1$ over the draw of $\boldsymbol{\epsilon}_m$,
\begin{equation}
\bigl|\epsilon_{\mathrm{noise}}(z)\bigr|
\le \frac{\tau}{\lambda} \sigma_m(z) \sqrt{2\log\frac{2}{\delta_1}}.
\end{equation}
\end{lemma}
\begin{proof}
The noise-induced error at a fixed point $z\in\mathcal{Z}$ is the linear functional
\begin{equation}
\epsilon_{\mathrm{noise}}(z)
= \mathbf{K}_m(z)^\top (K_m + \lambda^2 \mathds{1}_m)^{-1} \boldsymbol{\epsilon}_m
= \vec{\xi}_m^\top \boldsymbol{\epsilon}_m = \sum_{i=1}^m \xi_i \epsilon_i,
\end{equation}
where $\vec{\xi}_m = (K_m + \lambda^2 \mathds{1}_m)^{-1} \mathbf{K}_m(z) \in \mathbb{R}^m$ is the same coefficient vector introduced in the proof of Lemma~\ref{lem:deterministic}.
Since the noise variables $\epsilon_1,\dots,\epsilon_m$ are independent and $\tau^2$-sub-Gaussian, $\epsilon_{\mathrm{noise}}(z)$ is itself sub-Gaussian with parameter $\tau^2 \|\vec{\xi}_m\|_2^2$. The Chernoff--Hoeffding tail bound for a zero-mean sub-Gaussian random variable \cite{subgaussian} then reads
\begin{equation}
\Pr\bigl[ |\epsilon_{\mathrm{noise}}(z)| \ge t \bigr]
\le 2 \exp\left( -\frac{t^2}{2 \tau^2 \|\vec{\xi}_m\|_2^2} \right)
\end{equation}
for all $t>0$. In order to have the right-hand side at most $\delta_1$, it is sufficient to take
\begin{equation}
t \geq \tau \|\vec{\xi}_m\|_2 \sqrt{2 \log \frac{2}{\delta_1}}.
\end{equation}
Hence, with probability at least $1-\delta_1$,
\begin{equation}
|\epsilon_{\mathrm{noise}}(z)|
\le \tau \|\vec{\xi}_m\|_2 \sqrt{2 \log \frac{2}{\delta_1}}.
\end{equation}
To prove the statement of the lemma, it remains to relate $\|\vec{\xi}_m\|_2$ to the predictive uncertainty $\sigma_m(z)$. From the characterization of $\sigma_m(z)$ in \Cref{lemma:vakili}, we have the decomposition
\begin{equation}
\sigma_m^2(z) = \sup_{\|g\|_{\mathcal{H}_K}\le 1} \bigl( g(z) - \vec{\xi}_m^\top \mathbf{g}_m \bigr)^2 + \lambda^2 \|\vec{\xi}_m\|_2^2.
\end{equation}
As the first term on the right-hand side is non-negative, this implies the elementary inequality
\begin{equation}
\lambda^2 \|\vec{\xi}_m\|_2^2 \le \sigma_m^2(z),
\qquad\text{or equivalently}\qquad
\|\vec{\xi}_m\|_2 \le \frac{\sigma_m(z)}{\lambda}.
\end{equation}
Substituting this bound into the concentration inequality finally yields
\begin{equation}
|\epsilon_{\mathrm{noise}}(z)|
\le \frac{\tau}{\lambda} \sigma_m(z) \sqrt{2\log\frac{2}{\delta_1}},
\end{equation}
which is the claimed statement.
\end{proof}

Now that we have an upper bound on both $|\epsilon_{\mathrm{det}}(z)|$ and $|\epsilon_{\mathrm{noise}}(z)|$, we can upper bound the right-hand side of Eq.~\eqref{eq:det-noise}. Applying the triangle inequality, and using both Lemma~\ref{lem:deterministic} and Lemma~\ref{lem:noise}, gives, for each fixed $z\in\mathcal{Z}$ and each fixed training set $\{z_i\}_{i=1}^m$,
\begin{equation}\label{eq:pointwise_bound_main}
\Pr_{\boldsymbol{\epsilon}_m}\!\left[
|g_{k-1}(z) - F_k(z)|
\le \sigma_m(z) \left( \|g_{k-1}\|_{\mathcal{H}_K} + \frac{\tau}{\lambda} \sqrt{2\log\frac{2}{\delta_1}} \right)
\right] \ge 1-\delta_1.
\end{equation}
Since $\hat{Q}_k=\Pi_{[-V_{\max},V_{\max}]}F_k$ and $g_{k-1}(z)\in[-V_{\max},V_{\max}]$, projection onto this interval cannot increase pointwise distance to any point already inside the interval, so
\begin{equation}
|g_{k-1}(z)-\hat{Q}_k(z)| \le |g_{k-1}(z)-F_k(z)|.
\end{equation}
The bound in Eq.~\eqref{eq:pointwise_bound_main} therefore transfers directly from the unclipped RKHS predictor $F_k$ to the clipped iterate $\hat{Q}_k$, and we use this fact silently from here on.

To go from the pointwise bound~\eqref{eq:pointwise_bound_main} to a bound on the integrated error $\mathbb{E}_{z\sim U(\mathcal{Z})}[|g_{k-1}(z)-\hat{Q}_k(z)|^2]$ (still for a fixed training set), we first establish two auxiliary results on Lipschitz continuity and discretization. The key challenge is that the pointwise bound holds only at individual points $z$, while the integrated error requires control over the full continuous state-action space. We address this by introducing a finite covering net and extending the bound to all points via Lipschitz continuity.

\begin{lemma}[Lipschitz continuity]\label{lemm:lipschitz}
For normalized kernels of the form
\begin{equation}
K_{(\mathcal{D},w)}(x,x') = \frac{1}{\|w\|_2^2} \sum_{i=0}^M w_i^2 \cos\bigl(\langle \omega_i, x - x'\rangle\bigr),
\end{equation}
every function $f \in \mathcal{H}_K$ (with arbitrary RKHS norm $\|f\|_{\mathcal{H}_K}$) is Lipschitz continuous with constant at most $L_{\mathcal{D}} \|f\|_{\mathcal{H}_K}$, where
\begin{equation}
L_{\mathcal{D}} \equiv \max_{0 \le i \le M} \|\omega_i\|_2.
\end{equation}
\end{lemma}

\begin{proof}
Let $\nu_i = w_i^2 / \|w\|_2^2$ (so $\sum_i \nu_i = 1$). For any $f \in \mathcal{H}_K$ the reproducing property yields
\begin{equation}
|f(z) - f(z')| \le |\langle f,K(\cdot,z)-K(\cdot,z')\rangle|\leq \|f\|_{\mathcal{H}_K}\|K(\cdot,z)-K(\cdot,z')\|_{\mathcal{H}_K},
\end{equation}
where we used Cauchy--Schwarz. Expanding
\begin{equation}
\|K(\cdot,z)-K(\cdot,z')\|_{\mathcal{H}_K}^2 = K(z,z) + K(z',z') - 2K(z,z')
\end{equation}
and using that the kernel is normalized yields
\begin{equation}
|f(z) - f(z')| \leq \|f\|_{\mathcal{H}_K} \sqrt{2\bigl(1 - K(z,z')\bigr)}.
\end{equation}
Using the elementary inequality $1 - \cos\theta \le \theta^2/2$ and the fact that $\sum_i \nu_i=1$, we obtain
\begin{align}
1 - K(z,z') &= \sum_{i=0}^M \nu_i (1 - \cos\langle\omega_i,z-z'\rangle) \nonumber\\
&\le \frac{1}{2} \sum_{i=0}^M \nu_i \langle\omega_i,z-z'\rangle^2 \nonumber\\
&\le\frac{1}{2}\sum_{i=0}^M \nu_i\|\omega_i\|_2^2\|z-z'\|_2^2 \nonumber\\
&\le\frac{1}{2} \|z-z'\|_2^2 L_{\mathcal{D}}^2, \nonumber
\end{align}
where the last step uses $\sum_i \nu_i = 1$. Hence
\begin{equation}
\sqrt{2(1-K(z,z'))} \le L_{\mathcal{D}} \|z-z'\|_2,
\end{equation}
which proves the Lipschitz claim.
\end{proof}

Since the state space $\mathcal{S}$ is a continuous subset of $\mathbb{R}^{d_{\mathcal{S}}}$, we cannot directly apply a union bound over all $z\in\mathcal{Z}$. Instead, we introduce a finite covering net (discretization) of $\mathcal{Z}$ and use Lipschitz continuity to control the error at points between grid nodes.

\begin{lemma}[Discretization property]\label{lem:discretization}
Let $f,g\in\mathcal{H}_K$ with $C_f \equiv \|f\|_{\mathcal{H}_K}$ and $C_g\equiv\|g\|_{\mathcal{H}_K}$. Assume that the state space $\mathcal{S}\subset\mathbb{R}^{d_{\mathcal{S}}}$ is compact with Euclidean diameter $D=\sup_{s,s'\in\mathcal{S}}\|s-s'\|_2$. For any $\varepsilon>0$ there exists a finite discretization $\mathcal{Z}_\varepsilon\subseteq\mathcal{Z}$ (where $\mathcal{Z}=\mathcal{S}\times\mathcal{A}$) such that, for any $z\in\mathcal{Z}$,
\begin{equation}
|f(z) - f([z])| \le \varepsilon, \qquad |g(z) - g([z])| \le \varepsilon,
\end{equation}
where $[z]$ is the closest point in $\mathcal{Z}_\varepsilon$ to $z$ (with respect to the Euclidean norm), and
\begin{equation}
|\mathcal{Z}_\varepsilon| \le |\mathcal{A}|\Bigl( D L_{\mathcal{D}} \max(C_f, C_g) / \varepsilon + 1 \Bigr)^{d_{\mathcal{S}}}.
\end{equation}
The constant $L_{\mathcal{D}}$ is the Lipschitz constant from the previous lemma.
\end{lemma}

\begin{proof}
By the Lipschitz lemma, it suffices to ensure $\|z-[z]\|_2 \le \varepsilon_{\mathrm{grid}}$ for all $z\in\mathcal{Z}$, where
\begin{equation}
\varepsilon_{\mathrm{grid}} \equiv \frac{\varepsilon}{L_{\mathcal{D}} \max(C_f,C_g)}.
\end{equation}
Let $\mathcal{S}_\varepsilon$ be any $\varepsilon_{\mathrm{grid}}$-net of $\mathcal{S}$ (w.r.t.\ the Euclidean norm). Then
\begin{equation}
|\mathcal{S}_\varepsilon| \le \Bigl( \frac{D}{\varepsilon_{\mathrm{grid}}} + 1 \Bigr)^{d_{\mathcal{S}}} = \Bigl( D L_{\mathcal{D}} \max(C_f,C_g) / \varepsilon + 1 \Bigr)^{d_{\mathcal{S}}}.
\end{equation}
Define $\mathcal{Z}_\varepsilon \equiv \mathcal{S}_\varepsilon \times \mathcal{A}$. Since $|\mathcal{A}|$ is finite, we obtain the claimed cardinality bound. For the closest point we have $\|z-[z]\|_2\le\varepsilon_{\mathrm{grid}}$, and the desired approximation errors follow immediately from the Lipschitz property.
\end{proof}

Since $\mathcal{S}\subseteq[0,1]^{d_{\mathcal{S}}}$, the Euclidean diameter satisfies $D\le\sqrt{d_{\mathcal{S}}}$. We now apply \Cref{lem:discretization} with $\varepsilon=1/m$ to construct a covering net $\mathcal{Z}_m$. A subtlety arises here: for the union bound over grid points to be valid, the net must be fixed \emph{before} the noise $\boldsymbol{\epsilon}_m$ is drawn. If we were to use $C_g=\|F_k\|_{\mathcal{H}_K}$ as in the discretization lemma, the net would depend on the noise realization (since $F_k$ depends on $\boldsymbol{\epsilon}_m$), invalidating the union bound. We therefore replace $\max(C_f,C_g)$ by a deterministic upper bound. Since $F_k$ is the KRR minimizer of $\sum_i(y_i-F_k(z_i))^2+\lambda^2\|F_k\|_{\mathcal{H}_K}^2$, we have $\lambda^2\|F_k\|_{\mathcal{H}_K}^2\le\sum_i y_i^2\le m V_{\max}^2$, so $\|F_k\|_{\mathcal{H}_K}\le\sqrt{m}\,V_{\max}/\lambda$. Define
\begin{equation}
G_K \equiv \max_{1\le k\le K}\|T\hat{Q}_{k-1}\|_{\mathcal{H}_K},
\end{equation}
which is a deterministic quantity depending only on the MDP and the kernel (not on the noise), and set
\begin{equation}
B_{\mathrm{net}}(m) \equiv \max\!\left\{G_K,\; \frac{\sqrt{m}\,V_{\max}}{\lambda}\right\}.
\end{equation}
We then apply Lemma~\ref{lem:discretization} with the pair $f=g_{k-1}$, $g=F_k$ and the deterministic RKHS-norm envelope $B_{\mathrm{net}}(m)\ge\max(\|g_{k-1}\|_{\mathcal{H}_K},\|F_k\|_{\mathcal{H}_K})$, which ensures the net is fixed independently of the noise. The corresponding grid radius is
\begin{equation}
\varepsilon_{\mathrm{grid}} = \frac{1}{m\,L_{\mathcal{D}}\,B_{\mathrm{net}}(m)},
\end{equation}
and the net cardinality satisfies
\begin{equation}
|\mathcal{Z}_m| \;\le\; |\mathcal{A}|\bigl(\sqrt{d_{\mathcal{S}}}\,L_{\mathcal{D}}\,B_{\mathrm{net}}(m)\,m + 1\bigr)^{d_{\mathcal{S}}}.
\end{equation}
For every fixed $z'\in\mathcal{Z}_m$ and every fixed training set, the pointwise bound~\eqref{eq:pointwise_bound_main} gives
\begin{equation}\label{eq:pointwise_grid}
|g_{k-1}(z') - F_k(z')|
\le \sigma_m(z') \left( \|g_{k-1}\|_{\mathcal{H}_K} + \frac{\tau}{\lambda} \sqrt{2\log\frac{2}{\delta'}} \right)
\end{equation}
with probability at least $1-\delta'$ over the noise. Choosing $\delta'=\delta_1/|\mathcal{Z}_m|$ and applying a union bound over the finite set $\mathcal{Z}_m$, we obtain that, with probability at least $1-\delta_1$ over the noise (for every fixed training set),
\begin{equation}\label{eq:cond_uniform}
|g_{k-1}(z') - F_k(z')|
\le \sigma_m(z') \Bigl( \|g_{k-1}\|_{\mathcal{H}_K} + \frac{\tau}{\lambda} \sqrt{2\log\frac{2|\mathcal{Z}_m|}{\delta_1}} \Bigr)
\end{equation}
holds simultaneously for every $z'\in\mathcal{Z}_m$.

We now extend this bound from the grid $\mathcal{Z}_m$ to the full continuous space $\mathcal{Z}$. For any $z\in\mathcal{Z}$, let $[z]\in\mathcal{Z}_m$ denote its nearest grid point. By the triangle inequality,
\begin{equation}
|g_{k-1}(z) - F_k(z)|\leq|g_{k-1}(z)-g_{k-1}([z])|+|F_k(z)-F_k([z])|+|g_{k-1}([z])-F_k([z])|.
\end{equation}
By construction of the grid, the first two terms satisfy $|g_{k-1}(z)-g_{k-1}([z])|\le 1/m$ and $|F_k(z)-F_k([z])|\le 1/m$ (from Lemma~\ref{lem:discretization} with $\varepsilon=1/m$). The third term is bounded by~\eqref{eq:cond_uniform} evaluated at $[z]$. Combining, the following holds with probability at least $1-\delta_1$ over the noise, uniformly for all $z\in\mathcal{Z}$:
\begin{equation}\label{eq:uniform_all_z}
|g_{k-1}(z) - F_k(z)|\leq \sigma_m([z]) \Bigl( \|g_{k-1}\|_{\mathcal{H}_K} + \frac{\tau}{\lambda} \sqrt{2\log\frac{2|\mathcal{Z}_m|}{\delta_1}} \Bigr)+\frac{2}{m}.
\end{equation}

We now need to pass from $\sigma_m([z])$ to $\sigma_m(z)$ in order to compute the expectation over $z\sim U(\mathcal{Z})$. This requires controlling the regularity of the predictive uncertainty, which we establish in the following lemma.

\begin{lemma}[Lipschitz continuity of the predictive variance]\label{lem:lipschitz_sigma}
For any fixed training set $\{z_i\}_{i=1}^m$ and any $z,z'\in\mathcal{Z}$,
\begin{equation}
\bigl|\sigma_m^2(z)-\sigma_m^2(z')\bigr| \le 2L_{\mathcal{D}}\,\|z-z'\|_2.
\end{equation}
In particular, since $\sigma_m(z)\ge 0$ for all $z$, the identity $|a-b|\le\sqrt{|a^2-b^2|}$ for $a,b\ge 0$ implies
\begin{equation}
\bigl|\sigma_m(z)-\sigma_m(z')\bigr| \le \sqrt{2L_{\mathcal{D}}\,\|z-z'\|_2}.
\end{equation}
\end{lemma}

\begin{proof}
The following derivation uses standard operator-theoretic tools from the kernel methods literature (see, e.g., \cite{Vakili2021a}) to rewrite the predictive variance in a form amenable to Lipschitz analysis. The definition of the predictive variance involves the finite-dimensional Gram matrix $K_m\in\mathbb{R}^{m\times m}$ and the kernel vector $\mathbf{K}_m(z)\in\mathbb{R}^m$. To prove a Lipschitz bound, it will be convenient to lift this expression to an operator-theoretic form in the infinite-dimensional RKHS $\mathcal{H}_K$. This is achieved via the \emph{empirical covariance operator}
\begin{equation}
C_m \equiv \sum_{i=1}^m K(\cdot,z_i)\otimes K(\cdot,z_i),
\end{equation}
which is a positive semi-definite self-adjoint operator on $\mathcal{H}_K$. Here, for any $f,g\in\mathcal{H}_K$, the tensor product $(f\otimes g)$ denotes the rank-one operator that acts on $h\in\mathcal{H}_K$ via $(f\otimes g)(h) = \langle g,h\rangle_{\mathcal{H}_K} f$. The operator $C_m$ encodes the same information as the Gram matrix $K_m$: specifically, its action on any kernel section $K(\cdot,z)$ is
\begin{equation}
C_m\, K(\cdot,z) = \sum_{i=1}^m K(\cdot,z_i)\,K(z_i,z) = \sum_{i=1}^m [\mathbf{K}_m(z)]_i\, K(\cdot,z_i),
\end{equation}
so that $C_m$ acts in the span of the kernel sections $\{K(\cdot,z_i)\}_{i=1}^m$ in the same way that $K_m$ acts in $\mathbb{R}^m$.

We now explain how the predictive variance $\sigma_m^2(z)$ can be rewritten in terms of $C_m$. Recall the definition:
\begin{equation}
\sigma_m^2(z) = K(z,z) - \mathbf{K}_m(z)^\top (K_m + \lambda^2 \mathds{1}_m)^{-1}\mathbf{K}_m(z).
\end{equation}
The key observation is the \emph{push-through identity} (also called the Woodbury or matrix-inversion identity). In the finite-dimensional setting it states that for any matrix $A\in\mathbb{R}^{m\times m}$ and vector $\mathbf{v}\in\mathbb{R}^m$,
\begin{equation}
\mathbf{v}^\top(A + \lambda^2 \mathds{1}_m)^{-1}\mathbf{v} = \frac{1}{\lambda^2}\bigl(\mathbf{v}^\top \mathbf{v} - \mathbf{v}^\top(A + \lambda^2\mathds{1}_m)^{-1}A\,\mathbf{v}\bigr).
\end{equation}
The infinite-dimensional analog, applied in $\mathcal{H}_K$ with the operator $C_m$ playing the role of $A$ and the element $K(\cdot,z)$ playing the role of $\mathbf{v}$ (see, e.g., \cite{Vakili2021a}), yields
\begin{equation}\label{eq:push_through}
\sigma_m^2(z) = \lambda^2 \bigl\langle K(\cdot,z),\, (C_m + \lambda^2 I)^{-1} K(\cdot,z)\bigr\rangle_{\mathcal{H}_K},
\end{equation}
where $I$ denotes the identity operator on $\mathcal{H}_K$. To verify this, note that the right-hand side equals
\begin{align}
\lambda^2 \bigl\langle K(\cdot,z),\, (C_m + \lambda^2 I)^{-1} K(\cdot,z)\bigr\rangle_{\mathcal{H}_K}
&= \bigl\langle K(\cdot,z),\, K(\cdot,z)\bigr\rangle_{\mathcal{H}_K} - \bigl\langle K(\cdot,z),\, C_m(C_m + \lambda^2 I)^{-1} K(\cdot,z)\bigr\rangle_{\mathcal{H}_K}\nonumber\\
&= K(z,z) - \bigl\langle K(\cdot,z),\, C_m(C_m + \lambda^2 I)^{-1} K(\cdot,z)\bigr\rangle_{\mathcal{H}_K},\nonumber
\end{align}
where we used $\lambda^2(C_m+\lambda^2 I)^{-1} = I - C_m(C_m+\lambda^2 I)^{-1}$ and the reproducing property. The second term, when restricted to the span of $\{K(\cdot,z_i)\}$, reduces exactly to $\mathbf{K}_m(z)^\top(K_m+\lambda^2\mathds{1}_m)^{-1}\mathbf{K}_m(z)$, confirming~\eqref{eq:push_through}.

Since $(C_m+\lambda^2 I)$ is strictly positive (as $C_m\ge 0$ and $\lambda^2>0$), it admits a unique positive self-adjoint square root $(C_m+\lambda^2 I)^{1/2}$, and we can write
\begin{equation}
\sigma_m^2(z) = \lambda^2 \bigl\|(C_m + \lambda^2 I)^{-1/2} K(\cdot,z)\bigr\|_{\mathcal{H}_K}^2.
\end{equation}

With this representation in hand, the Lipschitz bound follows from elementary Hilbert-space manipulations. Define
\begin{equation}
u \equiv (C_m+\lambda^2 I)^{-1/2}K(\cdot,z), \qquad v \equiv (C_m+\lambda^2 I)^{-1/2}K(\cdot,z').
\end{equation}
Then $\sigma_m^2(z)=\lambda^2\|u\|_{\mathcal{H}_K}^2$ and $\sigma_m^2(z')=\lambda^2\|v\|_{\mathcal{H}_K}^2$, so using the factorization $\|u\|^2-\|v\|^2 = \langle u+v,\,u-v\rangle$ we get
\begin{equation}
\bigl|\sigma_m^2(z)-\sigma_m^2(z')\bigr| = \lambda^2\bigl|\|u\|_{\mathcal{H}_K}^2 - \|v\|_{\mathcal{H}_K}^2\bigr| = \lambda^2 \bigl|\langle u+v,\, u-v\rangle_{\mathcal{H}_K}\bigr| \le \lambda^2 \|u+v\|_{\mathcal{H}_K}\,\|u-v\|_{\mathcal{H}_K},
\end{equation}
where the last step is Cauchy--Schwarz. We now bound each factor separately.

\medskip
\noindent\textbf{Bounding $\|u+v\|_{\mathcal{H}_K}$:} Since the kernel is normalized, $\|K(\cdot,z)\|_{\mathcal{H}_K}=\sqrt{K(z,z)}=1$ for all $z$. The operator $(C_m+\lambda^2 I)^{-1/2}$ has operator norm
\begin{equation}
\bigl\|(C_m+\lambda^2 I)^{-1/2}\bigr\|_{\mathrm{op}} = \frac{1}{\sqrt{\lambda_{\min}(C_m+\lambda^2 I)}} \le \frac{1}{\sqrt{\lambda^2}} = \frac{1}{\lambda},
\end{equation}
where $\lambda_{\min}(C_m+\lambda^2 I)\ge \lambda^2$ because $C_m$ is positive semi-definite. Therefore
\begin{equation}
\|u\|_{\mathcal{H}_K} = \bigl\|(C_m+\lambda^2 I)^{-1/2}K(\cdot,z)\bigr\|_{\mathcal{H}_K} \le \frac{1}{\lambda}\|K(\cdot,z)\|_{\mathcal{H}_K} = \frac{1}{\lambda},
\end{equation}
and similarly $\|v\|_{\mathcal{H}_K}\le 1/\lambda$. Hence $\|u+v\|_{\mathcal{H}_K} \le 2/\lambda$.

\medskip
\noindent\textbf{Bounding $\|u-v\|_{\mathcal{H}_K}$:} By linearity of $(C_m+\lambda^2 I)^{-1/2}$,
\begin{equation}
\|u-v\|_{\mathcal{H}_K} = \bigl\|(C_m+\lambda^2 I)^{-1/2}\bigl(K(\cdot,z)-K(\cdot,z')\bigr)\bigr\|_{\mathcal{H}_K} \le \frac{1}{\lambda}\bigl\|K(\cdot,z)-K(\cdot,z')\bigr\|_{\mathcal{H}_K},
\end{equation}
using the same operator-norm bound. From the proof of Lemma~\ref{lemm:lipschitz}, we established
\begin{equation}
\bigl\|K(\cdot,z)-K(\cdot,z')\bigr\|_{\mathcal{H}_K} = \sqrt{2(1-K(z,z'))} \le L_{\mathcal{D}}\|z-z'\|_2.
\end{equation}

\medskip
\noindent\textbf{Combining:} Substituting the two bounds above,
\begin{equation}
\bigl|\sigma_m^2(z)-\sigma_m^2(z')\bigr| \le \lambda^2 \cdot \frac{2}{\lambda} \cdot \frac{L_{\mathcal{D}}\|z-z'\|_2}{\lambda} = 2L_{\mathcal{D}}\,\|z-z'\|_2.
\end{equation}
For the bound on $\sigma_m$ itself, note that for $a,b\ge 0$ we have $(a-b)^2 \le |a-b|(a+b) = |a^2-b^2|$, hence $|a-b|\le\sqrt{|a^2-b^2|}$. Applying this with $a=\sigma_m(z)$ and $b=\sigma_m(z')$ gives
\begin{equation}
\bigl|\sigma_m(z)-\sigma_m(z')\bigr| \le \sqrt{2L_{\mathcal{D}}\,\|z-z'\|_2}.\qedhere
\end{equation}
\end{proof}

Applying Lemma~\ref{lem:lipschitz_sigma} with $\|z-[z]\|_2\le \varepsilon_{\mathrm{grid}} = \frac{1}{m\,L_{\mathcal{D}}\,B_{\mathrm{net}}(m)}$ gives
\begin{align}
\bigl|\sigma_m^2(z)-\sigma_m^2([z])\bigr| &\le \frac{2}{m\,B_{\mathrm{net}}(m)}, \label{eq:sigma2_correction}\\
\bigl|\sigma_m(z)-\sigma_m([z])\bigr| &\le \sqrt{\frac{2}{m\,B_{\mathrm{net}}(m)}}. \label{eq:sigma_correction}
\end{align}
In particular, using $\sigma_m([z])\le\sigma_m(z)+|\sigma_m(z)-\sigma_m([z])|$ and $\sigma_m([z])^2\le\sigma_m(z)^2+|\sigma_m^2(z)-\sigma_m^2([z])|$, we can replace $\sigma_m([z])$ by $\sigma_m(z)$ in~\eqref{eq:uniform_all_z} at the cost of the above corrections.

\paragraph{Two sources of randomness.} Before proceeding, we emphasize that the bound~\eqref{eq:uniform_all_z} is conditional on the training set: it holds with probability $1-\delta_1$ over the noise, for an \emph{arbitrary fixed} draw of $\{z_i\}_{i=1}^m$. However, $\sigma_m(z)$ on the right-hand side is itself a random variable through its dependence on the training set. We now address the training-set randomness. The remainder of the argument involves two independent sources of randomness:
\begin{itemize}
  \item \textbf{Training-set randomness:} the $m$ points $\{z_i\}_{i=1}^m$ are drawn i.i.d.\ from $U(\mathcal{Z})$. The predictive uncertainty $\sigma_m(z)$ depends on these points and is therefore a random function of the training set.
  \item \textbf{Noise randomness:} the labels $\boldsymbol{\epsilon}_m$ are drawn independently of the training set; this is the randomness already handled in Lemma~\ref{lem:noise} and the union bound above.
\end{itemize}

The following lemma provides a bound on the expected integrated uncertainty, exploiting the fact that the training data are drawn from the same distribution over which the uncertainty is averaged.

\begin{lemma}[Expected integrated uncertainty]\label{lem:info_gain}
Under Assumption~\ref{ass:sampling-model}, when $\{z_i\}_{i=1}^m \overset{\mathrm{iid}}{\sim} U(\mathcal{Z})$,
\begin{equation}
\mathbb{E}_{\{z_i\}}\!\left[\mathbb{E}_{z\sim U(\mathcal{Z})}[\sigma_m^2(z)]\right]
\;\le\;
\frac{2\,\Gamma(m+1)}{(m+1)\,\log(1+1/\lambda^2)},
\end{equation}
where $\Gamma(n)=\sup_{z_1,\dots,z_n\in\mathcal{Z}}\frac{1}{2}\log\det\bigl(\mathds{1}_n+K_n/\lambda^2\bigr)$ is the maximum information gain of the kernel after $n$ observations.
\end{lemma}

\begin{proof}
Define the sequence
\begin{equation}
a_t \;\equiv\; \mathbb{E}_{z_1,\dots,z_{t-1}}\!\left[\mathbb{E}_{z_t\sim U(\mathcal{Z})}\bigl[\sigma^2_{t-1}(z_t)\bigr]\right], \qquad t=1,2,\dots,
\end{equation}
where $\sigma^2_{t-1}(z_t)$ is the predictive variance at a fresh point $z_t\sim U(\mathcal{Z})$ given the training set $\{z_1,\dots,z_{t-1}\}$. The left-hand side of the lemma is precisely $a_{m+1}$, since it conditions on $m$ training points and evaluates the predictive variance at a fresh $(m+1)$-th point drawn from $U(\mathcal{Z})$.

The sequence $(a_t)_{t\ge 1}$ is non-increasing: for any $t' \ge t$, the predictive variance $\sigma_{t'-1}^2(z)$ is at most $\sigma_{t-1}^2(z)$ pointwise, because conditioning on additional observations can only decrease the predictive variance of the kernel ridge predictor (this follows from the positive semi-definiteness of $C_m$ and the monotonicity of the map $C\mapsto (C+\lambda^2 I)^{-1}$ on the positive cone). Therefore $a_{m+1}\le a_t$ for every $t\le m+1$. By exchangeability of the i.i.d.\ sequence $(z_1,\dots,z_{m+1})$, we have $\mathbb{E}_{S_{t-1},z_t}[\sigma_{t-1}^2(z_t)] = \mathbb{E}_{S_{t-1},z}[\sigma_{t-1}^2(z)]$ for any fresh $z\sim U(\mathcal{Z})$ independent of $S_{t-1}$, confirming that $a_t$ is well-defined regardless of which point in the sequence is designated as the test point. Since $a_{m+1}\le a_t$ for every $t\in\{1,\dots,m+1\}$, it is in particular bounded by the average:
\begin{equation}
a_{m+1} \;\le\; \frac{1}{m+1}\sum_{t=1}^{m+1} a_t.
\end{equation}
It remains to bound this sum. By Fubini's theorem,
\begin{equation}
\sum_{t=1}^{m+1} a_t \;=\; \sum_{t=1}^{m+1} \mathbb{E}_{z_1,\dots,z_t}\!\bigl[\sigma^2_{t-1}(z_t)\bigr]
\;=\; \mathbb{E}_{z_1,\dots,z_{m+1}}\!\left[\sum_{t=1}^{m+1}\sigma^2_{t-1}(z_t)\right].
\end{equation}
We now recall the following standard bound on cumulative predictive variances.

\begin{lemma}[Cumulative uncertainty bound {\cite[Lemma~5.3]{Srinivas}}]
\label{lem:srinivas}
For any sequence of points $\{z_l\}_{l=1}^n\subset\mathcal{Z}$,
\begin{equation}
\sum_{l=1}^n \sigma_{l-1}^2(z_l)
\le
\frac{2\,\Gamma(n)}{\log(1+1/\lambda^2)}.
\end{equation}
\end{lemma}

Applying Lemma~\ref{lem:srinivas} with $n=m+1$ inside the expectation (the bound holds for every realization of $z_1,\dots,z_{m+1}$, so it passes through the expectation),
\begin{equation}
\sum_{t=1}^{m+1} a_t \;\le\; \frac{2\,\Gamma(m+1)}{\log(1+1/\lambda^2)}.
\end{equation}
Combining the above,
\begin{equation}
\mathbb{E}_{\{z_i\}}\!\left[\mathbb{E}_{z\sim U(\mathcal{Z})}[\sigma_m^2(z)]\right]
\;=\; a_{m+1} \;\le\; \frac{1}{m+1}\sum_{t=1}^{m+1} a_t \;\le\; \frac{2\,\Gamma(m+1)}{(m+1)\,\log(1+1/\lambda^2)},
\end{equation}
which completes the proof.
\end{proof}

\paragraph{Combining the bounds.} We now assemble the pieces. Define the abbreviation
\begin{equation}
\mathcal{F}_m \;\equiv\; \|g_{k-1}\|_{\mathcal{H}_K} + \frac{\tau}{\lambda} \sqrt{2\log\frac{4K|\mathcal{Z}_m|}{\delta}},
\end{equation}
where the factor $4K$ inside the logarithm accounts for the union bound over grid points and over the $K$ iterations of the algorithm (as explained below). Starting from~\eqref{eq:uniform_all_z} and using $\sigma_m([z])\le\sigma_m(z)+|\sigma_m(z)-\sigma_m([z])|\le \sigma_m(z)+\sqrt{\frac{2}{m\,B_{\mathrm{net}}(m)}}$ (from~\eqref{eq:sigma_correction}), we have --- with probability at least $1-\delta/(2K)$ over the noise, for any fixed training set --- the pointwise bound
\begin{equation}
|g_{k-1}(z) - F_k(z)| \;\le\; \left(\sigma_m(z)+\sqrt{\frac{2}{m\,B_{\mathrm{net}}(m)}}\right)\mathcal{F}_m \;+\; \frac{2}{m},
\end{equation}
valid for all $z\in\mathcal{Z}$ simultaneously. Squaring, taking the expectation over $z\sim U(\mathcal{Z})$, and using the elementary inequality $(a+b+c)^2\le 3(a^2+b^2+c^2)$, we obtain
\begin{equation}\label{eq:integrated_bound}
\mathbb{E}_{z\sim U(\mathcal{Z})}\!\left[|g_{k-1}(z)-F_k(z)|^2\right]
\;\le\; 3\,\mathbb{E}_{z\sim U(\mathcal{Z})}[\sigma_m^2(z)]\,\mathcal{F}_m^2 \;+\; \frac{6\,\mathcal{F}_m^2}{m\,B_{\mathrm{net}}(m)} \;+\; \frac{12}{m^2}.
\end{equation}
This bound holds at each fixed iteration $k$ with probability at least $1-\delta/(2K)$ over the noise, for any fixed training set. The quantity $\mathbb{E}_{z}[\sigma_m^2(z)]$ on the right-hand side remains a random variable through the training set.

We now convert the in-expectation bound from Lemma~\ref{lem:info_gain} into a high-probability statement via Markov's inequality. Define the non-negative random variable $Y\equiv\mathbb{E}_{z\sim U(\mathcal{Z})}[\sigma_m^2(z)]$, whose randomness comes entirely from the draw of the training set $\{z_i\}_{i=1}^m$. By Lemma~\ref{lem:info_gain},
\begin{equation}
\mathbb{E}_{\{z_i\}}[Y] \;\le\; \frac{2\,\Gamma(m+1)}{(m+1)\,\log(1+1/\lambda^2)}.
\end{equation}
Since $Y\ge 0$, Markov's inequality gives $\Pr[Y\ge t]\le \mathbb{E}[Y]/t$ for any $t>0$. Setting the failure probability to $\delta/(2K)$, to account for the union bound over $K$ iterations, yields: with probability at least $1-\delta/(2K)$ over the draw of the training set at each iteration,
\begin{equation}\label{eq:markov_sigma}
\mathbb{E}_{z\sim U(\mathcal{Z})}[\sigma_m^2(z)] \;\le\; \frac{4K\,\Gamma(m+1)}{\delta\,(m+1)\,\log(1+1/\lambda^2)}.
\end{equation}

To make the conditional-to-unconditional step explicit, define $\mathcal{B}_k$ as the round-$k$ bad event (i.e., at least one of the two round-$k$ bounds fails). From the conditional analysis at round $k$,
\begin{equation}
\Pr(\mathcal{B}_k\mid \mathcal{F}_{k-1}) \le \frac{\delta}{K} 
\end{equation}
The tower property (law of total expectation for indicators) gives
\begin{equation}
\Pr(\mathcal{B}_k)=\mathbb{E}\!\left[\Pr(\mathcal{B}_k\mid \mathcal{F}_{k-1})\right] \le \frac{\delta}{K}.
\end{equation}
Now apply the union bound over iterations:
\begin{equation}
\Pr\!\left(\bigcup_{k=1}^K \mathcal{B}_k\right) \le \sum_{k=1}^K \Pr(\mathcal{B}_k) \le \delta,
\end{equation}
which is exactly why the final statement holds simultaneously for all $k=1,\dots,K$.
We now combine both sources of randomness. At each iteration $k$, we have two independent high-probability events: the noise bound~\eqref{eq:integrated_bound} (holding with probability $\ge 1-\delta/(2K)$) and the Markov bound~\eqref{eq:markov_sigma} (holding with probability $\ge 1-\delta/(2K)$). By a union bound, both hold simultaneously at a given iteration with probability at least $1-\delta/K$. Since there are $K$ iterations in total, a further union bound over all iterations gives that the bound holds simultaneously at every $k\in\{1,\dots,K\}$ with probability at least $1-\delta$ over the joint randomness. Substituting~\eqref{eq:markov_sigma} into~\eqref{eq:integrated_bound}, we obtain that with probability at least $1-\delta$,
\begin{equation}\label{eq:full_bound}
\max_{1\le k\le K}\,\mathbb{E}_{z\sim U(\mathcal{Z})}\!\left[|g_{k-1}(z)-F_k(z)|^2\right]
\;\le\; \frac{12K\,\Gamma(m+1)\,\mathcal{F}_m^2}{\delta\,(m+1)\,\log(1+1/\lambda^2)} \;+\; \frac{6\,\mathcal{F}_m^2}{m\,B_{\mathrm{net}}(m)} \;+\; \frac{12}{m^2}.
\end{equation}
We note that the right-hand side does not depend on the iteration index $k$: the bound is uniform across iterations, and the maximum on the left-hand side is controlled by the same expression at every $k$. The factor $K$ in the leading term arises solely from the union bound over iterations. Since the clipping bound $|g_{k-1}(z)-\hat{Q}_k(z)|\le|g_{k-1}(z)-F_k(z)|$ holds pointwise, the same bound applies with $F_k$ replaced by $\hat{Q}_k$.

Returning to the error decomposition in Theorem~\ref{th:fqi_error_propagation} with $\phi_{U,U}\le\sqrt{|\mathcal{A}|\,p_{\max}}$ (Lemma~\ref{lem:concentrability-UU}), our strategy is to ensure that each of the two terms in~\eqref{eq:error_decomposition_main} is at most $\epsilon/2$. The algorithmic term is handled by condition~\eqref{eq:K_condition} on the number of iterations $K$. For the statistical term, we require
\begin{equation}
\frac{2\phi_{U,U}\gamma}{(1-\gamma)^2}\max_{1\le k\le K}\left(\mathbb{E}_{z\sim U(\mathcal{Z})}\!\left[|\hat{Q}_k(z)-T\hat{Q}_{k-1}(z)|^2\right]\right)^{1/2} \;\le\; \frac{\epsilon}{2}.
\end{equation}
Squaring both sides and rearranging, this is equivalent to requiring
\begin{equation}\label{eq:target_bound}
\max_{1\le k\le K}\,\mathbb{E}_{z\sim U(\mathcal{Z})}\!\left[|\hat{Q}_k(z)-T\hat{Q}_{k-1}(z)|^2\right] \;\le\; \frac{\epsilon^2(1-\gamma)^4}{16\,\phi_{U,U}^2\,\gamma^2} \;\equiv\; \varepsilon_{\mathrm{stat}}^2,
\end{equation}
where we define the target threshold
\begin{equation}
\varepsilon_{\mathrm{stat}} \equiv \frac{\epsilon(1-\gamma)^2}{4\gamma\sqrt{|\mathcal{A}|\,p_{\max}}}.
\end{equation}
It therefore suffices to ensure that the right-hand side of~\eqref{eq:full_bound} is at most $\varepsilon_{\mathrm{stat}}^2$. This is guaranteed if each of the three terms is at most $\varepsilon_{\mathrm{stat}}^2/3$, which yields the following conditions on $m$:
\begin{enumerate}
\item[(i)] From the information-gain term:
\begin{equation}
\frac{12K\,\Gamma(m+1)\,\mathcal{F}_m^2}{\delta\,(m+1)\,\log(1+1/\lambda^2)} \;\le\; \frac{\varepsilon_{\mathrm{stat}}^2}{3},
\end{equation}
which is equivalent to
\begin{equation}\label{eq:condition_1}
m+1 \;\ge\; \frac{36K\,\Gamma(m+1)\,\mathcal{F}_m^2}{\delta\,\varepsilon_{\mathrm{stat}}^2\,\log(1+1/\lambda^2)}.
\end{equation}

\item[(ii)] From the discretization correction:
\begin{equation}
\frac{6\,\mathcal{F}_m^2}{m\,B_{\mathrm{net}}(m)} \;\le\; \frac{\varepsilon_{\mathrm{stat}}^2}{3},
\end{equation}
which is equivalent to
\begin{equation}\label{eq:condition_2}
m \;\ge\; \frac{18\,\mathcal{F}_m^2}{\varepsilon_{\mathrm{stat}}^2\,B_{\mathrm{net}}(m)}.
\end{equation}

\item[(iii)] From the grid residual:
\begin{equation}
\frac{12}{m^2} \;\le\; \frac{\varepsilon_{\mathrm{stat}}^2}{3},
\end{equation}
which is equivalent to
\begin{equation}\label{eq:condition_3}
m \;\ge\; \frac{6}{\varepsilon_{\mathrm{stat}}}.
\end{equation}
\end{enumerate}

When all three conditions are satisfied simultaneously, we have $\max_{1\le k\le K}\mathbb{E}_z[|\hat{Q}_k(z)-T\hat{Q}_{k-1}(z)|^2]\le\varepsilon_{\mathrm{stat}}^2$ with probability at least $1-\delta$, and therefore
\begin{equation}
\frac{2\phi_{U,U}\gamma}{(1-\gamma)^2}\max_{1\le k\le K}\left(\mathbb{E}_z[|\hat{Q}_k(z)-T\hat{Q}_{k-1}(z)|^2]\right)^{1/2} \;\le\; \frac{2\phi_{U,U}\gamma}{(1-\gamma)^2}\,\varepsilon_{\mathrm{stat}} \;=\; \frac{\epsilon}{2}.
\end{equation}
Combined with the algorithmic term (also at most $\epsilon/2$ by choice of $K$), the total loss satisfies
\begin{equation}
\mathbb{E}_{(s,a)\sim U(\mathcal{Z})}\!\left[|Q^*(s,a)-Q^{\pi_K}(s,a)|\right] \;\le\; \epsilon
\end{equation}
with probability at least $1-\delta$ over the joint randomness (training set and noise) completing the proof.

\section{Proof of \Cref{lem:gamma_intermediate}}\label{a:gamma_upper}
This proof is adapted from the proof of Lemma 2 in \cite{vakili2024kernelizedreinforcementlearningorder}, and we restate it here for clarity. Recall the definition 
\begin{equation}
    \Gamma(m)=\max_{\{z_i\}_{i=1}^{m}\subset \mathcal{Z}}\frac{1}{2}\log\det\left(\mathds{1}_m+\frac{1}{\lambda^2}K_{\{z_i\}_{i=1}^m}\right),
\end{equation}
where $K_{\{z_i\}_{i=1}^m}$ denotes the Gram matrix of the kernel $K_{(\mathcal{D},w)}$ for the corresponding subset of data points $\{z_i\}_{i=1}^m$.  
We will derive a uniform upper bound for $\log\det\left(\mathds{1}_m+\frac{1}{\lambda^2}K_{\{z_i\}_{i=1}^m}\right)$, valid for any subset $\{z_i\}_{i=1}^{m}\subset \mathcal{Z}$. 

Recall from \Cref{eq:PQC-mercer_1,eq:PQC-mercer_2,eq:PQC-mercer_3} that the kernel admits the decomposition
\begin{equation}
    K_{(\mathcal{D},w)}(z,z') = \sum_{i=0}^{2M} \nu_i \, \phi^{(i)}(z) \, \phi^{(i)}(z'),
\end{equation}
where 
\begin{equation}
    \nu_0=\frac{w_0^2}{\|w\|_2^2},\quad \nu_{2k-1}=\nu_{2k}=\frac{w_k^2}{2\|w\|_2^2},\quad k=1,\dots,M
\end{equation}
and $w = (w_0,\ldots, w_M)$ is the weight vector, with $w_i\geq w_{i+1}> 0$ for all $i$. Note that as a result of the chosen normalization one has $K_{(\mathcal{D},w)}(z,z)=1$ for all $z\in\mathcal{Z}$.

We now truncate this kernel decomposition into two parts: the first part corresponding to the $T$ largest reweighted features, and the second containing the remaining components. More specifically, we have
\begin{equation}
    \begin{split}
        K_{(\mathcal{D},w)}(z,z')&=\underbrace{\sum_{i=0}^{T} \nu_i\, \phi_i^{\top}(z)\,\phi_i(z')}_{\equiv K^{T}(z,z')}+\underbrace{\sum_{i=T+1}^{2M} \nu_i\, \phi_i^{\top}(z)\,\phi_i(z')}_{\equiv K^{0}(z,z')}.
    \end{split}
\end{equation}
This decomposition directly translates to the Gram matrix as
\begin{equation}
    K_{\{z_i\}_{i=1}^m}=K_{\{z_i\}_{i=1}^m}^{T}+K_{\{z_i\}_{i=1}^m}^0.
\end{equation}
For what follows, note that both $K_{\{z_i\}_{i=1}^m}^{T}$ and $K_{\{z_i\}_{i=1}^m}^{0}$ are positive semidefinite.

We are now ready to manipulate the following expression:
\begin{equation}
    \begin{split}
        \log\det\left(\mathds{1}_m+\frac{1}{\lambda^2}K_{\{z_i\}_{i=1}^m}\right)
        &=\log\det\left(\mathds{1}_m+\frac{1}{\lambda^2}K_{\{z_i\}_{i=1}^m}^T+\frac{1}{\lambda^2}K_{\{z_i\}_{i=1}^m}^0\right)\\
        &=\log\det\left(\left(\mathds{1}_m+\frac{1}{\lambda^2}K_{\{z_i\}_{i=1}^m}^T\right)\left(\mathds{1}_m+\frac{1}{\lambda^2}\left(\mathds{1}_m+\frac{1}{\lambda^2}K_{\{z_i\}_{i=1}^m}^T\right)^{-1}\,K_{\{z_i\}_{i=1}^m}^0\right)\right)\\
        &=\log\det\left(\mathds{1}_m+\frac{1}{\lambda^2}K_{\{z_i\}_{i=1}^m}^T\right)
        +\log\det\left(\mathds{1}_m+\frac{1}{\lambda^2}\left(\mathds{1}_m+\frac{1}{\lambda^2}K_{\{z_i\}_{i=1}^m}^T\right)^{-1}\,K_{\{z_i\}_{i=1}^m}^0\right),
    \end{split}
\end{equation}
where, in the second equality, we used the identity $\mathds{1}+A+B=(\mathds{1}+A)(\mathds{1}+(\mathds{1}+A)^{-1}B)$, and in the third, the property $\det(MN)=\det(M)\,\det(N)$.

Let us define $\vec{\phi}_T(z)=[\phi_1,\ldots,\phi_T]^\top$.  
This allows us to define the matrix $\Phi_T=[\vec{\phi}_T(z_i)]^\top_{i=1,\ldots,m}$, which is an $m\times T$ matrix stacking the truncated feature vectors of the $m$ data points as rows.   Additionally, let $\Sigma_T=\mathrm{diag}(\nu_1,\ldots,\nu_{T})$ containing the first $T$ normalized weights.

Using these definitions, we can upper bound the first term as follows:
\begin{equation}
    \begin{split}
        \log\det\left(\mathds{1}_m+\frac{1}{\lambda^2}K_{\{z_i\}_{i=1}^m}^T\right)
        &=\log\det\left(\mathds{1}_m+\frac{1}{\lambda^2}\Phi_T\Sigma_T\Phi_T^\top\right)\\
        &= \log\det\left(\mathds{1}_T+\frac{1}{\lambda^2}\Sigma_T^{1/2}\,\Phi_T^\top\,\Phi_T \,\Sigma_T^{1/2}\right)\\
        &\leq T\log\left(\frac{\mathrm{tr}(\mathds{1}_T+\frac{1}{\lambda^2}\Sigma_T^{1/2}\,\Phi_T^\top\,\Phi_T \,\Sigma_T^{1/2})}{T}\right)\\
        &=T\log\left(1+\frac{1}{\lambda^2 T}\mathrm{tr}(\Phi_T\,\Sigma_T\,\Phi_T^\top)\right)\\
        &\leq T\log\left(1+\frac{m}{\lambda^2 T}\right),
    \end{split}
\end{equation}
where in the second line we used the Weinstein–Aronszajn identity \cite{pozrikidis_introduction_2014}; in the third, the inequality $\log\det(A)\leq T \log(\mathrm{tr}(A)/T)$ that holds for a $T\times T$ positive definite matrix $A$ and follows from the arithmetic–geometric mean inequality applied to the eigenvalues of $A$; and finally, we used $\Phi_T\,\Sigma_T\,\Phi_T^\top=K_{\{z_i\}_{i=1}^m}^T$ and  $K(z,z)=1$. 

For the second term, we have
\begin{equation}
    \begin{split}
        \log\det\left(\mathds{1}_m+\frac{1}{\lambda^2}(\mathds{1}_m+\frac{1}{\lambda^2}K_{\{z_i\}_{i=1}^m}^T)^{-1}\,K_{\{z_i\}_{i=1}^m}^0\right)
        &\leq m\log\left(\frac{\mathrm{tr}(\mathds{1}_m + \frac{1}{\lambda^2}(\mathds{1}_m+\frac{1}{\lambda^2}K_{\{z_i\}_{i=1}^m}^T)^{-1}\,K_{\{z_i\}_{i=1}^m}^0)}{m}\right)\\
        &= m\log\left(1 + \frac{1}{\lambda^2m}\mathrm{tr}\left((\mathds{1}_m+\frac{1}{\lambda^2}K_{\{z_i\}_{i=1}^m}^T)^{-1}\,K_{\{z_i\}_{i=1}^m}^0\right)\right)\\
        &\leq m \log\left(1+\frac{1}{\lambda^2m}\mathrm{tr}(K_{\{z_i\}_{i=1}^m}^0)\right)\\
        &\leq m \log\left(1+\frac{\epsilon_T}{\lambda^2}\right)\\
        &\leq \frac{m\epsilon_T}{\lambda^2},
    \end{split}
\end{equation}
where in the first step we used the same property as before; in the third, the fact that $(\mathds{1}_m+\frac{1}{\lambda^2}K_{\{z_i\}_{i=1}^m}^T)^{-1}$ is positive definite and its largest eigenvalue is bounded by 1, together with the fact that $K_m^0$ is positive semi-definite. In particular, we applied the inequality that for two positive semi-definite matrices $A,B$, $\mathrm{tr}(AB)\leq \lambda_A \mathrm{tr}(B)$, where $\lambda_A$ is the largest eigenvalue of $A$. In the fourth step we used $K_{(\mathcal{D},w)}^0(z,z)\leq \epsilon_T$, and in the last one, the approximation $\log(1+x)\leq x$.

Combining both results, we obtain the desired upper bound
\begin{equation}
    \Gamma(m)\leq \frac{T}{2}\log\left(1+\frac{m}{\lambda^2 T}\right)+\frac{m\epsilon_T}{2\lambda^2}.
\end{equation}

\section{Proof of \Cref{lem:gamma_poly}}\label{a:gamma_poly}
In \Cref{a:gamma_upper}, we showed that for any positive integer $n$ and any $T \ge 1$ we have
\begin{equation}
\label{eq:gamma-general-bound}
\Gamma(n)
\le 
\frac{T}{2}\log\!\left(1+\frac{n}{\lambda^2 T}\right)
+
\frac{n}{2\lambda^2}\,\epsilon_T,
\end{equation}
where
\begin{equation}
\epsilon_T
=
\sup_{z\in\mathcal{Z}} K^0_{(\Omega_{\mathcal{D}},w)}(z,z)
=
\sum_{i=T+1}^{2M} \nu_i .
\end{equation}
Under Assumption~\ref{ass:poly-eigendecay}, there exist constants $C>0$ and $p>1$ such that the Mercer eigenvalues $\{\nu_j\}_{j=0}^{2M}$ satisfy
\begin{equation}
\nu_0 \le C,
\qquad
\nu_{2i-1} = \nu_{2i} \le \frac{C}{2} (i+1)^{-p}, \quad i=1,\ldots,M.
\end{equation}
In particular, the ordered Mercer eigenvalues decay polynomially (up to multiplicity) with exponent $p$.

Therefore, 
\begin{equation}
\label{eq:tail-bound}
\epsilon_T
=
\sum_{i=T+1}^{2M} \nu_i
\le
C \sum_{i=T+1}^{2M} i^{-p}
\le
C \int_T^{2M} x^{-p}\,dx
=
\frac{C}{p-1}\, T^{1-p}.
\end{equation}

We apply~\eqref{eq:gamma-general-bound} with $n = m+1$. Substituting~\eqref{eq:tail-bound}, we obtain
\begin{equation}
\label{eq:gamma-poly-intermediate}
\Gamma(m+1)
\le
\frac{T}{2}\log\!\left(1+\frac{m+1}{\lambda^2 T}\right)
+
\frac{C}{2\lambda^2(p-1)}\,(m+1)\, T^{1-p},
\end{equation}
which holds for any choice of $T$.

For $m+1 \geq \lambda^2 T$ (since $\lambda>0$ is a fixed constant),
\begin{equation}\label{eq:asymptotic-m-large-enough}
\log\!\left(1+\frac{m+1}{\lambda^2 T}\right)
\le
\log\!\left(\frac{2(m+1)}{\lambda^2 T}\right)
=
\log(m+1) - \log T + \mathcal{O}(1).
\end{equation}
We now choose
\begin{equation}
T = \left(\frac{m+1}{\log(m+1)}\right)^{1/p}.
\end{equation}

For the first term in~\eqref{eq:gamma-poly-intermediate}, using Eq.~\eqref{eq:asymptotic-m-large-enough}, we obtain
\begin{equation}
\frac{T}{2}\log\!\left(1+\frac{m+1}{\lambda^2 T}\right)
=
\mathcal{O}\!\left(
(m+1)^{1/p} (\log(m+1))^{(p-1)/p}
\right).
\end{equation}
For the second term,
\begin{equation}
(m+1)\, T^{1-p}
=
(m+1) \left(\frac{m+1}{\log(m+1)}\right)^{(1-p)/p}
=
\mathcal{O}\!\left(
(m+1)^{1/p} (\log(m+1))^{(p-1)/p}
\right).
\end{equation}
Both contributions scale in the same way, and therefore
\begin{equation}
\Gamma(m+1)
=
\mathcal{O}\!\left(
(m+1)^{1/p} (\log(m+1))^{(p-1)/p}
\right)
=
\mathcal{O}\!\left(
\bigl((m+1) (\log(m+1))^{p-1}\bigr)^{1/p}
\right).
\end{equation}
Hence there exists a constant $C_4>0$ (depending only on $C$, $p$, and $\lambda$) such that for all sufficiently large $m$
\begin{equation}
\Gamma(m+1) \le C_4 \bigl( (m+1) \, (\log(m+1))^{p-1} \bigr)^{1/p}.
\end{equation}
Since $p>1$, we have $(m+1)^{1/p} (\log(m+1))^{(p-1)/p} = o(m)$, which proves the claim.

\begin{remark}\label{rem:m+1-vs-m}
Since $m+1 \le 2m$ for $m\ge 1$ and $\log(m+1)\le 2\log m$ for $m\ge 2$, the bound above implies
\begin{equation}
\Gamma(m+1) \le 2C_4\bigl(m(\log m)^{p-1}\bigr)^{1/p}
\end{equation}
for all sufficiently large $m$.
\end{remark}

\section{Proof of \Cref{cor:poly_sample}}\label{a:coro_sample}
We begin by recalling the three conditions from \Cref{lem:sufficient_m_K}
that must be simultaneously satisfied:
\begin{align}
m+1 &\ge \frac{36K\,\Gamma(m+1)\,\mathcal{F}_m^2}{\delta\,\varepsilon_{\mathrm{stat}}^2\,\log(1+1/\lambda^2)},
\label{eq:mcond1_ab}\\
m &\ge \frac{18\,\mathcal{F}_m^2}{\varepsilon_{\mathrm{stat}}^2\,B_{\mathrm{net}}(m)},
\label{eq:mcond2_ab}\\
m &\ge \frac{6}{\varepsilon_{\mathrm{stat}}}, \label{eq:mcond3_ab}
\end{align}
where
\begin{equation}
\varepsilon_{\mathrm{stat}} =
\frac{\epsilon(1-\gamma)^2}{4\gamma\sqrt{|\mathcal{A}|\,p_{\max}}},
\qquad
G_K = \max_{1\le k\le K}\|T\hat{Q}_{k-1}\|_{\mathcal{H}_{(\mathcal{D},w)}},
\end{equation}
\begin{equation}
B_{\mathrm{net}}(m) =
\max\!\left\{G_K,\;\frac{\sqrt{m}\,V_{\max}}{\lambda}\right\},
\qquad
\mathcal{F}_m = G_K +
\frac{\tau}{\lambda}\sqrt{2\log\frac{4K|\mathcal{Z}_m|}{\delta}},
\end{equation}
and $|\mathcal{Z}_m| \le
|\mathcal{A}|\bigl(\sqrt{d_{\mathcal{S}}}\;L_{\mathcal{D}}\,B_{\mathrm{net}}(m)\,m +
1\bigr)^{d_{\mathcal{S}}}$.

\medskip
\noindent\textbf{Step 1: Upper bounding $\mathcal{F}_m^2$.}

We first isolate the $m$-dependent part of $\mathcal{F}_m$. We use 
$B_{\mathrm{net}}(m) \le G_K + \sqrt{m}\,V_{\max}/\lambda$ together with the property $(a\,m+b)\le (a+b)\,m$ for $a,b\ge 0$, $m\ge 1$. This gives, for $m\ge 1$:
\begin{equation}
L_{\mathcal{D}}\,B_{\mathrm{net}}(m)\,m + 1
\;\le\; L_{\mathcal{D}}\,(G_K + \sqrt{m}\,V_{\max}/\lambda)\;m +1\;\le\; L_{\mathcal{D}}\bigl((G_K + V_{\max}/\lambda)\;\sqrt{m}\bigr)\,m+1\;\le\;
\bigl(L_{\mathcal{D}}(G_K + V_{\max}/\lambda) + 1\bigr)\,m^{3/2}.
\end{equation}
Defining the $m$-independent constant
\begin{equation}
B_* \;\equiv\; \sqrt{d_{\mathcal{S}}}\;L_{\mathcal{D}}(G_K + V_{\max}/\lambda) + 1,
\end{equation}
we obtain
\begin{equation}
|\mathcal{Z}_m| \;\le\;
|\mathcal{A}|\,B_*^{d_{\mathcal{S}}}\,m^{3d_{\mathcal{S}}/2},
\end{equation}
and therefore
\begin{equation}
\log\frac{4K|\mathcal{Z}_m|}{\delta}
\;\le\;
\underbrace{\log\frac{4K|\mathcal{A}|\,B_*^{d_{\mathcal{S}}}}{\delta}}_{\equiv L}
\;+\; \frac{3d_{\mathcal{S}}}{2}\log m.
\end{equation}
Combining this with $(a+b)^2 \le 2a^2 + 2b^2$, we bound
\begin{equation}
\mathcal{F}_m^2
\;\le\;
2G_K^2 + \frac{2\tau^2}{\lambda^2}\Bigl(2\log\frac{4K|\mathcal{Z}_m|}{\delta}\Bigr)
\;\le\;
\underbrace{2G_K^2 + \frac{4\tau^2}{\lambda^2}L}_{\equiv\,\mathcal{P}}
\;+\;
\underbrace{\frac{6\tau^2
d_{\mathcal{S}}}{\lambda^2}}_{\equiv\,\mathcal{Q}}\,\log m,
\end{equation}
Note that $\mathcal{P}$ and $\mathcal{Q}$ are independent of $m$.

\medskip
\noindent\textbf{Step 2: Reducing condition~\eqref{eq:mcond1_ab} using
polynomial weight decay.}

By \Cref{lem:gamma_poly} (applied under \Cref{ass:poly-eigendecay}),
there exists a constant $C_4>0$, depending only on $C$, $p$, and
$\lambda$, such that for all sufficiently large $m$,
\begin{equation}
\Gamma(m+1) \;\le\; C_4\bigl((m+1)(\log(m+1))^{p-1}\bigr)^{1/p}.
\end{equation}
We apply two elementary inequalities: $m+1\leq 2m$ (valid for $m\ge 1$) and $\log(m+1)\le 2\log m$ (valid for $m\ge 2$), giving
\begin{equation}
\Gamma(m+1) \;\le\; 2C_4\bigl(m(\log m)^{p-1}\bigr)^{1/p},
\end{equation}
and therefore
\begin{equation}
\frac{m+1}{\Gamma(m+1)}
\;\ge\;
\frac{m}{2C_4\bigl(m(\log m)^{p-1}\bigr)^{1/p}}
\;=\;
\frac{1}{2C_4}\left(\frac{m}{\log m}\right)^{(p-1)/p}.
\end{equation}
Setting $\alpha \equiv (p-1)/p \in (0,1)$ and defining the constant
\begin{equation}
C_5 \;\equiv\; \frac{72\,C_4}{\log(1+1/\lambda^2)},
\end{equation}
condition~\eqref{eq:mcond1_ab} is implied by
\begin{equation}\label{eq:reduced_cond1}
\left(\frac{m}{\log m}\right)^\alpha
\;\ge\;
\frac{C_5\,K\,\mathcal{F}_m^2}{\delta\,\varepsilon_{\mathrm{stat}}^2}.
\end{equation}
Substituting the bound $\mathcal{F}_m^2 \le \mathcal{P} +
\mathcal{Q}\log m$ from Step~1, condition~\eqref{eq:reduced_cond1} is
in turn implied by
\begin{equation}\label{eq:inversion_form}
\left(\frac{m}{\log m}\right)^\alpha
\;\ge\;
\frac{C_5\,K\,\mathcal{P}}{\delta\,\varepsilon_{\mathrm{stat}}^2}
\;+\;
\frac{C_5\,K\,\mathcal{Q}}{\delta\,\varepsilon_{\mathrm{stat}}^2}\,\log m.
\end{equation}
This is of the form $\left(\frac{m}{\log m}\right)^\alpha \ge A + B'\log m$ with
\begin{equation}
A \;\equiv\; \frac{C_5\,K\,\mathcal{P}}{\delta\,\varepsilon_{\mathrm{stat}}^2},
\qquad
B' \;\equiv\; \frac{C_5\,K\,\mathcal{Q}}{\delta\,\varepsilon_{\mathrm{stat}}^2}.
\end{equation}
This will be analyzed in Step 4.

\medskip
\noindent\textbf{Step 3: Handling condition~\eqref{eq:mcond2_ab}.}

Since $B_{\mathrm{net}}(m) \ge \sqrt{m}\,V_{\max}/\lambda$ for all
$m\ge 1$, condition~\eqref{eq:mcond2_ab} is implied by
\begin{equation}
m \;\ge\; \frac{18\lambda\,\mathcal{F}_m^2}{\varepsilon_{\mathrm{stat}}^2\,V_{\max}\,\sqrt{m}}.
\end{equation}
Applying $\mathcal{F}_m^2 \le \mathcal{P} +
\mathcal{Q}\log m$, it is enough to guarantee
\begin{equation}
m^{3/2} \;\ge\; \frac{18\lambda\,(\mathcal{P} + \mathcal{Q}\log
m)}{\varepsilon_{\mathrm{stat}}^2\,V_{\max}}.
\end{equation}
Using $\log m \le \sqrt{m}$ for all $m\ge 1$ (natural logarithm), together with the fact that $x\ge a+b$ is implied by both $x\ge 2a$ and $x\ge 2b$ holding simultaneously,
\begin{equation}
m^{3/2} \;\ge\;
\frac{36\lambda\,\mathcal{P}}{\varepsilon_{\mathrm{stat}}^2\,V_{\max}}
\qquad\text{and}\qquad
m \;\ge\; \frac{36\lambda\,\mathcal{Q}}{\varepsilon_{\mathrm{stat}}^2\,V_{\max}}.
\end{equation}
Then 
\begin{equation}\label{eq:m_disc}
m \;\ge\; \max\!\left\{
  \left(\frac{36\lambda\,\mathcal{P}}{\varepsilon_{\mathrm{stat}}^2\,V_{\max}}\right)^{2/3},\;
  \frac{36\lambda\,\mathcal{Q}}{\varepsilon_{\mathrm{stat}}^2\,V_{\max}}
\right\}, 
\end{equation}
implies condition~\eqref{eq:mcond2_ab}.

\medskip
\noindent\textbf{Step 4: An inversion lemma.}

We need to invert condition~\eqref{eq:inversion_form}, which has the form $(m/\log m)^\alpha \ge A + B'\log m$. 

We will repeatedly use the following elementary relation.
\begin{lemma}\label{lemm:tool_inv}
    For any $X\ge e$, if
    \begin{equation}
    m\ge 2X\log X,
    \end{equation}
    then 
    \begin{equation}
    \frac{m}{\log m}\ge X
    \end{equation}
    holds.
\end{lemma}
\begin{proof}
    Since $X\ge e$, we have $2X\log X\geq 2e \log e=2e>e$, and therefore $m\ge 2X\log X>e$.

    The function $f(m)=\frac{m}{\log m}$ is increasing on $(e,\infty)$. Since $m\ge 2X\log X$, we have
    \begin{equation}
    \frac{m}{\log m}\ge \frac{2X \log X}{\log(2X\log X)}.
    \end{equation}
    Therefore, it suffices to show that $\frac{2X \log X}{\log(2X\log X)}\ge X$. For $X>0$, this is equivalent to
    \begin{equation}
    \log X\ge \log 2+\log(\log X)
    \end{equation}
    which corresponds to $X\ge 2\log X$. The function $g(X)=X-2\log X$ is increasing for $X> 2$, and $g(e)=e-2>0$; therefore $X-2\log X>0$ in this regime. This proves $2\log X\ge \log(2X\log X)$ and consequently
    \begin{equation}
    \frac{m}{\log m}\ge \frac{2X \log X}{\log(2X\log X)}.
    \end{equation}
\end{proof}
Using Lemma~\ref{lemm:tool_inv}, we establish the following inversion result for condition~\eqref{eq:mcond1_ab}.

\begin{lemma}[Inversion lemma]\label{lem:inversion}
Let $\alpha\in(0,1)$, $A\ge e$, and $B'\ge 0$. Define
\begin{equation}
A_{\mathrm{inv}} \;\equiv\; \max\!\left\{
  \left(2A\right)^{1/\alpha},\;
  \left(2X^*\log X^*\right)^{\frac{\alpha+1}{\alpha}}
\right\},
\end{equation}
with $X^*=\max\left\{e,\frac{\alpha+1}{\alpha}(2B')^{\frac{1}{\alpha+1}}\right\}$. If $m \ge 2A_{\mathrm{inv}}\log A_{\mathrm{inv}}$, then
\begin{equation}
\left(\frac{m}{\log m}\right)^\alpha \;\ge\; A + B'\log m.
\end{equation}
\end{lemma}

\begin{proof}
We use again that $x\ge a+b$ follows from both $x\ge 2a$ and $x\ge 2b$ holding. The proof proceeds in two parts.

\medskip
\noindent\textbf{Part 1: $(m/\log m)^\alpha \ge 2A$.}

Since $A_{\mathrm{inv}} \ge (2A)^{1/\alpha}$, the choice $m\ge
2A_{\mathrm{inv}}\log A_{\mathrm{inv}}$ gives, by Lemma~\ref{lemm:tool_inv},
$m/\log m \ge A_{\mathrm{inv}} \ge (2A)^{1/\alpha}$, hence $(m/\log
m)^\alpha \ge 2A$.

\medskip
\noindent\textbf{Part 2: $(m/\log m)^\alpha \ge 2B'\log m$.}

Since $A_{\rm inv}\ge  (2X^*\log X^*)^{\frac{\alpha+1}{\alpha}}$, and since $X^*\ge e$ by definition and $(\alpha+1)/\alpha >1$, we have $A_{\rm inv}\ge (2e\log e)^{\frac{\alpha+1}{\alpha}}>2e>e$. Therefore,
\begin{equation}
m\ge 2A_{\rm inv}\log A_{\rm inv}>2A_{\rm inv}>A_{\rm inv}.
\end{equation}
Since $m>A_{\rm inv}$ and $A_{\rm inv}\ge  (2X^*\log X^*)^{\frac{\alpha+1}{\alpha}}$, we conclude
\begin{equation}
m\ge (2X^*\log X^*)^{\frac{\alpha+1}{\alpha}}.
\end{equation}
Defining $m'=m^{\frac{\alpha}{\alpha+1}}$, this condition becomes $m'\ge 2X^*\log X^*$. Applying Lemma~\ref{lemm:tool_inv} yields
\begin{equation}
\frac{m'}{\log m'}\ge X^*
\end{equation}
which, after substituting back $m'$ and rearranging, yields
\begin{equation}
\frac{m}{(\log m)^{\frac{\alpha+1}{\alpha}}}\ge \left(\frac{\alpha}{\alpha+1}X^*\right)^{\frac{\alpha+1}{\alpha}}.
\end{equation}
By definition, we have 
\begin{equation}
X^*\geq \frac{\alpha+1}{\alpha}(2B')^{\frac{1}{\alpha+1}},
\end{equation}
so plugging this here in,
\begin{equation}
\frac{m}{(\log m)^{\frac{\alpha+1}{\alpha}}}\geq (2B')^{1/\alpha}
\end{equation}
Raising both sides to the power $\alpha$, we obtain $m^\alpha/(\log m)^{\alpha+1}\geq 2B'$, which for $m>e$ implies
\begin{equation}
\frac{m^\alpha}{(\log m)^\alpha}\geq 2B'\log m
\end{equation}
as desired for Part 2.
\\
\\
\noindent\textbf{Combining Parts 1 and 2.}

Since both hold simultaneously,
\begin{equation}
  \left(\frac{m}{\log m}\right)^\alpha
  \;\ge\; \max(2A,\,2B'\log m)
  \;\ge\; A + B'\log m,
\end{equation}
where the last step uses $\max(2a,2b)\ge a+b$ for $a,b\ge 0$.
\end{proof}

\medskip
\noindent\textbf{Step 5: Choosing $m$ and verifying all conditions.}

We apply \Cref{lem:inversion} with $A = C_5 K\mathcal{P}/(\delta\varepsilon_{\mathrm{stat}}^2)$ and $B' = C_5 K\mathcal{Q}/(\delta\varepsilon_{\mathrm{stat}}^2)$. This gives
\begin{equation}
A_{\mathrm{inv}} \;\equiv\; \max\!\left\{
  \left(\frac{2C_5
K\mathcal{P}}{\delta\varepsilon_{\mathrm{stat}}^2}\right)^{1/\alpha},\;
  \left(2X^*\log X^*\right)^{(\alpha+1)/\alpha}
\right\}.
\end{equation}
where
\begin{equation}
X^*=\max\left\{e,\frac{\alpha+1}{\alpha}(2B')^{\frac{1}{\alpha+1}}\right\}.
\end{equation}
We assume $A \ge e$; if not, we replace $A$ by $e$ in the definition of $A_{\mathrm{inv}}$ (this only strengthens the sufficient condition), so the first branch $(2A)^{1/\alpha}$ reduces to a constant depending only on~$p$, and $A_{\mathrm{inv}}$ is then determined by the second branch involving~$B'$. Set
\begin{equation}\label{eq:final_m_correct}
m \;\equiv\; \max\!\left\{
  \frac{6}{\varepsilon_{\mathrm{stat}}},\quad
  \max\!\left\{
  \left(\frac{36\lambda\,\mathcal{P}}{\varepsilon_{\mathrm{stat}}^2\,V_{\max}}\right)^{2/3},\;
  \frac{36\lambda\,\mathcal{Q}}{\varepsilon_{\mathrm{stat}}^2\,V_{\max}}
\right\},\quad
  2A_{\mathrm{inv}}\log A_{\mathrm{inv}}
\right\}.
\end{equation}

\noindent\textit{Condition~\eqref{eq:mcond3_ab}.}
Satisfied directly by the first term in the maximum.

\noindent\textit{Condition~\eqref{eq:mcond2_ab}.}
Satisfied as established in Step~3.

\noindent\textit{Condition~\eqref{eq:mcond1_ab}.}
Since $m\ge 2A_{\mathrm{inv}}\log A_{\mathrm{inv}}$, \Cref{lem:inversion} gives $(m/\log m)^\alpha \ge A + B'\log m$, which by Step~2 implies condition~\eqref{eq:mcond1_ab}.

\medskip
\noindent\textbf{Step 6: Asymptotic form of the sample complexity.}
We now derive an explicit asymptotic expression for the sufficient sample size~\eqref{eq:final_m_correct} $m$. The three terms in the maximum scale as
\begin{equation}
\mathcal{O}(\varepsilon_{\mathrm{stat}}^{-1}),\quad \mathcal{O}\!\left(
\max\left\{
\frac{\mathcal P^{2/3}}{\varepsilon_{\mathrm{stat}}^{4/3}},
\frac{\mathcal Q}{\varepsilon_{\mathrm{stat}}^2}
\right\}
\right), \quad 2A_{\mathrm{inv}}\log A_{\mathrm{inv}}.
\end{equation}
Since
\begin{equation}
A=\widetilde{\Theta}(\varepsilon_{\mathrm{stat}}^{-2}\delta^{-1}),
\qquad
B'=\widetilde{\Theta}(\varepsilon_{\mathrm{stat}}^{-2}\delta^{-1}),
\end{equation}
where the $\widetilde{\Theta}$ absorbs the factors $K\mathcal{P}$ and $K\mathcal{Q}$, which depend polylogarithmically on $1/\epsilon$ and $1/\delta$ through $K=\mathcal{O}(\log(1/\epsilon))$ and $L=\log(4K|\mathcal{A}|B_*^{d_{\mathcal{S}}}/\delta)$, the definition of $A_{\mathrm{inv}}$ yields
\begin{equation}
A_{\mathrm{inv}}
=
\widetilde{\Theta}
\!\left(
\varepsilon_{\mathrm{stat}}^{-2/\alpha}
\delta^{-1/\alpha}
\right),
\end{equation}
where $\widetilde{\Theta}$ hides only polylogarithmic factors. Recalling that $\alpha=(p-1)/p\in(0,1)$, we have $\frac{2}{\alpha}=\frac{2p}{p-1}>2$. Hence,
\begin{equation}
2A_{\mathrm{inv}}\log A_{\mathrm{inv}}
=
\widetilde{\Theta}\!\left(
\varepsilon_{\mathrm{stat}}^{-2p/(p-1)}
\delta^{-p/(p-1)}
\right),
\end{equation}
which grows strictly faster, both as $\varepsilon_{\mathrm{stat}}\to0$ and as $\delta\to0$, than the remaining two terms, whose largest dependence on $\varepsilon_{\mathrm{stat}}$ is $\varepsilon_{\mathrm{stat}}^{-2}$ and which carry no explicit inverse powers of $\delta$. Therefore, $2A_{\mathrm{inv}}\log A_{\mathrm{inv}}$ determines the asymptotic sample complexity.

The expression for $A_{\mathrm{inv}}$ in Lemma~\ref{lem:inversion} is the maximum of two quantities,
\begin{equation}
A_{\mathrm{inv}}
=
\max\!\left\{
\left(2A\right)^{1/\alpha},
\;
\left(2X^*\log X^*\right)^{(\alpha+1)/\alpha}
\right\}.
\end{equation}
Consequently, the asymptotic behaviour of $2A_{\mathrm{inv}}\log A_{\mathrm{inv}}$ is determined by whichever of these two terms attains the maximum. We therefore analyze each possibility separately.

\medskip
\noindent\textit{First branch.}
Taking
\begin{equation}
A_{\mathrm{inv}} = (2A)^{1/\alpha} = \left( \frac{2C_5K\mathcal P} {\delta\varepsilon_{\mathrm{stat}}^2} \right)^{p/(p-1)}.
\end{equation}
Then
\begin{equation}
2A_{\mathrm{inv}}\log A_{\mathrm{inv}} = \widetilde{\mathcal{O}}\!\left(\left(\frac{K\mathcal P}{\delta\varepsilon_{\mathrm{stat}}^2}\right)^{p/(p-1)}\right).
\end{equation}

\medskip
\noindent\textit{Second branch.}
Now we take
\begin{equation}
A_{\mathrm{inv}} = (2X^*\log X^*)^{(\alpha+1)/\alpha},
\end{equation}
where
\begin{equation}
X^* = \max\!\left\{ e,\, \frac{\alpha+1}{\alpha}(2B')^{1/(\alpha+1)} \right\}, \qquad B' = \frac{C_5K\mathcal Q} {\delta\varepsilon_{\mathrm{stat}}^2}.
\end{equation}
Since $X^*=\widetilde{\mathcal{O}}\!\left(
(B')^{1/(\alpha+1)}\right)$, we obtain $A_{\mathrm{inv}} = \widetilde O\!\left( (B')^{1/\alpha} \right)$, and consequently
\begin{equation}
2A_{\mathrm{inv}}\log A_{\mathrm{inv}} = \widetilde O\!\left( \left( \frac{K\mathcal Q} {\delta\varepsilon_{\mathrm{stat}}^2} \right)^{p/(p-1)} \right).
\end{equation}
Taking into account that $\mathcal Q = \frac{6\tau^2d_{\mathcal S}}{\lambda^2}$, this becomes
\begin{equation}
2A_{\mathrm{inv}}\log A_{\mathrm{inv}} =
\widetilde O\!\left(
\left(
\frac{Kd_{\mathcal S}}
{\delta\varepsilon_{\mathrm{stat}}^2}
\right)^{p/(p-1)}
\right),
\end{equation}
where constants depending only on $\tau,\lambda,C,p,\gamma$ have been absorbed into the $\widetilde{\mathcal{O}}$ notation.

\medskip

Combining both branches,
\begin{equation}
2A_{\mathrm{inv}}\log A_{\mathrm{inv}} = \widetilde{\mathcal{O}}\!\left( \left( \frac{K} {\delta\varepsilon_{\mathrm{stat}}^2} \right)^{p/(p-1)} \max\!\left\{ \mathcal P^{p/(p-1)}, d_{\mathcal S}^{p/(p-1)} \right\} \right).
\end{equation}
Using $K=\mathcal{O}\left(\log\frac1\epsilon\right)$, the factor $K^{p/(p-1)}$ is polylogarithmic in $1/\epsilon$ and is therefore absorbed into the $\widetilde{\mathcal{O}}$ notation. Moreover, as $\varepsilon_{\mathrm{stat}}= \Theta\!\left( \frac{\epsilon} {\sqrt{|\mathcal A|\,p_{\max}}} \right)$, therefore
\begin{equation}
\varepsilon_{\mathrm{stat}}^{-2}=\Theta\!\left(\frac{|\mathcal A|\,p_{\max}}{\epsilon^2}\right).
\end{equation}
Finally, substituting the definitions
\begin{equation}
\mathcal P=2G_K^2+\frac{4\tau^2}{\lambda^2}\log\frac{4K|\mathcal A|B_*^{d_{\mathcal S}}}{\delta}, \qquad B_* = \sqrt{d_{\mathcal{S}}}\,L_{\mathcal D}\left(G_K+\frac{V_{\max}}{\lambda}\right)+1,
\end{equation}
yields
\begin{equation}
m = \widetilde{\mathcal{O}}\left( \left( \frac{|\mathcal A|\,p_{\max}} {\delta\,\epsilon^2} \right)^{p/(p-1)} \max\!\left\{ \left( G_K^2+ \log\frac{4K|\mathcal A|B_*^{d_{\mathcal S}}}{\delta} \right)^{p/(p-1)}, \; d_{\mathcal S}^{p/(p-1)} \right\} \right),
\end{equation}
where the $\widetilde{\mathcal{O}}$ notation hides only polylogarithmic factors in $1/\epsilon$, $1/\delta$, and constants depending only on $\gamma,\lambda,\tau,C,p$. Using the inequality 
\begin{equation}
(G_K^2 + L)^{p/(p-1)} \le 2^{p/(p-1)}\max\{G_K^2, L\}^{p/(p-1)}
\end{equation}
which follows from the elementary bound $a+b\leq 2\max \{a,b\}$, we obtain the bound stated in \Cref{cor:poly_sample}.

\section{Relaxing Assumption~\ref{as:func_classes}}
\label{ap:relaxing_closure}
In the main analysis we imposed the closure condition~\eqref{eq:stronger}, which requires that the image of the clipped PQC-RKHS under the Bellman operator $T$ is contained in $\mathcal{H}_{(\mathcal{D},w)}$. This condition eliminated approximation error entirely, since the ideal regression target $T\hat{Q}_{k-1}$ at each iteration was directly realizable in the hypothesis class. Here we show how the analysis changes when this assumption is dropped, and identify the additional error term that must be controlled.

\subsection{Setup: Two Distinct Function Classes}
Suppose that instead of a single RKHS $\mathcal{H}_{(\mathcal{D},w)}$, we now work with \emph{two} function classes:
\begin{itemize}
    \item $\mathcal{H}_{(\mathcal{D},w)}$ -- the hypothesis class over which Algorithm~\ref{alg:dequantization} optimizes at each iteration (the \textit{model} class);
    \item $\mathcal{G}$ -- a (possibly larger) function class such that $T\hat{Q}_{k-1}\in\mathcal{G}$ for all $k$, i.e., the class that contains the Bellman updates of all iterates (the \textit{target} class).
\end{itemize}
Under Assumption~\ref{as:func_classes} we had $\mathcal{G} = \mathcal{H}_{(\mathcal{D},w)}$. Without this assumption, $\mathcal{G}$ may contain functions not realizable in $\mathcal{H}_{(\mathcal{D},w)}$. As discussed in~\cite{chen2019information}, finite-sample guarantees in this setting typically require either Bellman completeness (i.e., Assumption~\ref{as:func_classes}) or an explicit approximation-error term in the bound.

\subsection{Decomposing the Per-Iteration Error}
At each iteration $k$, define the best approximation of the Bellman target within the hypothesis class:
\begin{equation}
\bar{g}_k \equiv \operatorname{argmin}_{f \in \mathcal{H}_{(\mathcal{D},w)}} \mathbb{E}_{z \sim U(\mathcal{S}\times\mathcal{A})}\bigl[|T\hat{Q}_{k-1}(z) - f(z)|^2\bigr].
\end{equation}
Inserting $\pm \bar{g}_k$ and applying the triangle inequality in $L^2(U(\mathcal{S}\times\mathcal{A}))$ gives:
\begin{align}
&\left(\mathbb{E}_{z \sim U(\mathcal{S}\times\mathcal{A})}\bigl[|T\hat{Q}_{k-1}(z) - \hat{Q}_k(z)|^2\bigr]\right)^{1/2} \nonumber \\
&\qquad \leq \underbrace{\left(\mathbb{E}_{z \sim U(\mathcal{S}\times\mathcal{A})}\bigl[|T\hat{Q}_{k-1}(z) - \bar{g}_k(z)|^2\bigr]\right)^{1/2}}_{\text{approximation error} \;\leq\; \beta_K}
+ \underbrace{\left(\mathbb{E}_{z \sim U(\mathcal{S}\times\mathcal{A})}\bigl[|\bar{g}_k(z) - \hat{Q}_k(z)|^2\bigr]\right)^{1/2}}_{\text{estimation error: } \varepsilon_{\mathrm{est}}}. \label{eq:two-class-decomp}
\end{align}
The approximation error at step $k$ is bounded by $\beta_K$ by definition in Equation~\eqref{eq:beta_K}.

\paragraph{Estimation error.}
This term measures how well the KRR estimator $\hat{Q}_k$ recovers the best realizable approximation $\bar{g}_k \in \mathcal{H}_{(\mathcal{D},w)}$ from $m$ samples. Since $\bar{g}_k\in\mathcal{H}_{(\mathcal{D},w)}$ by construction, the regression problem at iteration $k$ is now realizable (the target $\bar{g}_k$ lies in the hypothesis class), and the analysis of Appendix~\ref{a:proof_lemm_mk} applies verbatim with $T\hat{Q}_{k-1}$ replaced by $\bar{g}_k$. In particular, the bounds of Lemma~\ref{lem:sufficient_m_K} hold with $G_K = \max_{1\le k\le K}\|T\hat{Q}_{k-1}\|_{\mathcal{H}_{(\mathcal{D},w)}}$ replaced by $\max_{1\le k\le K}\|\bar{g}_k\|_{\mathcal{H}_{(\mathcal{D},w)}}$.

\paragraph{Approximation error.} This term quantifies the unavoidable mismatch between the Bellman target $T\hat{Q}_{k-1}\in\mathcal{G}$ and the hypothesis class $\mathcal{H}_{(\mathcal{D},w)}$. Let us define 
\begin{equation}
    \beta_K \equiv \max_{1 \leq k \leq K} \inf_{f \in \mathcal{H}_{(\mathcal{D},w)}} \left(\mathbb{E}_{z \sim U(\mathcal{S} \times \mathcal{A})}\bigl[|T\hat{Q}_{k-1}(z) - f(z)|^2\bigr]\right)^{1/2}
    \label{eq:beta_K}
\end{equation}
be the worst-case $L^2(U(\mathcal{S}\times\mathcal{A}))$ approximation error of the Bellman targets $\{T\hat{Q}_{k-1}\}_{k=1}^K$ within $\mathcal{H}_{(\mathcal{D},w)}$. This coefficient is the main object of the relaxed analysis:
\begin{itemize}
    \item If Assumption~\ref{as:func_classes} holds (Bellman closure), then $\beta_K = 0$ and we recover the original theorem.
    \item If Assumption~\ref{as:func_classes} fails, $\beta_K > 0$ is an irreducible approximation bias that limits the achievable accuracy regardless of sample size or number of iterations.
\end{itemize}
In particular, $\beta_K$ depends on the actual iterates $\hat{Q}_0, \ldots, \hat{Q}_{K-1}$ produced by the algorithm, not merely on the function classes $\mathcal{H}_{(\mathcal{D},w)}$ and $\mathcal{G}$.

\subsection{Propagating the Error through FQI}

Substituting the decomposition~\eqref{eq:two-class-decomp} into the FQI error propagation, the total suboptimality bound becomes:

\begin{lemma}[Relaxed suboptimality bound]\label{thm:relaxed}
Under the same conditions as Theorem~\ref{th:fqi_error_propagation} but without Assumption~\ref{as:func_classes}, we have
\begin{equation}
    \mathbb{E}_{z \sim U(\mathcal{S}\times\mathcal{A})}\bigl[|Q^*(z) - Q^{\pi_K}(z)|\bigr]
    \;\leq\;
    \frac{2\phi_{U,U}\gamma}{(1-\gamma)^2}\bigl(\varepsilon_{\mathrm{est}} + \beta_K\bigr)
    + \frac{4\gamma^{K+1}}{(1-\gamma)^2}R_{\max},
    \label{eq:relaxed-bound}
\end{equation}
where $\varepsilon_{\mathrm{est}} = \max_{1\le k\le K}\left(\mathbb{E}_{z\sim U(\mathcal{S}\times\mathcal{A})}[|\bar{g}_k(z) - \hat{Q}_k(z)|^2]\right)^{1/2}$.
\end{lemma}

\noindent When Assumption~\ref{as:func_classes} holds, $\beta_K = 0$ and~\eqref{eq:relaxed-bound} reduces to the bound of Theorem~\ref{th:fqi_error_propagation}. More generally, the bound decomposes cleanly into three terms: a statistical error $\varepsilon_{\mathrm{est}}$ controlled by the sample size $m$, an irreducible approximation bias $\beta_K$ controlled by the expressiveness of $\mathcal{H}_{(\mathcal{D},w)}$ relative to the Bellman targets, and an algorithmic error decaying geometrically in $K$.

In particular, for any target accuracy $\varepsilon > 0$, one can choose $m$ and $K$ so that the estimation and finite-iteration terms are each at most $\varepsilon/2$ as we do in Lemma~\ref{lem:sufficient_m_K}, yielding
\begin{equation}
\mathbb{E}_{z \sim U(\mathcal{S}\times\mathcal{A})}\bigl[|Q^*(z) - Q^{\pi_K}(z)|\bigr]
\;\leq\;
\varepsilon + \frac{2\phi_{U,U}\gamma}{(1-\gamma)^2}\beta_K.
\end{equation}
The remaining term $\frac{2\phi_{U,U}\gamma}{(1-\gamma)^2}\beta_K$ is irreducible: it cannot be driven to zero by increasing $m$ or $K$ alone. The theorem is therefore meaningful only when $\beta_K$ is small.

\subsection{A Conservative Upper Bound on \texorpdfstring{$\beta_K$}{betaK}}

In general, $\beta_K$ depends on the algorithm's iterates and must be controlled by structural assumptions on the problem. We record one simple, if crude, upper bound for completeness.

\begin{proposition}\label{prop:beta_crude}
Suppose $Q^* \in \mathcal{H}_{(\mathcal{D},w)}$ and that the iterates satisfy $\|\hat{Q}_{k-1}\|_\infty \leq V_{\max}$ for all $k$. Then
\begin{equation}
    \beta_K \;\leq\; 2\gamma V_{\max}.
    \label{eq:beta_crude}
\end{equation}
\end{proposition}
\begin{remark}
We note that the assumption $Q^* \in \mathcal{H}_{(\mathcal{D},w)}$ used in the following proposition is itself a strong realizability condition, essentially equivalent to what Corollary~\ref{cor:realizability} derives from the Bellman closure condition~\eqref{eq:stronger} that is being relaxed in this appendix. Proposition~\ref{prop:beta_crude} is therefore best understood as a sanity check confirming that $\beta_K$ is finite under mild conditions, rather than as a sharp or practically useful bound.
\end{remark}

\begin{proof}
Since $Q^* \in \mathcal{H}_{(\mathcal{D},w)}$ and $TQ^* = Q^*$, we have for each $k$:
\begin{equation}
\inf_{f \in \mathcal{H}_{(\mathcal{D},w)}} \|T\hat{Q}_{k-1} - f\|_\infty
\;\leq\; \|T\hat{Q}_{k-1} - Q^*\|_\infty
= \|T\hat{Q}_{k-1} - TQ^*\|_\infty
\leq \gamma\,\|\hat{Q}_{k-1} - Q^*\|_\infty
\leq 2\gamma V_{\max},
\end{equation}
where the second inequality uses $\gamma$-contractivity of $T$, and the last uses $\|\hat{Q}_{k-1}\|_\infty \leq V_{\max}$ and $\|Q^*\|_\infty \leq V_{\max}$. Since $\|\cdot\|_{L^2(U(\mathcal{S}\times\mathcal{A}))} \leq \|\cdot\|_\infty$, the same bound holds for $\beta_K$ as defined in~\eqref{eq:beta_K}.
\end{proof}

Substituting~\eqref{eq:beta_crude} into~\eqref{eq:relaxed-bound} gives
\begin{equation}
\frac{2\phi_{U,U}\gamma}{(1-\gamma)^2}\beta_K \;\leq\; \frac{4\phi_{U,U}\gamma^2 V_{\max}}{(1-\gamma)^2} = \frac{4\phi_{U,U}\gamma^2 R_{\max}}{(1-\gamma)^3},
\end{equation}
using $V_{\max} = R_{\max}/(1-\gamma)$. For $\gamma$ close to $1$, this term is $O((1-\gamma)^{-3})$. The bound~\eqref{eq:beta_crude} is offered as a sanity check confirming that $\beta_K$ is finite under mild conditions; it is \emph{not} a useful quantitative guarantee. Meaningful results require either Bellman closure ($\beta_K = 0$) or a problem-specific argument that $\beta_K$ is small, e.g., via covering-number comparisons between $\mathcal{G}$ and $\mathcal{H}_{(\mathcal{D},w)}$, or by leveraging structural properties of the PQC feature map.

\subsection{Summary}
Dropping Assumption~\ref{as:func_classes} introduces an irreducible approximation bias $\beta_K$, defined in~\eqref{eq:beta_K}, reflecting the mismatch between the Bellman targets and the hypothesis class $\mathcal{H}_{(\mathcal{D},w)}$. The resulting error decomposition~\eqref{eq:two-class-decomp} is standard in the analysis of fitted value iteration with two function classes \cite{chen2019information,theoretical_deepQ}: the estimation error $\varepsilon_{\mathrm{est}}$ is controlled by the sample complexity analysis of the main text, while $\beta_K$ must be controlled either by structural assumptions on the pair $(\mathcal{G}, \mathcal{H}_{(\mathcal{D},w)})$ or by a problem-specific argument. When $\beta_K = 0$ (i.e., Assumption~\ref{as:func_classes} holds), Lemma~\ref{thm:relaxed} reduces exactly to Theorem~\ref{th:fqi_error_propagation}. When $\beta_K > 0$, it constitutes an irreducible error floor that cannot be removed by increasing the sample size $m$ or the number of iterations $K$ alone. The conservative bound $\beta_K \leq 2\gamma V_{\max}$ (Proposition~\ref{prop:beta_crude}) confirms finiteness under mild conditions but is too coarse to yield useful guarantees near $\gamma = 1$. Tightening this bound via covering-number comparisons between $\mathcal{G}$ and $\mathcal{H}_{(\mathcal{D},w)}$, or by exploiting the structure of the PQC feature map, is a natural direction for future work, and the analysis above identifies precisely where $\beta_K$ enters the proof.

\printbibliography

\end{document}